\documentclass[11pt]{article}

\usepackage{amsmath, amssymb, amsthm, mathtools}
\usepackage{algorithm,algorithmicx,algpseudocode}
\usepackage{thmtools}
\usepackage{enumitem}

\usepackage{tikz}
\usetikzlibrary{cd}
\usepackage[margin=1in]{geometry}
\usepackage[utf8]{inputenc}
\usepackage[hyphens]{url}

\usepackage[hypertexnames=false]{hyperref}
\usepackage{cleveref}

\hypersetup{colorlinks=true,linkcolor=blue,citecolor=blue,urlcolor=blue}
\usepackage{complexity}
\usepackage{ytableau}

\theoremstyle{plain}
\newtheorem{theorem}{Theorem}
\newtheorem{claim}[theorem]{Claim}
\newtheorem{assumption}[theorem]{Assumption}
\newtheorem{lemma}[theorem]{Lemma}
\newtheorem{proposition}[theorem]{Proposition}

\newtheorem{fact}[theorem]{Fact}
\newtheorem{definition}[theorem]{Definition}
\newtheorem{note}[theorem]{Note}

\newcommand{\sgn}{\mathrm{sgn}}
\newcommand{\GL}{\mathrm{GL}}
\renewcommand{\S}{\mathrm{S}}

\providecommand{\ket}[1]{\ensuremath{\left|#1\right\rangle}}
\providecommand{\bra}[1]{\ensuremath{\left\langle#1\right |}}
\providecommand{\proj}[1]{\ensuremath{\ket{#1}\!\bra{#1}}}
\providecommand{\braket}[2]{\ensuremath{\left\langle#1\middle|#2\right\rangle}}

\DeclareMathOperator{\Span}{span}
\newcommand{\lp}{\left}
\newcommand{\rp}{\right}
\newcommand{\snorm}[2]{\lp\|#2\rp\|_{#1}}
\DeclareMathOperator{\Var}{Var}
\DeclareMathOperator{\tr}{Tr}

\usepackage{todonotes}

\newcommand{\partition}{\vdash}
\DeclareMathOperator{\Img}{Img}

\renewcommand{\C}{\mathbbm C}

\newcommand{\abs}[1]{\lp | #1 \rp|}
\usepackage{float} 

\usepackage{bbm}
\newcommand{\eps}{\epsilon}
\newcommand*{\defeq}{\stackrel{\text{def}}{:=}}

\newcommand{\swapc}{\mathrm{Swap}}
\newcommand{\swap}{\mathrm{Swap}_{o}}
\newcommand{\swapt}{\mathrm{Swap}_{o}}
\newcommand{\swapts}{\mathrm{OrderedSwap}_{o}}
\newcommand{\gp}{generic pre-processing}
\newcommand{\povm}{row symmetric measurement}
\newcommand{\Gp}{Generic Pre-Processing}
\newcommand{\Povm}{Row Symmetric Measurement}

\newcommand{\pisym}{\Pi_{\mathrm{sym}}^{(\lambda)}}
\newcommand{\pisyms}[1]{\Pi_{\mathrm{sym}}^{(#1)}}

\newcommand{\Un}[1]{ U^{\otimes #1} }
\newcommand{\Und}[1]{ U^{\dagger\otimes #1} }

\newcommand{\U}{\mathrm U} 
\newcommand{\leftComment}{\State $\triangleright\;$} 
\DeclareMathOperator{\rank}{rank}

\newcommand{\so}{S_1}
\newcommand{\st}{S_2}
\newcommand{\soo}{S_{11}}
\newcommand{\sot}{S_{12}}

\newcommand{\pisymt}{\Pi_{\mathrm{sym}}^{(2)}}
\newcommand{\pisch}{\Pi_{\textrm{Sch}}}
\newcommand{\us}{\chi}

\newcommand{\bigO}{ \mathcal{O} }
\newcommand{\N}{ \mathbb{N} }
\newcommand{\obs}{\mathrm{Obs}}
\newcommand{\sh}{\mathrm{S}}
\newcommand{\nc}{n_{\textrm{c}}}
\newcommand{\np}{n_{\textrm{p}}}

\newcommand{\iid}{{i.i.d.}\ }

\title{Improved classical shadows from local symmetries in the Schur basis}
\author{Daniel Grier\thanks{Department of Computer Science and Engineering and Department of Mathematics, UC San Diego. \texttt{dgrier@ucsd.edu}}  \and Sihan Liu\thanks{Department of Computer Science and Engineering, UCSD, California, CA 92092. \texttt{sil046@ucsd.edu}} \and Gaurav Mahajan\thanks{Institute for Foundations of Data Science, Yale University, Connecticut, CT 06511. \texttt{gaurav.mahajan@yale.edu}}}
\date{}

\begin{document}
\maketitle
\begin{abstract}
    We study the sample complexity of the classical shadows task: what is the fewest number of copies of an unknown state you need to measure to predict expected values with respect to some class of observables? Large joint measurements are likely required in order to minimize sample complexity, but previous joint measurement protocols only work when the unknown state is pure. We present the first joint measurement protocol for classical shadows whose sample complexity scales with the rank of the unknown state. In particular we prove $\mathcal O(\sqrt{rB}/\epsilon^2)$ samples suffice, where $r$ is the rank of the state, $B$ is a bound on the squared Frobenius norm of the observables, and $\epsilon$ is the target accuracy. In the low-rank regime, this is a nearly quadratic advantage over traditional approaches that use single-copy measurements. 
    
    We present several intermediate results that may be of independent interest: a solution to a new formulation of classical shadows that captures functions of non-identical input states; a generalization of a ``nice'' Schur basis used for optimal qubit purification and quantum majority vote; and a measurement strategy that allows us to use local symmetries in the Schur basis to avoid intractable Weingarten calculations in the analysis.
\end{abstract}

\section{Introduction}
One of the most fundamental questions in quantum state learning is as follows: given access to many copies of some unknown state $\us \in \C^{d \times d}$, what is the fewest number of copies you can measure to learn $\us$ or some property of it? Of course, the number of copies---known as the \emph{sample complexity}---depends on many specifics of the question. That said, in essentially every variant, the optimal strategy is to perform entangling measurements on large groups of copies simultaneously. Evidently, single-copy measurements are less efficient than joint measurements on $\us^{\otimes n}$ at extracting information about the unknown state.

Consider the shadow tomography task (originally introduced in \cite{aaronson2018shadow}), where the goal is to estimate $\tr(O_i \us)$ to additive imprecision $\epsilon$ for all observables $O_1, O_2, \ldots, O_M$ with high probability. 
While $\tilde{\mathcal O}(\epsilon^{-4} \log^2 M \log d)$ copies suffice with entangling measurements \cite{buadescu2021improved}, 
single-copy measurements require $\tilde{\Omega}(\min\{d, M\} \epsilon^{-2})$ samples \cite{chen2022exponential}.

In this paper, we consider the well-known communication variant of the shadow tomography task called \emph{classical shadows} (originally introduced in \cite{hkp_shadows}) for which much less is known about entangling measurements. In this setting, the measurements are not allowed to depend on the set of observables. This constraint models a scenario in which one party has access to the unknown quantum states and wants to send a classical snapshot (i.e., the ``shadow'') of the state for another party to analyze. With single-copy measurements, 
$\mathcal O(B \log M / \epsilon^2)$ samples suffice for the classical shadows task, where $B \ge \max\{\tr(O_1^2), \ldots, \tr(O_M^2)\}$ is a bound on the squared Frobenius norm of the observables  \cite{hkp_shadows}.\footnote{The sample complexity of the classical shadows protocol in \cite{hkp_shadows} depends on the exact measurement scheme, but we choose to focus on those schemes that are designed for global rather that local observables.}

While the sample complexity's linear dependency on $B$ is benign for low-rank observables, i.e., the rank-$1$ projectors $\proj{\psi}$ used for fidelity estimation, the dependency can be rather costly in general.
Indeed, even for observables with bounded spectral norm (i.e., $\snorm{\infty}{O} \le 1$), the squared Frobenius norm $B$ can be as large as $d$, the dimension of the space $\chi$ lies in. Concretely, this linear-in-dimension sample complexity is required to compute the overlap of $\us$ with some high-dimensional subspace or to estimate the expected value of $\us$ with respect to a high-weight Pauli observable.\footnote{Technically, low-weight Pauli observables also have high Frobenius norm, but in such cases, there exist more sample-efficient sampling techniques \cite{hkp_shadows}.}

This raises a natural question: is it possible to improve the dependency on $B$ with joint measurements? A first attempt at addressing this question was given in \cite{grier2022sample}, where multi-copy measurements lead to the improved sample complexity of $\mathcal O(\sqrt{B}/\epsilon + 1/\epsilon^2)$ for \emph{pure states} $\us$. This mirrors the history in multi-copy tomography where the pure state case \cite{holevo2011probabilistic, hayashi1998asymptotic} was solved before the general one \cite{haah2016sample, odonnellwright2016}.

From a technical perspective, this history may be attributed to the fact that
the joint state $\us^{\otimes n}$ lives in the \emph{symmetric subspace} when $\us$ is pure, and consequently enjoys many desirable properties that make analysis easier (see, e.g., \cite{harrow2013church}).
Meanwhile, the design of joint measurement protocols for mixed states generally requires a deeper understanding of the  actions of the unitary and the permutation groups on the entire space of $(\C^d)^{\otimes n}$, 
making a general solution significantly more complicated.
Indeed, the difficulty in analyzing joint measurement schemes, and the subsequent complications arising from the representations of $(\C^d)^{\otimes n}$, are reflected in the fact that for certain quantum information processing tasks, optimal protocols are only known for $d = 2$ \cite{buhrman2022quantum, cirac1999optimal, keyl2001rate}.

The main result of this paper is the first nontrivial joint measurement classical shadow protocol for \emph{mixed states}. Specifically, its sample complexity scales with the rank of the unknown state:
\begin{theorem} 
\label{thm:main_informal}
The multi-copy sample complexity of the classical shadows task\footnote{See \Cref{defn:classical_shadows} for a precise definition of the classical shadows task.} for unknown state $\us \in \C^{d \times d}$ of rank $r$ and observables $O_1, \ldots, O_M$ with $\snorm{\infty}{O_i} \le 1$ and $\tr(O_i^2) \le B$ is $\mathcal O( \sqrt{r B} \eps^{-2} \log M)$.
\end{theorem}
In the regime where both the rank $r$ and desired accuracy $\epsilon$ are constant, this sample complexity matches that of \cite{grier2022sample}, which is provably optimal. For a full comparison with previous works, see \Cref{table:sample_complexity_overview}. The astute reader may notice that the sample complexity of \Cref{thm:main_informal} appears worse than \cite{grier2022sample} when $r=1$ (i.e., $\us$ is pure) or worse than \cite{hkp_shadows} when $r \ge B$. In fact, our measurement procedure is nearly \emph{identical} to the previous works in those two extremes, so we can in fact achieve the same sample complexity as the prior works in these two regimes. 
This reveals an interesting property of our measurement scheme---it is in some sense a smooth interpolation between the two measurement schemes based on the rank of the unknown state.

\begin{table}
\centering
    \begin{tabular}{| r | c | c | c | }  \hline
    & \rule[-2ex]{0pt}{5ex} $\rank(\us) = 1$ & $\rank(\us) = r$ & $\rank(\us) = d$ \\ \hline
   \rule[-2ex]{0pt}{5ex}1-copy measurements  & $\mathcal O(\sqrt{Bd}/\epsilon + 1/\epsilon^2)$ \cite{grier2022sample} & $\mathcal O(B/\epsilon^2)$ \cite{hkp_shadows} & $\Theta(B/\epsilon^2)$ \cite{hkp_shadows} \\ \hline
    \rule[-2ex]{0pt}{5ex} Joint measurements & $\tilde{\Theta}(\sqrt{B}/\epsilon + 1/\epsilon^2)$ \cite{grier2022sample} & $\mathcal O(\sqrt{Br}/\epsilon^2)$ [Thm.~\ref{thm:main_informal}] & $\mathcal O(B/\epsilon^2)$ \cite{hkp_shadows} \\ \hline
    \end{tabular}
    \caption{Sample-complexity of classical shadows for single-copy measurements and joint measurements in various rank regimes. We write $\mathcal O(\cdot)$ to imply there is an algorithm that achieves that sample complexity, and we write $\Theta(\cdot)$ to imply there is both an algorithm and a matching lower bound. For simplicity, we have assumed the number of target observables $M = 1$ since all known algorithms can be amplified using a standard median-of-means trick to scale with $\mathcal O(\log M)$. Any algorithm for higher rank can be used for lower rank when it gives a better guarantee. For example, the general single-copy algorithm of \cite{hkp_shadows} (sample complexity $\mathcal O(B/\epsilon^2)$) can be used instead the pure state algorithm of \cite{grier2022sample} (sample complexity $\mathcal O(\sqrt{Bd}/\epsilon + 1/\epsilon^2)$) in regimes of $B$ and $\epsilon$ where it is superior.}
    \label{table:sample_complexity_overview}
\end{table}

\subsection{Proof overview and technical contributions}
Our solution starts by reducing the classical shadow task for mixed states to another problem we call the \emph{Population Classical Shadow Task}.

\paragraph{Population Classical Shadows}
To motivate the definition of this task, let's start by diagonalizing $\us$, so that $\us = U (\sum_{i=1}^d \alpha_i \proj{i}) U^\dag$. Therefore, we can expand $\chi^{\otimes n}$ as
\[
\us^{\otimes n} = \sum_{e \in [d]^n} \left( \prod_{i =1}^n \alpha_{e_i} \right) U^{\otimes n} \proj{e} (U^{\otimes n} )^\dag.
\]
In other words, we can view $\us^{\otimes n}$ as a statistical mixture of pure states of the form $U^{\otimes n} \ket{e}$ such that every $\ket{e_i}$ is distributed \iid according to the eigenvalues of $\chi$.
Consider the ``average'' state 
\[
P(\ket{e}, U):= \frac{1}{n} \sum_{i=1}^n U \ket{e_i} \bra{e_i} U^{\dagger},
\]
and notice that an unbiased estimator of $P(\ket{e}, U)$ immediately implies an unbiased estimator of $\chi$ by linearity. When clear, we suppress the explicit dependence on $\ket{e}$ and $U$, and simply write $P$ for this state.

Motivated by this, we define a variant of the classical shadow task in which one builds a ``population'' classical description $\hat P$ for the joint state $U^{\otimes n}\ket{e}$ such that $\tr(O \hat P)$ approximates $\tr\lp( O P \rp)$ (see \Cref{def:population-shadow}).
We show that this population classical shadow task can be used as a stepping stone (in a black-box manner) for the original classical shadow task for mixed states with a small $1/\eps^2$ additive overhead in the sample complexity (see \Cref{lem:population-reduction}). 
Moreover, the model is interesting in its own right---it generalizes the notion of classical shadow to settings where the input registers contain states that are not necessarily all identical. 

From now on, we focus on solving the population classical shadow task where the input is a pure state of the form $U^{\otimes n} \ket{e}$.
At the highest level, our classical shadow scheme has the usual basic outline: design a joint measurement on the input state $U^{\otimes n}\ket{e}$; from the outcome, construct an unbiased estimator $\hat P$ such that $\E[\hat P] = P$ (where the expectation is over the randomness in the measurement outcome); finally, for any observable $O$, bound the variance of $\tr(O\hat P)$ to show that it is a good estimate of the true expectation $\tr(O P)$ for sufficiently large $n$.

\paragraph{Measurements in Schur Basis} To construct the measurement, we first need to understand the structure of the space $(\C^d)^{\otimes n}$ in which $U^{\otimes n} \ket{e}$ lives. Our starting point will be the famous Schur-Weyl decomposition:
\[
(\C^d)^{\otimes n} \cong \bigoplus_{ \lambda \partition n } Q_{\lambda} \otimes P_{\lambda}
\]
where subspaces $Q_{\lambda}$ and $P_{\lambda}$ are irreps\footnote{We try to give sufficient background of all representation theory needed for our main result. While our treatment of the topic is far from exhaustive (see, e.g., textbooks \cite{etingof2011introduction, fulton2013representation}), we refer the reader to \Cref{sec:key_representation_theory} for key definitions and theorem statements. Additional proof details can be found in appendices \ref{app:more-rep-theory} and \ref{app:nice-schur-basis}.} of the unitary group $\U(d)$ and symmetric group $\S_n$, respectively. 
In other words, the decomposition gives a \emph{Schur basis} for $(\C^d)^{\otimes n}$: vectors $\ket{(\lambda, i, j)}$ indexed by $\lambda \partition n$, $i \in [\dim Q_\lambda]$, and $j \in [\dim  P_\lambda]$ 
such that $\{\ket{(\lambda, i, j)}\}_i$ spans an irrep of $\U(d)$ isomorphic to $Q_\lambda$ and $\{\ket{(\lambda, i, j)}\}_j$ spans an irrep of $\S_n$ isomorphic to $P_\lambda$.
The main advantage of this basis is that it  ``decouples'' the actions of the unitary and the permutation group on the space. 

To see why this is helpful, we remark that a key difficulty in handling 
$U^{\otimes n} \ket{e}$ is that it contains a lot of ``noisy'' permutation information irrelevant to the shadow task (i.e., permuting the elements of $e$ leads to a population state with the same average state $P$). 
When $\ket{e}$ contains only one type of symbol, this joint state lives in the symmetric subspace, and therefore contains no permutation information. 
In general, $\ket{e}$ could be far from a symmetric state.
Nonetheless, measuring in the Schur basis allows us to ``symmetrize'' the input state, and thereby discard the irrelevant permutation information.
This effectively distills the most important information about the unknown state---the partition $\lambda$, which is related to the frequencies of each symbol $i \in [d]$ in $\ket{e}$,  and the vector in $Q_\lambda$, which is related to the unitary $U$.

The key question is how we formally leverage measurements in the Schur basis 
to create an unbiased estimator. 
Unfortunately, while there are Schur basis measurement protocols used for full state tomography that output a $\hat\us$ that concentrates around $\us$ \cite{haah2016sample, odonnellwright2016}, 
their analysis requires the sample size $n$ to be at least $d$, which is prohibitively expensive for classical shadows. 
Furthermore, Schur bases are not unique and previous constructions tend to be fairly abstract, making rigorous analysis of the measurement's behavior in the low-sample regime difficult, i.e., $n \ll d$. A central challenge in our approach is to construct a \emph{concrete} Schur basis and a corresponding measurement protocol
such that the calculation of the expectation/variance of the induced estimator is tractable.

\paragraph{Nice Schur Basis and Local Symmetries} As Schur bases are \emph{not} unique, the specific choice that we make will affect the usefulness of the measurement result for analysis. 
A particularly useful property that has appeared before in various quantum information processing tasks \cite{buhrman2022quantum,cirac1999optimal}
on \emph{qubits} (i.e., $d=2$) is the following: each Schur basis vector completely lives in some subspace spanned by the standard basis vectors with the same \emph{Hamming weight}.
At a high-level, the property helps ensure that measuring in the Schur basis  will not disturb the information encoded in the relative frequencies of the symbols in $\ket{e}$ (which are related to the eigenvalues of $\us$).
To our knowledge, this ``nice'' Schur basis has never been generalized to arbitrary local dimension $d$. A key technical and conceptual contribution of this work, which may be of independent interest, is to construct one:
%
\begin{theorem}
\label{thm:nice_schur_basis_intro}
There is an algorithm to construct an orthonormal Schur basis $\{\ket{(\lambda, i, j)}\}_{\lambda, i, j}$ for $(\C^d)^{\otimes n}$ such that each $\ket{(\lambda, i, j)} \in \Span \{P_{\pi} \ket{v}  \mid \pi \in \S_n \}$, 
where $\ket{v}$ is some standard basis vector of $(\C^d)^{\otimes n}$ and $P_{\pi}$ permutes the tensor factors of $(\C^d)^{\otimes n}$ according to $\pi$.
\end{theorem}
We briefly discuss the ideas behind the construction.
Our starting point is a procedure (for qubits) used in \cite{buhrman2022quantum,cirac1999optimal} that can extend an arbitrary basis for an irrep $Q_{\lambda}$ (viewed as a subspace of $(\C^2)^{\otimes n}$) to a full Schur basis.
Using a representation-theoretic argument, we show the same procedure works in $(\C^d)^{\otimes n}$ as well (see \Cref{lem:interpolate-basis}). 
The more involved part is to ensure that the resulting Schur basis respects the additional weight structure and is \emph{orthonormal}.
Towards this goal, we use the fact that the Young Symmetrizer $Y_{\lambda}$ (see \Cref{def:young_symmetrizer}) is (proportional to) a projector onto some irrep $Q_{\lambda}$. 
This allows us to construct our initial $Q_{\lambda}$ basis by running the Gram-Schmidt process on the projected weight subspace $Y_{\lambda} \; \text{span} \{ P_{\pi} \ket{e}  \mid \pi \in S_n \}$ for each individual standard basis vector $\ket{e}$.
This effectively ensures orthonormality and the weight structure on the initial $Q_{\lambda}$ basis.
Lastly, we exploit irreducibility of $Q_{\lambda}$ and properties of the Young symmetrizer to show orthonormality of the full induced Schur basis (see \Cref{cor:base-orthogonality}).

Using this nice Schur basis, 
we obtain a general ``pre-processing'' routine (inspired by \cite{buhrman2022quantum}) for $U^{\otimes n} \ket{e}$ which produces a state $U^{\otimes n} \ket{\tau}$ 
where $\ket{\tau} \in \Span \{P_{\pi} \ket{e}  \mid \pi \in S_n \}$, and a partition $\lambda \partition n$ (arising from Schur sampling). 
Moreover, $\ket{\tau}$  has local symmetries reflecting the measured partition $\lambda = (\lambda_1, \lambda_2, \ldots, \lambda_r)$.
Specifically, if we partition the qudits of $\ket{\tau}$ according to $\lambda$ (i.e., into parts of size $\lambda_1, \ldots, \lambda_r$ in order), then the state is invariant under all local permutations within each part. 

We can then exploit these local symmetries by measuring our processed stated $U^{\otimes n} \ket{\tau}$ with a continuous POVM of the form
\[\{
(\proj{\psi_1})^{\otimes \lambda_1} \otimes (\proj{\psi_2})^{\otimes \lambda_2} \otimes \cdots \otimes (\proj{\psi_r})^{\otimes \lambda_r}
\}_{\psi_1, \psi_2, \ldots, \psi_r} \, ,
\]
where each $\psi_1, \psi_2, \ldots, \psi_r$ are drawn independently from the Haar measure. This construction ensures that the POVM elements conveniently integrate to a uniform sum of permutations that also respect the $\lambda$ partition.
As a direct consequence, 
the measurement is guaranteed to succeed due to the invariance of $\ket{\tau}$ under these local permutations.
After the measurement, our estimator $\hat P$ will be a simple linear combination of the $\psi_i$.
Moreover, the integrals involved in the calculation of the expectation/variance of $\hat P$ will simplify into a uniform sum of permutations as well.
This allows us to avoid complicated calculations with the Weingarten calculus that are often required for Haar integrals over general states \cite{roberts2017chaos}.

Putting all these techniques together (and brushing over some details), we arrive at the following theorem:\begin{theorem}[Informal]
\label{thm:variance_informal}
From the output of the POVM on the pre-processed state of $U^{\otimes n} \ket{e}$, there is an unbiased estimator $\hat P$ such that $\E[\hat P] = P$ and
\begin{align*}
    \Var[\tr(O\hat P)] = \bigO \! \lp(1+ \frac{rB}{n^2} \rp)    
\end{align*}
for any observable $O$ with $\snorm{\infty}{O} = 1, \tr(O^2) \le B$ and for any $e \in [d]^n$ with at most $r$ different symbols.
\end{theorem}
In particular, we will have that $\Var[\tr(O\hat P)] \leq O(1)$ whenever $n \gg \sqrt{rB}$.
Unfortunately, the variance bound from \Cref{thm:variance_informal} stays constant even as $n$ goes to infinity.\footnote{We discuss some intuition why this is happening for our estimator in \Cref{sec:var-comp}.}
Hence, we circumvent the issue with a repetition trick: partition the input state $U^{\otimes n}\ket{e}$ into $T = \Theta(1/\eps^2)$ many ``segments'', i.e., 
$U^{\otimes n}\ket{e} = \bigotimes_{t=1}^{T} 
U^{ \otimes (n / T)} |e^{(t)} \rangle $; compute a shadow estimator $\hat{P}^{(t)}$ for each segment separately; finally, take the empirical average.\footnote{If we are in the common setting where the input is given by $\chi^{\otimes n}$, this corresponds to repeating the algorithm for $1/\eps^2$ times and taking the mean estimator in the end.}
It is not hard to see that the resulting estimator is still unbiased and the variance now shrinks by a factor of $T$. 
The result then follows directly from applying Chebyshev's inequality with this variance bound.

\subsection{Discussion and open problems}
Our work leaves open many possible future directions. Perhaps the most straightforward is to find optimal bounds for classical shadows that scale with the rank of the unknown state. For example, for full state quantum tomography, the sample-complexity of the optimal joint measurement procedure scales linearly with the rank $r$---precisely, $\mathcal O(d r / \epsilon^2)$ \cite{odonnellwright2016, haah2016sample}. Interestingly, our \Cref{thm:main_informal} shows that sample complexity of classical shadows scales at most with the square root of $r$. Could this be improved? Technically, the only lower bound we have comes from the pure state setting, where it was shown that $\Omega(\sqrt{B}/\epsilon + 1/\epsilon^2)$ samples are required \cite{grier2022sample}. 

There are two reasons one might suspect that our algorithm is not yet achieving optimal sample complexity. First, when $\chi$ has full rank, our measurement can be reinterpreted as $\mathcal O(d)$ \emph{independent} measurements on $\us^{\otimes n}$ with high probability. In other words, there are no large entangling measurements, which are crucial for achieving an advantage in sample complexity. Indeed, this is why we achieve the same sample complexity as \cite{hkp_shadows} in the high rank regime.

Second, our sample complexity does not scale smoothly with rank. In particular, for rank 1 states, our measurement is identical to that of \cite{grier2022sample}, and so we achieve the optimal sample complexity: $\mathcal O(\sqrt{B}/\epsilon + 1/\epsilon^2)$. However, for rank 2 states, our sample complexity jumps to $\mathcal O(\sqrt{B}/\epsilon^2)$. Technically, this happens because of a cross term in the variance that does not appear in the pure state case. Unfortunately, we show that this jump in scaling is inherent with our current measurement and analysis (see \Cref{app:var-lb}), suggesting that the measurement may need to be changed if one believes that the scaling with rank should be smooth.

Finally, we ask if there is a version of our algorithm which is robust to small perturbations. Rank is an inherently fragile measure, and one might suspect that when $\chi$ has large overlap with a low-rank subspace, then the sample complexity should be small, even if $\chi$ is technically full rank.

\section{Preliminaries}
Given a vector space $V$, we denote by $L(V)$ the set of all linear operators on $V$, and $\GL(V)$ the set of all invertible linear operators on $V$ (also known as the \emph{general linear group} on $V$).
Given a linear operator $H \in L( \C^d )$, we write $\snorm{p}{H}$ to mean the Schatten $p$-norm. Specifically, $\snorm{1}{H} = \tr(|H|)$ is the trace norm, 
$\snorm{2}{H} = \sqrt{\tr(H H^{\dagger})}$ is the Frobenius norm, and 
$\snorm{\infty}{H}$ is the spectral norm. 
For any $B \in (0, d]$, we define $\obs(B)$ as the set of Hermitian observables $O \in \C^{d \times d}$ with $\snorm{\infty}{O} = 1$ and $\tr(O^2) \leq B$



The task of classical shadow is formulated as follows:
\begin{definition}[Classical Shadows Task]
\label{defn:classical_shadows}
We say an algorithm $\mathcal A$ is a valid protocol for classical shadows task with sample complexity $\nc(\eps, r, B)$ if for any unknown state $\us$ with rank $\leq r$ and for all $\eps, B>0$, algorithm $\mathcal A$ takes $\nc(\eps, r, B)$ many copies of state $\us$ and outputs a classical representation $\hat \us$ satisfying the following: for any observable $O\in \obs(B)$, with probability $>2/3$, 
\[
\left| \tr(O \hat \us) - \tr(O\us)\right| \leq \eps.
\] where the probability is over the randomness of algorithm $\mathcal A$.
\end{definition}

For simplicity, we have given a definition of the classical shadow task which is more stringent than the classical shadows task introduced by \cite{hkp_shadows} or \cite{grier2022sample}. Technically, the shadow can have any form which allows us to approximate $\tr(O\chi)$. In particular, the standard median-of-means procedure used to boost probability of success does not fall into this framework. That said, any algorithm satisfying our procedure can be boosted in the same way: repeat the shadow procedure $k$ times to produce a list of shadows $\hat \us_1, \ldots, \hat \us_k$; then, given an observable $O$, output the median of $\{\tr(O\hat \us_i)\}_{i  \in k}$. By standard analysis \cite{hkp_shadows}, the probability of failure decreases exponentially in~$k$. 

A mixed state $\us^{\otimes n}$ can be written as a statistical mixture of pure states $\Un{n} \ket{e}$ where $\ket{e}$ is some standard basis vector in $(\C^d)^{\otimes n}$ and $U$ is some unitary matrix. Motivated by this, we define a variant of classical shadow task, where the input is given by a list of non-identical but independent pure states.
\begin{definition}[Population Classical Shadow Task]
\label{def:population-shadow}
We say an algorithm $\mathcal A$ is a valid protocol for population classical shadows task with sample complexity $n = \np(\eps, r, B)$ if the following holds. Let $U$ be any unknown unitary and $\ket{e}$ be any unknown standard basis vector in $(\C^d)^{\otimes n}$. Assume $\ket{e}$ contains at most $r$ symbols. 
Then, the algorithm $\mathcal A$ takes state $U^{\otimes n} \ket{e}$ as input, and outputs a classical representation $\hat P$ satisfying: for any observable $O \in \obs(B)$, with probability $>9/10$, 
\[
\left| \tr(O \hat P) - 
\frac{1}{n} \sum_{i=1}^n
\tr(O  U \ket{e_i} \bra{e_i} U^{\dagger} )\right| \leq \eps \, ,
\] where the probability is over the randomness of algorithm $\mathcal A$.
\end{definition}

\subsection{Representation theory}
\label{sec:key_representation_theory}
In this section, we discuss the representation theory of the permutation group $\S_n$ and the unitary group $\U(d)$ in the space $(\C^d)^{\otimes n}$. 
For a more thorough walk through of the subject, we refer the readers to \cite{fulton2013representation}.
For simplicity, we assume all vector spaces have finite dimensions, and are over the field of complex numbers $\C$. 
\begin{definition}[Group Representation]
    A \emph{representation} of a group $G$ on a vector space $V$ is a group homomorphism $\rho \colon G \to \GL(V)$. The space $V$ is called the \emph{representation space}. But when the homomorphism is clear from context, it is also customary to refer to $V$ simply as the representation.
    A \emph{subrepresentation space} $W \subseteq V$ is a subspace of $V$ for which $\rho(G) W \subseteq W$.
    If a representation space $V$ has no proper, non-zero subrepresentation space, it is called \emph{irreducible} or an \emph{irrep} for short.
\end{definition} 
An important property of an irrep is that the vectors within the subspace are all ``connected'' by the corresponding group operations.
\begin{lemma}[Irrep Connectivity]
\label{lem:irrep-decomposition}
    Suppose the vector space $V$ is irreducible under the group homomorhism $\rho \colon G \to \GL(V)$. 
    Then for any vectors $\ket{v},\ket{w} \in V$ there exists at most $m \leq \dim V$ 
    group elements
    $g_1, \ldots, g_m \in G$ and scalars $\alpha_1, \ldots, \alpha_m \in \C$ such~that 
    $$
    \ket{v} = \sum_{i=1}^m \alpha_i \rho(g_i) \ket{w}.
    $$
\end{lemma}
\begin{proof}
    Let $W \coloneqq \Span\{ \rho(g) \ket{w} \mid g \in G \}$. If $\ket{v} \not\in W$, then $W \subsetneq V$ is a subrepresentation.
\end{proof}

We are interested in representations of the permutation and general linear groups, $\S_n \mapsto \GL((\C^d)^{\otimes n})$ and $\U(d) \mapsto \GL((\C^d)^{\otimes n})$. 
Let us start by describing the basic operators.
To each permutation $\pi \in \S_n$, we will associate a permutation operator over the tensor factors of $(\C^d)^{\otimes n}$. We define $P_\pi \colon (\C^d)^{\otimes n} \to (\C^d)^{\otimes n}$ as the linear operator
$$
P_\pi \coloneqq \sum_{i_1, \ldots, i_n \in [d]} \ket{i_{\pi^{-1}(1)}, i_{\pi^{-1}(2)}, \ldots, i_{\pi^{-1}(n)}}\bra{i_1, i_2, \ldots, i_n}.
$$
For each element $U \in \U(d)$, we associate it with the linear operator $\Un{n}$ in $(\C^d)^{\otimes n}$ that applies $U$ to each tensor factor simultaneously. In summary, we will be using the following representations:
$\pi \mapsto P_{\pi}$ for $\pi \in \S_n$ and
$U \mapsto \Un{n}$ for $U \in \U(d)$.
Such a pair of representations satisfy the following commutative property known as the Schur-Weyl duality:
\begin{lemma}[Schur-Weyl duality]
\label{lem:schur-weyl}
Let $K \in L( (\C^d)^{\otimes n} )$.
Suppose $K$ commutes with all $P_{\pi}$ where $\pi \in \S_n$.
Then $K = \sum c_i U_i^{\otimes n}$ for finitely many $U_i \in \U(d)$ and coefficients $c_i \in \C$.
Suppose $K$ commutes with all $\Un{n}$ where $U \in \U(d)$.
Then $K = \sum_{\pi} c_{\pi} P_{\pi}$ for coefficients $c_{\pi} \in \C$.
\end{lemma}
The space $(\C^d)^{\otimes n}$ is reducible under either of the representations.
In particular, it can be decomposed into irreps that can be classified by the notion of \emph{isomorphism}.
\begin{definition}[$G$-linear maps and subrepresentation space isomorphism]
\label{def:irrep-isomorphism}
Let  $\rho \colon G \mapsto \GL(V)$ be a representation in the space $V$, and $U,W$ be two subrepresentations of $V$.
We say that the linear map $F \colon {U} \to W$ is \emph{$G$-linear} if $ F ( \rho(g) \ket{x} ) = \rho(g) F( \ket{x} )$ for all $g \in G, \ket{x} \in U$. 
If there is an invertible $G$-linear map from $V$ to $W$, we say the two subrepresentations  are isomorphic, and write $V \cong W$.
\end{definition}
We state only the isomorphism between subrepresentation spaces.
We defer the more general definition of isomorphism between two arbitrary representations to appendix
(see \Cref{def:general-isomorphism}).


The classes of irreps of $\S_n$ and $\GL(d)$ induced by isomorphism can be labeled by partitions of $n$. 
A partition $\lambda$ of $n$ is a sorted integer vector $\lambda = (\lambda_1, \lambda_2, \ldots, \lambda_k)$ such that $\sum_{i=1}^k \lambda_i = n$.
We assume the convention that $\lambda_1 \ge \lambda_2 \ge \ldots \ge \lambda_k$. 
We write $\lambda \partition n$ for an arbitrary partition $\lambda$ of $n$, and write $\lambda \partition_d n$ when $\lambda$ is a partition of $n$ into at most $d$ parts.

\begin{lemma}[Section 5.12-5.14, 5.18-5.23 of \cite{etingof2011introduction}, or Section 3 of
\cite{haah2016sample} for a concise summary]
\label{lem:irrep_summary}
There exists a set of irreps $\mathcal Q_{\lambda}$, for every partition $\lambda \partition_d n$, such that any irrep of $\U(d)$ in $(\C^d)^{\otimes n}$
is isomorphic to exactly one $ \mathcal Q_{\lambda}$.
Similarly, there exists a set of irreps $\mathcal P_{\lambda}$, for every partition $\lambda \partition_d n$, such that any irrep of $\S_n$ in $(\C^d)^{\otimes n}$ is isomorphic to exactly one $\mathcal P_{\lambda}$.\footnote{If we do not restrict our attention to the specific representations of $ U \mapsto U^{\otimes n} $ and $\pi \mapsto P_{\pi}$ in $(\C^d)^{\otimes n}$,
there are other types of irreps of $\S_n$ and $\U(d)$ that are in general labeled by $\lambda \partition n$ that could have more than $d$ parts. That being said, in our context, the irreps of $\S_n$ and $\U(d)$ are indeed all labeled by partitions with at most $d$ parts.
}
\end{lemma}
In \Cref{app:more-rep-theory}, we discuss in more detail what these irreps are (see 
\Cref{def:Q-canonical}/\Cref{thm:young-Q-projector} for $\mathcal Q_{\lambda}$ and 
 \Cref{def:P-canonical}/\Cref{thm:canonic-P} for $\mathcal P_{\lambda}$).

As a consequence of \nameref{lem:schur-weyl} and the connection between the irreps and partitions, we have that the space $(\C^d)^{\otimes n}$ decomposes into irreps of $\S_n$ and $\U(d)$ in a highly structured form.
\begin{theorem}[Schur-Weyl decomposition]
\label{thm:schur-duality}
$(\C^d)^{\otimes n}$ can be decomposed into a
set of orthogonal subspaces $\Pi_{\lambda}$ labeled by $\lambda \partition_d n$, typically known as the $\lambda$-isotypic component,  such that the following hold:
\begin{enumerate}
    \item $(\C^d)^{\otimes n} = \bigoplus_{\lambda \partition_d n} \Pi_\lambda$.
    \item $\dim \Pi_{\lambda} = \dim \mathcal Q_{\lambda} \times \dim \mathcal P_{\lambda}$.
    \item \label{schur_basis_condition} There exists vectors $\ket{(\lambda,i ,j)}\in (\C^d)^{\otimes n}$ with $\lambda \partition_d n$, $0 \leq i < \dim \mathcal Q_{\lambda}$, and $0 \leq j < \dim \mathcal P_{\lambda}$ such that
    \begin{center}
    \begin{tabular}{l c l}
    $\forall \lambda$ & & $\Pi_\lambda = \Span \{ \ket{(\lambda, i, j)} \mid \forall i,j  \} $, \\
    $\forall \lambda, j$ & & $\mathcal Q_{\lambda} \cong \Span \{ \ket{(\lambda, i, j)} \mid \forall i  \} $, \\
    $\forall \lambda, i$ & & $\mathcal P_{\lambda} \cong \Span \{ \ket{(\lambda, i, j)} \mid \forall j  \}$.
    \end{tabular}
    \end{center}
\end{enumerate}

We refer to any set of vectors $\{\ket{(\lambda, i, j)}\}_{\lambda, i, j}$  satisfying \Cref{schur_basis_condition} above as a \emph{Schur basis}.
\end{theorem}


%
\subsection{Nice Schur Basis}
\label{sec:nicebasis}

In this section, we present an explicit construction of an \emph{orthonormal} Schur basis enjoying some additional \emph{weight} structure.
To describe this extra structure, we will introduce the notion of the \emph{weight subspace}.
\begin{definition}[Weight vectors and weight subspaces]
Let $\ket{e} = \ket{e_1, e_2, \ldots, e_n} \in (\C^d)^{\otimes n}$ where $e_i \in [d]$ represents the index of some standard basis vector in $\C^d$. We define the \emph{weight vector} $\mathcal W(\ket{e})$ of $\ket{e}$ to be the vector $w \in [n]^d$ such that $w_i \coloneqq \lp|\{j \in [n] \mid e_j = i\}\rp|$. In other words, the weight vector counts the repetitions of each basis element in $\ket{e}$.  We say $w \in [n]^d$ is a weight vector to imply that $w$ has the implicit constraint $\sum_i w_i = n$. 

For a weight vector $w \in [n]^d$\, we define the \emph{weight subspace} $V(w) \coloneqq \Span\{ \ket{e} \mid \mathcal W(\ket{e}) = w \} \subseteq (\C^d)^{\otimes n}$ to be the subspace spanned by the standard basis vectors whose weight vectors are $w$. 
\end{definition}

{We define a orthonormal Schur Basis with the property that each vector completely lies in some weight subspace. This property will be crucial in designing our measurement.}

\begin{definition}[Nice Schur Basis]
\label{defn:nice_schur_basis}
We say the set of vectors $\{\ket{(\lambda, i, j)} \mid 
\lambda \partition_d n,
0 \leq i < \dim \mathcal Q_{\lambda}, 0 \leq j < \dim \mathcal P_{\lambda}\}$ form a \emph{nice Schur Basis} if (i) they form an orthonormal Schur Basis, and (ii) for each vector $\ket{(\lambda, i, j)}$, there exists a weight vector $w \in [n]^d$ such that $\ket{(\lambda, i, j)} \in V(w)$.
\end{definition}

As the main technical contribution of the section, we provide a construction of a nice Schur basis.
\begin{restatable}[Explicit construction of nice Schur basis]{theorem}{niceschurbasistheorem}
\label{prop:nice-schur-basis}
Given an input of local dimension $d \in \mathbb N$ and number of copies $n\in \mathbb N$, \Cref{alg:orthonormal_schur_basis} outputs a nice Schur Basis 
$\{\ket{(\lambda, i, j)} \mid 
\lambda \partition_d n,
0 \leq i < \dim \mathcal Q_{\lambda}, 0 \leq j < \dim \mathcal P_{\lambda}\}$ in~$(\C^d)^{\otimes n}$.
\end{restatable}
Our starting point is a procedure 
introduced in 
\cite{cirac1999optimal}, and subsequently advocated in
\cite{buhrman2022quantum}, that can generate a complete Schur basis given a basis for each canonical irrep $ \mathcal Q_{\lambda}$ when $d=2$.
\begin{algorithm}
\caption{Schur basis completion}
\label{func:schur_basis_completion}

\begin{algorithmic}
\leftComment Input: For each $\lambda \partition_d n$, an orthonormal basis $\{\ket{(\lambda, i, 0)} \mid 0 \leq i < \dim \mathcal Q_{\lambda}\}$ for $\mathcal Q_\lambda$ 
\leftComment Output: Schur basis $\{\ket{(\lambda, i, j)} \mid \lambda \partition_d n, 0 \leq i < \dim \mathcal Q_{\lambda}, 0 \leq j < \dim \mathcal P_{\lambda}\}$ for $(\C^d)^{\otimes n}$
\vspace{.5em}

\Function{SchurBasisCompletion}{$\{\ket{(\lambda, i, 0)}\}_{i,\lambda}$}
\For{$\lambda \partition_d n$}
    \State $\mathcal P \gets \Span\{  P_{\pi} \ket{(\lambda, 0, 0)}  \mid \pi \in \S_n\}$
    \State $\{\ket{(\lambda, 0, j)}\}_{j=0}^{\dim \mathcal P -1} \gets$ Any orthonormal basis of $\mathcal P$
    \vspace{.5em}
    \For{$j$ from $0$ to $\dim \mathcal P - 1$}
        \State Write $\ket{(\lambda, 0, j)}
        = \sum_{\pi \in \S_n} \alpha_{ \pi}^{(\lambda, j)} P_{\pi} 
        \ket{(\lambda, 0, 0)}$ for coefficients $\alpha_{\pi}^{(\lambda, j)} \in \C$
    \EndFor
    \vspace{.5em}
    \For{$i$ from $0$ to $\dim \mathcal Q_\lambda - 1$}
        \For{$j$ from $0$ to $\dim \mathcal P - 1$} 
            \State $\ket{(\lambda, i, j)} \gets \sum_{\pi \in \S_n} \alpha_{\pi}^{(\lambda, j)} P_{\pi} \ket{(\lambda, i, 0)}$
        \EndFor
    \EndFor
    \vspace{.5em}
\EndFor
\State \Return $\{\ket{(\lambda, i, j)} \}_{\lambda, i,j}$.
\EndFunction
\end{algorithmic}
\end{algorithm}
We now describe this procedure. 
Its pseudocode can be found in \Cref{func:schur_basis_completion}.
For each $\lambda \partition_d n$, let $\{\ket{(\lambda, i, 0)}\}_{i=0}^{\dim \mathcal Q_{\lambda} -1}$ be the basis of some canonical irrep $\mathcal Q_{\lambda}$.
Consider the subspace $\mathcal P = \Span\{  P_{\pi} \ket{(\lambda, 0, 0)}  \mid \pi \in \S_n\}$. 
It turns out that $\mathcal P$ is an irrep of $S_n$ isomorphic to $\mathcal P_{\lambda}$
(see \Cref{lemma:grid}).
Thus, we have $\dim \mathcal P = \dim \mathcal P_{\lambda}$. 

Let $\{ \ket{(\lambda, 0, j)} \}_{j = 0}^{\dim \mathcal P - 1}$ be some basis for $\mathcal P$. By the construction of $\mathcal P$, each basis vector can expanded as
\begin{align*}
    \ket{(\lambda, 0, j)} \coloneqq \sum_{\pi \in \S_n} \alpha_{\pi}^{(\lambda, j)} P_{\pi} \ket{(\lambda, 0, 0)}.
\end{align*}
where the $\alpha_{\pi}^{(\lambda,j)} \in \C$ are complex coefficients. We will reuse these coefficients to ``interpolate'' the rest of the basis in the $\lambda$-isotypic subspace $\Pi_{\lambda}$.
In particular, for all $i$ such that $0 \le i < \dim \mathcal Q_\lambda$, we define
\begin{align*}
    \ket{(\lambda, i, j)} \coloneqq \sum_{\pi \in \S_n} \alpha_{\pi}^{(\lambda, j)} P_{\pi} \ket{(\lambda, i, 0)}.
\end{align*}
In \cite{buhrman2022quantum}, it has been shown that the above process yields a Schur basis for the space $(\C^2)^{\otimes n}$.
We show that this can be generalized to the space $(\C^d)^{\otimes n}$ for arbitrary $d$.
We provide the proof in \Cref{app:schur-basis-interpolation}.
\begin{restatable}[Schur Basis Completion]{lemma}{interpolatebasis}
\label{lem:interpolate-basis}
Given as input a basis $\{\ket{(\lambda, i, 0)}\}_{i=0}^{\dim \mathcal Q_{\lambda} -1}$ for some canonical irrep $Q_\lambda$ for each partition $\lambda \partition_d n$, Schur Basis Completion algorithm (\Cref{func:schur_basis_completion}) outputs a set of vectors $
\{\ket{(\lambda, i, j)}\}_{\lambda, i, j}$ which form a Schur basis.
\end{restatable}
A Schur basis constructed through \Cref{func:schur_basis_completion} starting from an arbitrary basis for $\mathcal Q_\lambda$ does not necessarily enjoy orthonormality and the required weight structure of a nice Schur basis.
We now describe a specific orthonormal Schur basis for a choice of canonical irrep $Q_\lambda$, which ensures that the output of \Cref{func:schur_basis_completion} is a \textit{nice} Schur basis.

We first specify the canonical irreps $\mathcal Q_{\lambda}$ for $\U(d)$.
To do so, we need to introduce the concept of the Young diagram/tableau and the Young symmetrizer.
\begin{definition}[Young diagram/tableau]
A \emph{Young diagram} is a visualization of an ordered partition $\lambda = (\lambda_1, \ldots, \lambda_d) \partition n$ with $\lambda_1 \ge \lambda_2 \ge \ldots \ge \lambda_d$ in which there are $\lambda_i$ boxes in the $i$-th row of the diagram. 
A \emph{Young tableau} with partition $\lambda$ is a Young diagram where each box has been filled by some number.
\end{definition}
For an element from $\S_n$, there is a natural corresponding group action\footnote{A group action is a function $\alpha \colon G \times X \to X$ that takes an element from a group $G$ and a set $X$ and produces the ``action'' of the group on the set. Formally, a group action $\alpha$ must map $x \in X$ to itself on the identity group element and be composable in the expected way: $\alpha(g, \alpha(h,x)) = \alpha(gh, x)$ for all $g,h \in G$ and $x \in X$.} on a Young tableau that permutes the boxes of the Young tableau (assuming the boxes are labeled from $1$ to $n$ in some order).
Combining this correspondence with the partition information of a Young tableau gives rise to two natural subgroups of $\S_n$.

\begin{definition}[Row/Column Permutation Groups] \label{def:row-column-permute}
Fix a partition $\lambda \partition n$, and denote the set of Young tableaux of partition $\lambda$ as $\mathcal Y_{\lambda}$.
We define a group action $\mathcal Y_{\lambda} \times \S_n \mapsto \mathcal Y_{\lambda}$
on Young tableaux:  $\pi$ 
sends the number $i$ to the box where the number $\pi_i$ lies in.
Now, fix a canonical standard tableau $T$ of shape $\lambda$ with the numbers $1$ up to $n$ filled from top to bottom and then from left to right in order.\footnote{We remark this choice is arbitrary. Any order of indexing is valid as long as the definition is consistent.}
The group action induces two subgroups of $\S_n$ which we denote as $A_{\lambda}$ and $B_{\lambda}$. Specifically, $A_{\lambda}$ is the set of permutations $\pi$ such that the group action only permutes the numbers within the same rows of the tableau, and $B_{\lambda}$ is the analogous set for columns.\footnote{The more standard way of defining these two subgroups is to parameterize them with a Young tableau of partition $\lambda$ filled by numbers from $1$ to $n$. 
For our purpose, it suffices for us to consider the subgroups defined by the Young tableau filled with the numbers from left to right and then top to down in order. Hence, we assume the subgroups are defined with only the partition $\lambda$.
}
\end{definition}

By taking a (signed) sum of the corresponding linear operators of the elements from $A_{\lambda}$ and $B_{\lambda}$ from \Cref{def:row-column-permute}, we arrive at the definition of the Young symmetrizer. 
\begin{restatable}[Young symmetrizer]{definition}{youngsymmetrizerdefinition}
\label{def:young_symmetrizer}
For partition $\lambda \partition n$, the \emph{Young symmetrizer} $Y_{\lambda} \colon (\C^d)^{\otimes n} \to (\C^d)^{\otimes n}$ is the linear operator given by
$$
Y_{\lambda} = \sum_{a \in A_{\lambda}} P_a \sum_{ b \in B_{\lambda}  } \sgn(  b ) \; P_{b}.
$$
\end{restatable}
The Young symmetrizer $Y_{\lambda}$ is proportional to a projector onto some irrep of $\U(d)$ (see \Cref{thm:young-Q-projector}). Thus, we will simply use the image of $Y_{\lambda}$ as our canonical irrep, i.e., $\mathcal Q_{\lambda}:= Y_{\lambda} (\C^d)^{\otimes n}$.

We need a set of vectors
$\{\ket{(\lambda, i, 0)}\}_{i=0}^{\dim \mathcal Q_{\lambda} -1}$ that form a basis of $\mathcal Q_{\lambda} = Y_{\lambda} (\C^d)^{\otimes n}$.
Instead of choosing these vectors to be an arbitrary basis, we require each basis vector to live in some weight subspace completely.
We achieve this through the  process described in \Cref{alg:orthonormal_schur_basis}.
We remark that the basis  constructed via \Cref{alg:orthonormal_schur_basis} and \Cref{func:schur_basis_completion} can be viewed as a natural generalization of the nice Schur basis from \cite{buhrman2022quantum} with a subtle twist:
\begin{note}[Connection with the Nice Schur Basis from \cite{buhrman2022quantum}]
Recall that the Nice Schur Basis from \cite{buhrman2022quantum} 
for some $\lambda$ that has two parts is defined to be
$$
\ket{01}^{\otimes \lambda_2} \otimes \ket{ s_{\lambda_1 - \lambda_2} ( y - \lambda_2  ) } \, ,
$$
 where $\ket{s_{\ell}(w)}$ is the uniform superposition of bitstrings of length $\ell$ with $w$ number of $1$s.
 Suppose we invert the order of row and column permutations in the definition of $Y_{\lambda}$. 
 In particular, define $\bar Y_{\lambda} := \sum_{\pi \in B_{\lambda}} \sgn(\pi) P_{\pi} \sum_{\sigma \in A_{\lambda}} P_{\sigma}$ \footnote{In fact, the two definitions are morally equivalent as both of them are proportional to projectors onto irreps of $U(d)$.}.
Then we have
$\ket{01}^{\otimes \lambda_2} \otimes \ket{ s_{\lambda_1 - \lambda_2} ( y - \lambda_2  ) } \propto \bar Y_{\lambda} \lp( \ket{0}^{\otimes n-\ell} \ket{1}^{\otimes \ell} \rp)
$ whereas our Schur basis will contain the vector 
$Y_{\lambda} \lp( \ket{0}^{\otimes n-\ell} \ket{1}^{\otimes \ell} \rp)$.
\end{note} 

To prove the correctness of \Cref{alg:orthonormal_schur_basis}, we first show that the output of \Cref{alg:orthonormal_schur_basis} is a Schur basis. By the properties of the Schur Basis Completion procedure (\Cref{func:schur_basis_completion}), it is enough to show that the set of vectors $\{ \ket{(\lambda, i, 0)} \}_i$ created on \crefrange{line:Q_lambda_basis_begin}{line:Q_lambda_basis_end} form a basis for the canonical irrep $\mathcal Q_{\lambda} = Y_{\lambda} (\C^d)^{\otimes n}$.
But this is true because $\Span(\bigcup_{w \in [n]^d} V(w)) = (\C^d)^{\otimes n}$, so $\Span(\bigcup_{w \in [n]^d} Y_\lambda V(w)) = Y_\lambda (\C^d)^{\otimes n}$.

Next, we observe that each basis vector lies entirely within a single weight subspace. To see this, notice that the Young symmetrizer is a linear combination of permutation operators and every weight subspace is invariant under permutation operators.

Finally, we need to show the basis is orthonormal. Since the proof is somewhat involved, we defer the majority of the details to \Cref{app:orthogonality}. Below, we sketch the key ideas in a sequence of lemmas. 
To start, we require a sufficient condition for the null space of $Y_{\lambda}$. 
This is a standard property of the Young symmetrizer (see \cite{haah2016sample} e.g.), and
we provide its proof in \Cref{app:orthogonality} for completeness.
\begin{restatable}[Vanishing Condition for Young Symmetrizer]{lemma}{nullspaceyoungsymmetrizer}
\label{clm:non-vanish}
Let $\lambda \partition n$.
Let $\ket{e}$ be a standard basis vector with weight vector $\mathcal W(\ket{e})$.
Suppose $\mathcal W(\ket{e})$ is not majorized by $\lambda$, i.e.,
there exists some $m$ such that
$\sum_{i=1}^m \mathcal W(\ket{e})_{(i)}^{\downarrow}
> \sum_{i=1}^m \lambda_i$ where $\mathcal W(\ket{e})_{(i)}^{\downarrow}$ denotes the $i$-th largest entry of $\mathcal W(\ket{e})$.
Then, $Y_{\lambda} \ket{e} = 0$.
\end{restatable}
Since \Cref{alg:orthonormal_schur_basis} iterates over weight vectors in reverse lexicographic order, this claim immediately gives a key property about $\ket{(\lambda, 0, 0)}$, which lives in the lexicographically largest weight subspace that has
nontrivial intersection with the image of $Y_\lambda$:


\begin{algorithm}
\caption{Nice Schur basis}\label{alg:orthonormal_schur_basis}
\begin{algorithmic}[1]
\leftComment Input: Local dimension $d \in \mathbb N$ and number of copies $n \in \mathbb N$
\leftComment Output: Nice Schur basis $\{\ket{(\lambda, i, j)} \mid \lambda \partition_d n, 0 \leq i < \dim \mathcal Q_{\lambda}, 0 \leq j < \dim \mathcal P_{\lambda}\}$ for $(\C^d)^{\otimes n}$
\vspace{.5em}
\For{$\lambda \partition_d n$}
\State Initialize $S$ to be an empty list
\For{$w \in [n]^d$ in reverse lexicographic order\footnotemark}
    \If{$
    Y_{\lambda} V(w) \neq 0$}
        \State Append to $S$ any orthonormal basis of $Y_{\lambda} V(w)$
    \EndIf
\EndFor
\For{$i$ from $0$ to $|S|-1$} \label{line:Q_lambda_basis_begin} 
\State $\ket{(\lambda, i, 0)} \gets S_i$
\EndFor \label{line:Q_lambda_basis_end}
\EndFor
\State \Return \hyperref[func:schur_basis_completion]{\textsc{SchurBasisCompletion}}($\{\ket{(\lambda, i, 0)}\}_{i,\lambda}$)
\end{algorithmic}
\end{algorithm}
\footnotetext{For vectors $v, w \in [n]^d$, we say that $v < w$ in lexicographic order if there exists index $i \in [d]$ such that $v_i < w_i$ and for all $j < i$, $v_j = w_j$. The vectors $w^{(1)}, w^{(2)}, \ldots, w^{(k)}$ are in reverse lexicographic order whenever $w^{(1)} \ge w^{(2)} \ge \ldots \ge w^{(k)}$.}

\begin{restatable}{lemma}{onedminweightspace}
\label{lem:1d-min-weight-space}
Let $\lambda = (\lambda_1, \ldots, \lambda_\ell) \partition_d n$.
Define $w \in [n]^d$ such that $w_i \coloneqq \left\{\begin{smallmatrix}  \lambda_i & i \le \ell \\ 0 & \text{otherwise} \end{smallmatrix}\right.$. In other words,
$w$ is the lexicographically largest weight vector whose non-zero elements match those of the partition $\lambda$. We have 
\begin{itemize}[itemsep = 0pt]
    \item $\dim Y_{\lambda} V(w) = 1$, and
    \item $\dim Y_{\lambda} V(v) = 0$ for all weight vectors $v \in [n]^d$ lexicographically larger than $w$.
\end{itemize}
\end{restatable}
As an immediate corollary, we obtain pairwise orthogonality between vectors of the form $\ket{(\lambda, 0, k)}$ and $\ket{(\lambda, i, \ell)}$ for $i > 0$ as they lie in different weight subspaces.
\begin{restatable}[Base Case Orthogonality]{lemma}{basecaseorthogonality}
\label{cor:base-orthogonality}
For all $\lambda \partition_d n$, $k, \ell \ge 0$, and $i > 0$, $\braket{(\lambda, i, \ell)}{(\lambda, 0, k)} = 0$.
\end{restatable}
To show orthogonality among the rest of vectors, 
we will make crucial use of the irreducibility of the subspace $\mathcal Q_{\lambda}$.
At a high-level, irreducibility allows us to 
rewrite $\ket{(\lambda, j, k)}$ as
$\sum_{ U } \beta_U \Un{n} \ket{(\lambda, 0, k)}$ (see \Cref{lem:irrep-decomposition}).
Then, one can reduce orthogonality of other pairs of vectors into orthogonality between the vectors $\ket{(\lambda, 0, k)}, \ket{(\lambda, i, \ell)}$, which follows from their definitions when $i = 0$ and $k \neq \ell$, or
our base case orthogonality shown in \Cref{cor:base-orthogonality} when $j \neq 0$. 
The actual proof is a bit more subtle than that, and we refer the readers to \Cref{app:orthogonality} for the detail.
\begin{restatable}[Orthonormality]{lemma}{completeschurbasisorthogonality}
\label{lem:basis-orthogonality}
For all $\lambda \partition n$ and $i,j,k,\ell \ge 0$, we have $\braket{(\lambda, i, k)}{(\lambda, j, \ell)} = \delta_{i, j} \; \delta_{k, \ell}$.
\end{restatable}

\section{Building classical shadows}
We now present our algorithm to build classical representation $\hat \us$ for a $1$-qudit quantum state $\us$.  Later, when we receive the observable $O$, we just output $\tr(O \hat \us)$. 

\subsection{Reduction to Population Classical Shadows}
Using standard probability theory, we can show that an algorithm for population classical shadow task implies an algorithm for the (original) classical shadow task for mixed states.

\begin{lemma}[Classical Shadows using Population Classical Shadows]
\label{lem:population-reduction}
Suppose algorithm $\mathcal A$ is a valid protocol for population classical shadows task with sample complexity $\np(\eps, r, B)$.
Then it is also a valid for classical shadows task with sample complexity $\nc(\eps, r, B) = \max(\np(\eps/2, r, B), \eps^{-2})$.
\end{lemma}
\begin{proof}
Say we run algorithm $\mathcal A$ on $n$ copies of an unknown state $\us$. We will show that this means, with probability at least $9/10$, we run algorithm $\mathcal A$ on the joint state $U^{\otimes n} \ket{e}$ satisfying that
$$
\abs{\frac{1}{n}
\sum_{i=1}^n \tr( O U \ket{e_i}\bra{e_i} U^{\dagger} )
- \tr( O \chi )
} \leq \eps.
$$
This will be enough for our main claim, because $\mathcal A$ is a valid protocol for population classical shadows task implies conditioned on running algorithm $\mathcal A$ on the state $U^{\otimes n} \ket{e}$, the algorithm outputs $\hat P$ which satisfies: for any observable $O$, with probability $9/10$,
$$\abs{\tr( O \hat P ) - \frac{1}{n}
\sum_{i=1}^n \tr( O U \ket{e_i}\bra{e_i} U^{\dagger} )} \leq \eps$$
Therefore, together this implies: on running algorithm $\mathcal A$ on joint state $\us^{\otimes n}$, it outputs $\hat P$ such that for any observable $O$, with probability $>4/5$, 
$$
\abs{
\tr( O \hat P )
- \tr( O \chi )
} \leq 2\eps \, .
$$

We now prove our claim. Suppose the unknown $1$-qudit rank $r$  quantum state $\us$ can be written as $\sum_{i=1}^r \sigma_i U \ket{i} \bra{i} U^\dagger$. Then, the state $\us^{\otimes n}$ can be written as a statistical mixture of the form:
$$
\us^{\otimes n} =
\Un{n}
\sum_{ w \in [n]^d: |w| = n }
\Pr[X = w] 
\sum_{ \pi \in S_n }
P_{\pi} \bigotimes_{i=1}^d 
\lp( \ket{i}\bra{i} \rp)^{ \otimes w_i }
P_{\pi}^{\dagger}
{U^{\dagger}}^{\otimes n} 
\, ,
$$
where $X$ follows the multinomial distribution
$\text{Mul}(n, \sigma)$.
Define the random variable
\begin{align*}
S:= 
\sum_{i=1}^d X_i \tr \lp( O U \ket{i} \bra{i} U^{\dagger}  \rp) \, ,
\end{align*}
where $X$ follows the multinomial distribution $\text{Mul}(n, \sigma)$.
We note that
$$
\E[S] = 
\sum_{i=1}^d \E[X_i] \tr \lp( O U \ket{i} \bra{i} U^{\dagger}  \rp)
= n \; \sum_{i=1}^d \sigma_i U \ket{i} \bra{i} U^{\dagger}
= n \tr(O \chi)
$$
For convenience, we define
$o_i =: \tr \lp( O U \proj{i} U^{\dagger} \rp)$.
Then
we can bound the variance of $S$ from above by
\begin{align*}
\Var[S] = \E[ S^2 ] - \E^2[S]    
&= 
\sum_{i=1}^d 
o_i^2
\Var[X_i]
+ \sum_{i \neq j}  
o_i \; o_j \;
\lp(
\E[ X_i X_j ] 
- \E[X_i] \E[X_j]
\rp) \\
&\leq 
\sum_{i=1}^d \Var[X_i]
\leq \sum_{i=1}^d n \sigma_i = n \, ,
\end{align*}
where in the first inequality 
we use the fact that $\tr( O U \ket{i} \bra{i} U^{\dagger} ) \leq 1$, and that the covariance between any two coordinates
$X_i, X_j$ of a multinomial distribution is non-positive, 
in the second inequality we use the fact that $\Var[X_i] = n \sigma_i (1 - \sigma_i)$.
Therefore, by Chebyshev's inequality, we have
\begin{align*}
\Pr\lp[ \abs{ S/n - \tr(O \chi) }
\leq \eps
\rp]  \geq 9/10 \, ,
\end{align*}
as long as $n \gg \eps^{-2}$. This concludes our proof.
\end{proof}

\subsection{Algorithm for Population Classical Shadow Task}
In previous subsection, we proved a reduction, \Cref{lem:population-reduction}, which showed that an algorithm for population classical shadows task is also an algorithm for population classical shadows task. In this subsection, we present an algorithm for population classical shadows task. In particular, we show the following:
\begin{proposition}
\label{prop:population-shadow}
\Cref{alg:build} is a valid protocol for population classical shadows task with sample complexity $\np(\eps, r, B) = \bigO( \sqrt{ B r} / \eps^2 )$.
\end{proposition}

\begin{algorithm}[H]
\caption{Algorithm for Population Classical Shadow Task} \label{alg:build}
    \begin{algorithmic}[1]
    \Require Accuracy parameter $\eps$, and input state $U^{\otimes n} \ket{e}$ with $n \gg \eps^{-2}$.
    \State Set $T = 10 \eps^{-2}$.
    \State Initialize $\hat \sh = \{\}$.
    \State Partition the input state $U^{\otimes n} \ket{e}$ into
    $ 
    \bigotimes_{t=1}^T
    U^{\otimes (n/T)} \ket{e^{(t)}} $.
    \For{ $t = 1 \cdots T$ }
    \State Step 1: Perform generic pre-processing $\Gamma$ on the state $U^{\otimes n} \ket{e^{(t)}}$ (see \Cref{alg:pp}).
    \State \qquad \quad~ Obtain the resulting state $U^{\otimes n}\ket{\tau}$, and a partition $\lambda$. 
    \label{step:pp}
    \State Step 2: Run \povm~$M^{(\lambda)}$ on $U^{\otimes n} \ket{\tau}$ (see \Cref{def:measurement}). 
    \State \qquad \quad~Obtain a classical description $\Psi$ from the measurement result (see \Cref{sec:exp-comp}). \label{step:povm} 
    \State \qquad \quad~Add 
    $ \frac{T}{n} \; \lp( \Psi - h(\lambda) \; I \rp) $ into the list $\hat \sh$, where $h(\lambda)$ is the number of parts in the partition~$\lambda$.
    \EndFor
    \State Return $\hat P$, the average of the matrices in $\hat \sh$.
    \end{algorithmic}
\end{algorithm}

We now describe the outline of \Cref{alg:build}. The algorithm consists of two main steps at a high level: \textit{\Gp}~and \textit{\Povm}.
In the first step, \Gp, we perform {projective measurements followed by a
unitary transformation, both defined by the set of nice Schur basis from \Cref{prop:nice-schur-basis}},
that translates an arbitrary state of the form $\Un{n} \ket{e}$, where $\ket{e}$ is some standard basis vector
in $(\C^d)^{\otimes n}$, 
into a superposition of 
the form 
$\Un{n}  \ket{\tau} = \Un{n} \sum_{i} \alpha_i \ket{( \lambda^*, i, 0 )} $, where $\ket{( \lambda^*, i, 0 )}$ is the orthonormal Schur basis constructed in \Cref{prop:nice-schur-basis}. We describe this 
{pre-processing step} in detail in \Cref{sec:schur-transform}.

Notably, due to the weight structure of the Schur basis, 
this transformation also preserves the weight structure --- the state $\ket{\tau}$ will live in the same weight subspace as $\ket{e}$.
Furthermore, the output state is highly symmetrized as it lives in the image of the Young symmetrizer (since each $\ket{(\lambda^*, i, 0)}$ is a basis vector for $\mathcal Q_{\lambda} = Y_{\lambda} (\C^d)^{\otimes n}$
).

Motivated by this, 
as the second step of the algorithm, we define a joint measurement for the \emph{Row Symmetric Subspace} (\Cref{def:rowsym}) in \Cref{sec:row-subspace-measurement}.
Specifically, given a partition $\lambda \partition_d n$, this is the subspace invariant under $P_{\pi}$ for any $\pi \in A_{\lambda}$, where $A_{\lambda}$ is the Row Permutation Group defined in \Cref{def:row-column-permute}.
Since the Young symmetrizer is of the form 
$Y_{\lambda} = \sum_{ a \in A_{\lambda} } P_a \sum_{b \in B_{\lambda}} \sgn(b) P(b) $, it is not hard to check that the image of $Y_{\lambda}$, in which the post-processing state $U^{\otimes n} \ket{\tau}$ lives,
is included in the row symmetric subspace.
We discuss this in more detail in \Cref{lem:live-in-row-sym}.
Moreover, we remark that the measurement can be viewed as a natural generalization of the standard symmetric joint measurement used in the prior work \cite{grier2022sample}.
\footnote{Indeed, we will recover their measurement if we set $\lambda = (n)$.}

We then proceed to analyze the resulting estimator in a way similar to \cite{grier2022sample}.
In \Cref{sec:exp-comp},
we show that the measurement (after some elementary algebraic post-processing) gives us an unbiased estimator of the target state $\us$. 
In \Cref{sec:var-comp}, we show that the final estimator $\tr(O \hat \us)$ has its variance appropriately bounded.
In \Cref{sec:combine}, we combine the analysis to conclude the proof of \Cref{prop:population-shadow}, which implies \Cref{thm:main_informal} by \Cref{lem:population-reduction}.


\subsection{Step 1: Generic Pre-Processing\texorpdfstring{~$\Gamma$}{}}
\label{sec:schur-transform}
In \Cref{sec:nicebasis}, we showed an explicit construction of a nice Schur basis with basis vectors $\ket{(\lambda, i, j)}$. 
In our first \gp~step, we use this nice Schur basis to project the input state onto a symmetrized state 
in a similar fashion to weak Schur Sampling~\cite{harrow2005,weakschur2007,testing2015,buhrman2022quantum}. In particular, we first use this orthonormal Schur basis to define a projective measurement which on measuring $(\lambda, j)$ projects the input state onto some irrep $Q_{\lambda, j} \cong \mathcal Q_{\lambda}$.

\begin{definition}[Schur Projective Measurement] \label{def:ss}We define the Schur Projective Measurement on $n$-qudits as $\pisch = \{\Pi_{\lambda, j}\}_{\lambda, j}$ with elements $\Pi_{\lambda, j}$ denoting the projection operator onto the subspace formed by the nice Schur basis vectors corresponding to indices $(\lambda,\cdot, j)$, for each partition $\lambda \partition_d n$ and $0\leq j < \dim \mathcal P_{\lambda}$.
\end{definition}

The projective measurement is followed by a unitary transformation
$H_{\lambda, j}$ defined as 
the change of basis operation such that each pair of the basis $\lp( \ket{ (\lambda, i, 0)  }, \ket{ (\lambda, i, j)  } \rp)_{i=0}^{\dim \mathcal Q_{\lambda}-1}$  are swapped and all the other basis are left unchanged. 
In other words, this is an invertible linear map between
the two irreps $Q_{\lambda, j}$ and $Q_{\lambda, 0} = \mathcal Q_{\lambda}$.
One may wonder whether the map is $\U(d)$-linear. We show this is indeed the case in \Cref{lem:change-of-basis-commute}.

With this, we present the first step of our algorithm, \gp~$\Gamma$, in \Cref{alg:pp}.

\begin{algorithm}[H]
\caption{Generic Pre-Processing~$\Gamma$} \label{alg:pp}
    \begin{algorithmic}[1]
    \leftComment Input: An $n$-qudit quantum state $U^{\otimes n} \ket{e}$.
\vspace{.5em}
    \State Run Schur Projective Measurement $\pisch$ on input state $U^{\otimes n} \ket{e}$. 
    \State Obtain a measurement $(\lambda,j)$ and resulting state $\Pi_{\lambda,j} U^{\otimes n} \ket{e}$.
    \label{line:schur-projection}
    \State Perform the basis change operation $H_{\lambda,j}$.
    \State Return partition $\lambda$ and resulting state $H_{\lambda,j} \Pi_{\lambda,j} U^{\otimes n} \ket{e}$.
    \end{algorithmic}
\end{algorithm}



We now show generic pre-processing $\Gamma$ which transforms an arbitrary state of the form $ \Un{n} \ket{e} $ for some standard basis vector $\ket{e} \in (\C^d)^{\otimes n}$ into a highly ``symmetrized'' state.
In particular, 
\begin{lemma}
\label{lem:preprocess-weight-space}
Given an input state of the form $ \Un{n} \ket{e} $ where $\ket{e}$ is a standard basis vector in $(\C^d)^{\otimes n}$, 
the \gp~$\Gamma$ (given by \Cref{alg:pp}) outputs a partition 
$ \lambda \partition_d n$,
and a state $\Un{n} \ket{\tau}$ 
that can be written as
$$
\Un{n} \ket{\tau} =
\Un{n}  \sum_{ 0 \leq i < \dim \mathcal Q_{\lambda}} \alpha_i  \ket{(\lambda, i, 0)}\, ,
$$
where $\alpha_i \neq 0$ only if $\ket{(\lambda, i, 0)}$ lives in the weight subspace $V(\mathcal W(\ket{e}))$.
Moreover, if $\ket{e}$ contains at most $r$ symbols, $\lambda$ will have at most $r$ parts with probability $1$.
\end{lemma}
\begin{proof}
    Since $\{\ket{(\lambda, i, j)}\}$ forms an orthonormal basis, we can write
\begin{align}
\label{eq:schur-decomposition}
\ket{e} 
= \sum_{ \lambda \partition_d n}
\sum_{ 0 \leq  j < \dim \mathcal P_{\lambda}}
\alpha_{\lambda, j}
\sum_{ 0 \leq  i < \dim \mathcal Q_{\lambda}}
\eta_i^{(\lambda, j)} 
\ket{(\lambda, i, j)} \, ,
\end{align}
for some $\alpha_{\lambda,j}, \eta_i^{(\lambda, j)} \in \C$.
Recall that we have constructed our Schur basis to respect certain weight structure. In particular, for each $\ket{(\lambda, i, j)}$, there exists a weight vector $w \in [n]^d$ such that $\ket{(\lambda, i, j)}$ lives in the weight subspace $V(w)$.
This gives us an important observation: 
$\eta_{i}^{(\lambda, j)}$ must be $0$ if 
$ \ket{(\lambda, i, j)} \not \in V(\mathcal W(\ket{e}))
$.
Moreover, we note that for any $\lambda$ that has more than $r$ parts, since the weight vector of any $\ket{e}$ that contains at most $r$ different symbols is not majorized by such $\lambda$,
we have $Y_{\lambda} \ket{e} = 0$ by \Cref{clm:non-vanish}.
This suggests that we will never measure a partition $\lambda$ that has more than $r$ parts.

Since $\ket{(\lambda, i, j)}$ lives in the irrep $Q_{\lambda, j}$ of $\U(d)$
spanned by $ \{ \ket{(\lambda, i', j)} \mid 0 \leq i' < \dim \mathcal Q_{\lambda}\} $
, for any unitary matrix $U$,
we can always write
\begin{align}
\label{eq:subspace-unitary-invariance}
\sum_{ 0 \leq  i < \dim \mathcal Q_{\lambda}}
\eta_i^{(\lambda, j)}
\Un{n} \ket{(\lambda, i, j)} 
= 
\sum_{ 0 \leq  i < \dim \mathcal Q_{\lambda}}
\eta_i^{(\lambda, j, U)} \ket{(\lambda, i, j)} \, ,
\end{align}
for some $\eta_i^{(\lambda, j, U)} \in \C$.
Substituting \Cref{eq:subspace-unitary-invariance} into \Cref{eq:schur-decomposition} then gives
\begin{align}
\Un{n} \ket{e} 
= \sum_{ \lambda \partition n}
\sum_{ 0 \leq  j < \dim \mathcal P_{\lambda}}
\alpha_{\lambda, j}
\sum_{ 0 \leq  i < \dim \mathcal Q_{\lambda}}
\eta_i^{(\lambda, j, U)}
\ket{(\lambda, i, j)}.
\end{align}
Now, suppose we do Schur Projective measurement, and obtain measurement outcome $(\lambda, j)$.
This brings us to the state
$$
\Pi_{\lambda, j}\lp( 
\Un{n} \ket{e} 
\rp)
= 
C_{\lambda, j}
\sum_{ 0 \leq  i < \dim \mathcal Q_{\lambda}}
\eta_i^{(\lambda, j, U)} 
\ket{(\lambda, i, j)}.
$$
where $\Pi_{\lambda, j}$ is the projection operator onto the subspace formed by the nice Schur basis vectors corresponding to indices $(\lambda,\cdot, j)$ and 
$C_{\lambda, j} \in \C$ is some normalization constant.
Substituting back our definition of $\eta_i^{(\lambda, j, U)}$ (\Cref{eq:subspace-unitary-invariance}), 
we can alternatively write this as
$$
\Pi_{\lambda, j}\lp( 
\Un{n} \ket{e} 
\rp)
= 
C_{\lambda, j}
\sum_{ 0 \leq  i < \dim \mathcal Q_{\lambda}}
\eta_i^{(\lambda, j)}
\Un{n}
\ket{(\lambda, i, j)}.
$$
Now, our state lives entirely in the space of $Q_{\lambda, j}$ and we perform the change of basis $H_{\lambda, j}$.

Since it is a change of orthonormal basis, it is a unitary operation that we can implement on a quantum device. Moreover, we show in \Cref{lem:change-of-basis-commute} that the map $H_{\lambda, j}$ is a $\U(d)$-linear map, i.e., it commutes with $\Un{n}$ for all $U \in \U(d)$. Applying the map on the resulting state gives
\begin{align*}
H_{\lambda, j}
\Pi_{\lambda, j}\lp( 
\Un{n} \ket{e} 
\rp)
= 
C_{\lambda, j}
\sum_{ 0 \leq  i < \dim \mathcal Q_{\lambda}}
\eta_i^{(\lambda, j)}
\Un{n}
\ket{(\lambda, i, 0)} \, ,
\end{align*}
where $\eta_i^{(\lambda, j)}$ is non zero only if $ \ket{(\lambda, i, 0)}
$ and $\ket{e}$ live in some common weight subspace.
\end{proof}

\begin{claim}[$\U(d)$-linearity of the change of basis operation $H_{\lambda, j}$]
\label{lem:change-of-basis-commute}
Fix $\lambda \partition_d n$ with at most $d$ parts and $0 \leq j < \dim \mathcal P_{\lambda}$.
Define $H_{\lambda,j}$ as the change of basis operation that swaps the pairs $( \ket{(\lambda, i, 0)}, \ket{(\lambda, i, j)} )$ and leaves the other basis unchanged. 
Then $ H_{\lambda, j}$ commutes with $\Un{n}$ for any $U \in \U(d)$.
\end{claim}
\begin{proof}
It suffices to show that for any
$\ket{v} \in \{ \ket{(\lambda, i, 0)}, \ket{(\lambda, i, j)} \mid 0 \leq i < \dim \mathcal Q_{\lambda}\}$, 
we have that
$$
\Un{n} H_{\lambda, j} \ket{v}
= H_{\lambda, j} \Un{n} \ket{v}.
$$
Since the argument for $\ket{(\lambda, i, j)}$ is symmetric to that for $\ket{(\lambda, i, 0)}$, we focus on the case $\ket{v} = \ket{(\lambda, i, 0)}$.
Assume that 
\begin{align}
\label{eq:unitary-expansion}
\Un{n} \ket{(\lambda, i, 0)}
= 
\sum_{0 \leq k < \dim \mathcal Q_{\lambda}}
\gamma_{\lambda} (k, U, i)  \ket{(\lambda, k, 0)}.    
\end{align}
Then we have
\begin{align*}
\Un{n} H_{\lambda, j} \ket{(\lambda, i, 0)}
&= \Un{n} \ket{(\lambda, i, j)} 
& \text{(Definition of $H_{\lambda, j}$)}
\\
&= \Un{n} \sum_{ \pi \in S_n }
\alpha_{\pi}^{ (\lambda, j) } P_{\pi} \ket{(\lambda, i, 0)} 
& \text{(Definition of $\ket{(\lambda, i, j)}$)}
\\
&=  \sum_{ \pi \in S_n }
\alpha_{\pi}^{ (\lambda, j) } P_{\pi} 
\Un{n}
\ket{(\lambda, i, 0)} 
&\text{($\Un{n}$ commutes with $P_{\pi}$)}
\\
&= \sum_{ \pi \in S_n }
\alpha_{\pi}^{ (\lambda, j) } P_{\pi} 
\sum_{0 \leq k < \dim \mathcal Q_{\lambda}}
\gamma_{\lambda} (k, U, i)  \ket{(\lambda, k, 0)} 
&\text{(\Cref{eq:unitary-expansion})}
\\
&= 
\sum_{0 \leq k < \dim \mathcal Q_{\lambda}}
\gamma_{\lambda} (k, U, i)
\sum_{ \pi \in S_n }
\alpha_{\pi}^{ (\lambda, j) } P_{\pi} \ket{(\lambda, k, 0)} 
\\
&= \sum_{0 \leq k < \dim \mathcal Q_{\lambda}}
\gamma_{\lambda} (k, U, i) 
\ket{(\lambda, k, j)}
&\text{(Definition of $\ket{(\lambda, k, j)}$)}.
\end{align*}
On the other side, we have
$$
H_{\lambda, j} \Un{n}  \ket{(\lambda, i, 0)}
= 
\sum_{0 \leq k < \dim \mathcal Q_{\lambda}}
\gamma_{\lambda} (k, U, i) 
H_{\lambda, j}  \ket{(\lambda, k, 0)}
= 
\sum_{0 \leq k < \dim \mathcal Q_{\lambda}}
\gamma_{\lambda} (k, U, i) \ket{(\lambda, k, j)} \, ,
$$
where the first equality follows from \Cref{eq:unitary-expansion}, and the second follows from the definition of $H_{\lambda, j}$.
\end{proof}

\subsection{Step 2: Row Symmetric Measurement for partition\texorpdfstring{~$\lambda$}{}}
\label{sec:row-subspace-measurement}
In this subsection, we describe our joint measurement. Throughout {the rest of the paper}, we will assume a correspondence between the qudits of a quantum state in $(\C^d)^{\otimes n}$ and the boxes of a Young diagram of partition $\lambda \partition_d  n$. 
Here we assume $\lambda$ is the output of the generic pre-processing $\Gamma$ in the algorithm. We will associate each qudit of this state with a box in the Young diagram from left to right and then from top to bottom. Assume $\lambda$ has $k$ non-zero entries and $\lambda_i$ is the length of the $i$-th row in the corresponding Young diagram. 
{For a standard basis vector $\ket{a} \in (\C^d)^{\otimes n}$, we denote by $\ket{a_{j,p}}$ 
the ``$p$-th qudit in the $j$-th row'' under this Young diagram based indexing, i.e., 
the $(\lambda_1 + \lambda_2 + \ldots + \lambda_{j-1} + p)$-th qudit from $\ket{a}$ in the standard indexing.}
The number in the cell in figure below represents the qudit index for $n=8$ and $\lambda = {4,3,1}$.

\begin{center}
    \begin{ytableau}
    1 & 
    2 & 
    3 & 
    4 \\
    5 & 6 & 7 \\
    8
\end{ytableau}
\end{center}

Recall the row-stabilizing permutation group $A_\lambda$, a sub-group of $\S_n$ induced by the partition $\lambda$ (see \Cref{def:row-column-permute} ). We now define \textit{row symmetric subspace}, a subspace invariant under the permutations from $A_\lambda$.

\begin{definition}[Row symmetric subspace]\label{def:rowsym}
The $\lambda$-row symmetric subspace of an $n$-qudit system $(\C^d)^{\otimes n}$ is a subspace invariant under $P_\pi$ for all $\pi \in A_\lambda$. We define $\pisym$ to be the projector onto it\[
\pisym \defeq \frac{1}{|A_{\lambda}|} \sum_{\pi \in A_{\lambda}} P_\pi\, .
\]
\end{definition}
We note that the subspace $\pisym$ is simply the tensor product of many ``local'' symmetric subspaces, each corresponding to a row in the Young tableau of shape $\lambda$.

\begin{fact}[Dimension of row symmetric subspace]
Let $\pisyms{s}$ be the projector on the symmetric subspace of an $s$-qudit system, i.e., 
$ \pisyms{s} = \frac{1}{s!}\sum_{\pi \in \S_s} P_{\pi}$.
Then  $\pisyms{s}$ and $\pisym$, the projector onto $\lambda$-row symmetric subspace for partition $\lambda \partition_d n$ with $k \le d$ parts, are related by
\begin{equation}\label{fact:row-symmetric-decomposition}
\pisym = \bigotimes_{i=1}^k \pisyms{\lambda_i}\,. 
\end{equation}
The dimension formula of the row symmetric subspace follows from the dimension formula of the symmetric subspace. The dimension of $\lambda$-row symmetric subspace $\pisym$, denoted by $\kappa_\lambda$, and dimension of symmetric subspace  $\pisyms{s}$, denoted by $\kappa_s$ are given by
\begin{equation}
    \label{fact:row-symmetric-dimension}
    \kappa_s = \binom{s + d - 1}{d-1} \quad \text{and}\quad \kappa_\lambda = \prod_{i=1}^k \kappa_{\lambda_i}\, .
\end{equation}
\end{fact}

Our measurement is based on an important measure over the unitary group --- the Haar measure. The Haar measure on unitary group $\U(d)$ is the unique probability measure $\nu_H$ that is both left and right invariant over the group $\U(d)$. It will be more natural for us to consider the measure induced by Haar measure over density matrices of pure states. This is a measure over $\C^{d\times d}$ which we call Haar random state measure, and denote by $\mu_H$.
We use $\mu_H^k$ to denote the corresponding product measure on $(\C^{d \times d})^{\otimes k}$. Our measurement depends on the $\lambda$ partition outputted by the \gp~$\Gamma$ (see \Cref{lem:preprocess-weight-space}).

\begin{definition}[Row symmetric joint measurement] \label{def:measurement}We define the row symmetric joint measurement on $n$-qudits with partition $\lambda \partition_d  [n]$ {that has exactly $k$ parts}
as the following POVM: $$\mathcal{M}^{(\lambda)} = \left\{ F^{(\lambda)}(\psi_1, \ldots, \psi_k)\right\}_{(\psi_1, \ldots, \psi_k)} \cup \{I - \pisym\}$$ with elements $$F^{(\lambda)}(\psi_1, \ldots, \psi_k) \defeq \bigotimes_{i=1}^k \kappa_{\lambda_i}
    \psi_i^{\otimes \lambda_i}\, ,$$ for all density matrices $(\psi_1, \ldots, \psi_k)$ proportional to the product Haar measure $\mu_H^k$, plus a ``fail'' outcome $I - \pisym$ for states in orthogonal complement of the row symmetric subspace. Here $\kappa_{\lambda_i}$ is the dimension of symmetric subspace of an $\lambda_i$-qudit system (see \Cref{fact:row-symmetric-dimension} for the definition of $\kappa_{\lambda_i}$).
\end{definition}
The validity of our measurement follows from the validity of the standard joint measurement over the symmetric subspace.
\begin{lemma}[Orthogonal projector on row symmetric subspace] 
\label{lemma:povm} 
For any partition $\lambda = \{\lambda_1, \ldots, \lambda_k\}$, the expectation of $F^{(\lambda)}(\psi_1, \ldots, \psi_k)$ under the Haar measure is a projection onto the row symmetric subspace: 
\[
\E_{(\psi_1, \psi_2 \ldots , \psi_k) \sim \mu_H^k} \left[F^{(\lambda)}(\psi_1, \ldots, \psi_k)\right] = \pisym.
    \]
\end{lemma}
\begin{proof}
The lemma follows from 
\Cref{fact:row-symmetric-decomposition} and
the following standard equality (see \cite{scott2006tight}) concerning integral over the Haar measure:
$
\kappa_s \E_{ \psi \sim \mu_H }\lp[ 
\psi^{\otimes s}
\rp] = \Pi_{\text{sym}}^{(s)}.
$
\end{proof}
The row symmetric joint measurement is appealing to us for the following reason.
The \gp~$\Gamma$ produces a state that lives in the canonical irrep $\mathcal Q_\lambda := Y_{\lambda} (\C^d)^{\otimes n}${, which is included in the $\lambda$-row symmetric subspace}.
Thus, the $\lambda$-row symmetric joint measurement 
 always outputs meaningful outcome of the form $F^{(\lambda)}(\psi_1, \ldots, \psi_k)$ for some pure states $\psi_1, \ldots, \psi_k$ instead of the ``failure'' outcome $I - \pisym$.
This is {much more efficient than the} naive strategy of directly performing the standard joint measurement over the symmetric subspace on the initial input $\us^{\otimes n}$, which produces ``meaningful'' outcome with negligible probability as $\us^{\otimes n}$ may have trivial components within the symmetric subspace when $\us$ has rank $>1$.
\begin{lemma}
\label{lem:live-in-row-sym}
Let $U^{\otimes n}\ket{\tau}$ be a state obtained via the pre-processing step from \Cref{lem:preprocess-weight-space}. 
Then we have
$$
\pisym  U^{\otimes n}\ket{\tau}  = U^{\otimes n} \ket{\tau}.
$$
\end{lemma}
\begin{proof}
Recall that we define $\mathcal Q_{\lambda}$ to be the image of $Y_{\lambda}$ in $(\C^d)^{\otimes n}$.
Thus, we have
${\ket{\tau}} = Y_{\lambda} \ket{u}$ for some state $\ket{u} \in  (\C^d)^{\otimes n}$.
From the definition of row symmetric subspace above, \Cref{def:rowsym}, we have
$$
\pisym= 
\frac{1}{ \abs{A_{\lambda}} }
\sum_{ \pi \in A_{\lambda} } P_\pi.
$$
Recall the definition of Young symmetrizer $Y_\lambda$, \Cref{def:young_symmetrizer}, implies for any permutation $\pi \in A_\lambda$,  $$
P_\pi Y_{\lambda} = P_\pi \sum_{a \in A_{\lambda}} P_a \sum_{ b \in B_{\lambda}  } \sgn(  b ) \; P_{b} = Y_\lambda .$$ 
Since $\pisym$ commutes with $U^{\otimes n}$,  the lemma hence follows.
\end{proof}

\subsection{Expected value of estimator\texorpdfstring{~$\Psi$}{}}
\label{sec:exp-comp}
Suppose the \gp~$\Gamma$ produces a partition $\lambda \partition_d n$ {that has exactly $k$ parts} and a state $\U^{\otimes n} \ket{\tau}$ (see \Cref{lem:preprocess-weight-space}).
We then proceed to run the POVM $\mathcal M^{(\lambda)}$ (see \Cref{def:measurement}), which gives a measurement outcome $F^{(\lambda)}(\psi_1, \ldots, \psi_k)$. We then combine $\psi_1, \ldots, \psi_k$ into a single $d \times d$ shadow matrix description {$\Psi$}.
In particular, we consider the following linear combination parametrized by the partition $\lambda$:
\begin{align}
\label{eq:shadow-def}
\Psi = 
\sum_{i=1}^k
\lp( d + \lambda_i\rp) \psi_i.    
\end{align}
In this subsection, we compute the expected value of the shadow matrix $\Psi$ conditioned on the post-processing state $U^{\otimes n} \ket{\tau}$ and partition $\lambda$.
{From there, it is not hard to see that it can be processed into an unbiased estimator of the initial mixed state $\us$ with some elementary algebraic operations.}

\begin{lemma}
\label{lem:expectation}
Suppose the output of our pre-processing step $\Gamma$ is a state $\Un{n} \ket{\tau}$ and some partition $\lambda$ with exactly $k$ parts
(see \Cref{lem:preprocess-weight-space} for description of the pre-processing step). Moreover, assume that $\ket{\tau} \in V(w)$ for some weight vector $w \in [n]^d$.
Let $F^{(\lambda)}(\psi_1, \cdots, \psi_k)$ be the measurement outcome after performing the joint measurement $\mathcal M^{(\lambda)}$, 
and $\Psi$ be defined as in \Cref{eq:shadow-def}.
Then we have
$$
\E[ \Psi ] = 
k I + \sum_{i=1}^d w_i U \ket{i}\bra{i} U^{\dagger} \, ,
$$ 
where the expectation is over the randomness of the measurement $\mathcal M^{(\lambda)}$.
\end{lemma}
\begin{proof}
In this proof, we will be working with an $(n+1)$ qudit quantum state, where we will refer to the first qudit as the \textit{output} qudit, and associate a Young diagram with partition $\lambda \partition_d  [n]$ to the remaining $n$-qudit system. 
Specifically, we use $(j,p)$ to index the qudit corresponding to the $p$-th cell in the $j$-th row of the Young diagram.

The expectation of $\Psi$ is given by
\begin{align}
\E[ \Psi ] &=
\sum_{j=1}^k
(d + \lambda_j)
\E_{\mu_H^k}\left[
\psi_j
\tr \lp( \bigg(\bigotimes_{i=1}^k 
    \kappa_{\lambda_i} \;
    \psi_i^{\otimes \lambda_i}\bigg) \cdot \bigg(
    \Un{n}\ket{\tau} \bra{\tau} \Und{n}\bigg)
     \rp) \right]
    \;     
    \nonumber \\
&=
\sum_{j=1}^k
(d + \lambda_j)
\E_{\mu_H^k}\left[
\tr_{-1} \lp(\bigg( 
\psi_j \otimes \bigotimes_{i=1}^k 
    \kappa_{\lambda_i} \;
    \psi_i^{\otimes \lambda_i}\bigg)\cdot \bigg(I \otimes (\Un{n}\ket{\tau} \bra{\tau}\Und{n})\bigg) \rp)
    \right] \nonumber \\
&=
\sum_{j=1}^k
(d + \lambda_j)
\tr_{-1} \lp( 
\E_{\mu_H^k}\bigg[\psi_j \otimes \bigotimes_{i=1}^k 
    \kappa_{\lambda_i} \;
    \psi_i^{\otimes \lambda_i}\bigg]
    \cdot \bigg(I \otimes (\Un{n}\ket{\tau} \bra{\tau}\Und{n})\bigg) \rp).
\label{eq:exp-part-I}
\end{align}
A key step in the computation is to tackle the following expectation expression
$$ \E_{\mu_H^k} \lp[ \psi_j \otimes \bigotimes_{i=1}^k \kappa_{\lambda_i} \psi_i^{\otimes \lambda_i} \rp]\, , $$ where $(\psi_1, \cdots, \psi_k) \sim \mu_H^k$, the $k$-fold Haar measure over a $(n+1)$-qudit system. Similar to the proof of \Cref{lemma:povm}, this can be simplified into the form of $\bigotimes_{i=1}^k L_i$, where $L_i$ is the average of all permutations over the qudits lying in the $i$-th row if $i \neq j$, and {$L_j$ is the average of all permutations over the qudits lying in the $j$-th row and the output qudit}.
One can see that this is almost the definition of $ \frac{1}{ |A_{\lambda}| } 
\sum_{\pi \in A_{\lambda}} P_{\pi}
$ except for the fact that $L_j$ contains permutations involving the output qudit. {To simplify this, note that we can further decompose each permutations in $L_j$ into two parts: a permutation over the qudits in the $j$-th row and an operation that swaps the output qudit with some qudit from the $j$-th row or itself.
We formalize the swap between the output qudit and an arbitrary qudit in $j$th row or itself with the operation below.}
\begin{definition}[Output Swap Permutation, $\swap$]\label{def:swap}
We use $\swap((i, j))$ to denote the permutation which swaps the output qudit with the $(i,j)$ qudit. We overload this definition to define $\swap((i, \lambda_i + 1))$ to be the identity matrix. 
\end{definition}
The observation leads to the following equality:
\begin{align}
\label{eq:exp-integral-simplify}
&\E_{\mu_H^k}\left[\psi_j \otimes \bigotimes_{i=1}^k 
    \kappa_{\lambda_i} \;
    \psi_i^{\otimes \lambda_i}\right] = \frac{\kappa_{\lambda_j}}{\kappa_{\lambda_j+1}} 
    \frac{1}{ \lambda_j + 1 }
    \left(\sum_{p=1}^{\lambda_j + 1} \swap((j, p))\right)  \; \lp( I \otimes \frac{1}{|A_{\lambda}|} \sum_{ \pi \in A_{\lambda}} P_{\pi} \rp).
\end{align}
To simplify this further, note that 
\begin{align}
\label{eq:exp-coefficient-simplify}
\frac{1}{ 1 + \lambda_j }
\; \frac{\kappa_{\lambda_j} }{ \kappa_{\lambda_j+1} } = \frac{1}{ d + \lambda_j }.    
\end{align}
Moreover, an important property (\Cref{lem:live-in-row-sym}) of our post-processed state $\Un{n}\ket{\tau}$ is that: 
\begin{align}
\label{eq:exp-permutation-simplify}    
\frac{1}{ |A_{\lambda}| }
\sum_{\pi \in A_{\lambda}}
P_{\pi} \Un{n}\ket{\tau} = \Un{n}\ket{\tau}\,.
\end{align}
Substitute \Cref{eq:exp-coefficient-simplify} into \Cref{eq:exp-integral-simplify}, then substitute \Cref{eq:exp-integral-simplify} into \Cref{eq:exp-part-I}, and lastly simplify the expression using \Cref{eq:exp-permutation-simplify}.
We thus reach the identity
\begin{align}
\label{eq:permutation-simplify}
\E[ \Psi]
&= 
\sum_{j=1}^k
\sum_{p=1}^{ \lambda_j + 1}
\tr_{-1} \lp( 
\swap((j, p)) 
\Big( 
I \otimes (\Un{n} \ket{\tau} \bra{\tau} U^{\dagger\otimes n})\Big)  \rp).
\end{align}
By the property of our pre-processing step (\Cref{lem:preprocess-weight-space}),
we know $\ket{\tau}$ lives in the same weight subspace as $\ket{e}$. 
Hence, it can be written as the following linear combination:
\begin{align*}
\ket{\tau} = \sum_{ \ket{a} \in E_n(w)}
c_{a} \ket{a} \, ,    
\end{align*}
where we denote by $E_n(w)$ the set of standard basis vectors on $n$ qudits that 
lie in the weight subspace of $V(w)$, i.e., $\ket{a} \in E_n(w)$ have $w_i$ many qudits equal to $\ket{i}$ for $i \in [d]$.
Therefore, we can expand \Cref{eq:permutation-simplify}:
\begin{align*}
\E[ \Psi]
&= 
\sum_{\ket{a}, \ket{a'} \in E_n(w)}
\sum_{j=1}^k
\sum_{p=1}^{ \lambda_j  + 1}
c_a \; c_{a'}^{\dagger} \;
\tr_{-1} \lp( 
\swap((j, p)) \;
\bigg( I \otimes (\Un{n}\ket{a} \bra{a'} \Und{n})\bigg) \rp).
\end{align*}
Fix some $a, a'$.
We claim that the inner expression vanishes if $\ket{a} \neq \ket{a'}$.
Assume $\ket{a} \neq \ket{a'}$. 
Since the weight vectors of $\ket{a}$ and $\ket{a'}$ are the same, $\ket{a}$ must disagree with $\ket{a'}$ in at least two qudits. 
Hence, in the end, there must be at least one qudit where $\ket{a}$ disagrees with $\ket{a'}$ that has been traced out, making the partial trace vanish. 
It then follows that
\begin{align*}
\E[ \Psi]
&= 
\sum_{\ket{a} \in E_n(w)}
|c_a|^2
\sum_{j=1}^k
\sum_{p=1}^{ \lambda_j  + 1}
\tr_{-1} \lp( 
\swap((j, p)) \;
\bigg( I \otimes (\Un{n}\ket{a} \bra{a} \Und{n})\bigg) \rp).
\end{align*}
Now, the partial trace 
$$
\tr_{-1} \lp( 
\swap((j, p)) \;
\bigg( I \otimes (\Un{n}\ket{a} \bra{a} \Und{n})\bigg) \rp)
$$
can be evaluated for a fixed $j$ and $p$.
In particular, it evaluates to $I$ whenever $p=\lambda_j + 1$, and evaluates to $U \ket{a_{j,p}}\bra{a_{j,p}} U^{\dagger}$ otherwise. 
Since $\ket{a} \in E_n(w)$, we get
\begin{align}
\sum_{j=1}^k
\sum_{p=1}^{ \lambda_j  + 1}
\tr_{-1} \lp( 
\swap((j, p)) \;
\bigg( I \otimes \Un{n}\ket{a} \bra{a} \Und{n}\bigg) \rp) \nonumber
=
k \; I 
+ \sum_{i=1}^d w_i \; U \ket{i} \bra{i} U^{\dagger}.
\end{align}
Since the above sum is the same for each $\ket{a} \in E_n(w)$ and $\sum_{\ket{a} \in E_n(w)} |c_a|^2 = 1\, ,$ 
we reach the conclusion
\begin{align*}
\E [ \Psi ]    
= k I + 
\sum_{i=1}^d w_i U \ket{i}\bra{i} U^{\dagger}.
\end{align*}
\end{proof}

\subsection{Variance of \texorpdfstring{$\tr(O\Psi)$}{estimator}}
\label{sec:var-comp}
In this subsection, we show how to bound the variance of $\tr(O\Psi)$ conditioned on the partition $\lambda$. 
\begin{lemma}    
\label{lem:second-moment}
Let $O$ be the observable received and
$\Psi$ be the shadow matrix defined as in \Cref{lem:expectation}. Then we have 
$$
\Var[\tr(O \Psi)] = \E[ \tr(O \Psi)^2 ]
- \E[ \tr(O \Psi) ]^2
\leq k \; \tr(O^2) + \bigO(n) \snorm{\infty}{O^2}
+ \bigO(n^2) \snorm{\infty}{O}^2 \, ,
$$ where the variance is over the randomness of the measurement $\mathcal M^{(\lambda)}$.
\end{lemma}
Before we dive into the proof, we briefly comment on the term $\bigO(n^2) \snorm{\infty}{O}^2$ in our upper bound.
In order to process $\Psi$ into an unbiased estimator of the unknown state $\chi$, we will scale it down by a factor of $n$. 
The ``scaled'' variance then becomes
$$
\frac{k}{n^2} \; \tr(O^2) + \bigO(1/n) \snorm{\infty}{O^2}
+ \bigO(1) \snorm{\infty}{O}^2.
$$
One can see that while the remaining terms approach $0$ as $n$ increases, the last term remains a constant for all $n$.
Thus, in order to make sure that our final estimator has its variance bounded from above by $\eps^2$, we need to repeat the entire process $1/\eps^2$ many times, and take the empirical average. 
This leads to the multiplicative factor $\eps^{-2}$ in our algorithm's final sample complexity.

One may wonder whether the term is simply an artifact of our analysis. To answer this, we show in \Cref{app:var-lb} that there is a natural barrier in improving the variance bound.
In particular, we demonstrate a concrete setup of quantum state $\ket{\tau} = 
\sum_{\pi} \alpha_{\pi} P_{\pi}
\ket{0}^{\otimes p} \ket{1}^{\otimes q}$, partition $\lambda = (n)$ and observable $O$, where we show that the scaled variance 
$\frac{1}{n^2} \Var\lp[ O \Psi \rp]$
is indeed lower bounded by some constant \emph{regardless of} the size of $n$ due to an unexpected ``cross-term''. 

Surprisingly, this indicates that the output of the measurement goes through a ``phase-change'' as the quantum state $\ket{\tau}$ evolves from a highly symmetric pure state  
$\ket{0}^{\otimes n}$ to another non-symmetric pure state $
\sum_{\pi} \alpha_{\pi} P_{\pi}
\ket{0}^{\otimes p} \otimes \ket{1}^{\otimes q}$ --- a state that naturally arises from rank-$2$ classical shadow.
This is because for $\lambda = (n)$,
the measurement $\mathcal M^{(\lambda)}$ is identical to 
the standard symmetric subspace joint measurement
from \cite{grier2022sample}. Thus, for the symmetric state  $\ket{\tau} = \ket{0}^{\otimes n}$, the distribution of the estimator is identical to that of the classical shadow estimator from \cite{grier2022sample} built for the pure state $U \ket{0} \bra{0} U^{\dagger}$, implying that its variance will go to $0$ as $n$ increases. But as we discussed above, the variance for the non-symmetric pure state $
\sum_{\pi} \alpha_{\pi} P_{\pi}
\ket{0}^{\otimes p} \otimes \ket{1}^{\otimes q}$ always stays lower bounded by a constant.

We now present the proof of \Cref{lem:second-moment}.
\begin{proof}
First, we argue that we can assume that $O$ is a traceless observable without loss of generality, i.e., $\tr(O) = 0$. 
In particular, consider
$O_0 = O - \tr(O) I /d$.
$O_0$ is a traceless observable.
Assume that we can show \Cref{lem:second-moment} for $O_0$, then we will have 
$$
\Var[ \tr(O \Psi) ]
= 
\Var[ \tr(O_0 \Psi) ]
\leq k \tr(O_0^2) + \bigO(n) \snorm{\infty}{O_0^2}
+ \bigO(n^2) \snorm{\infty}{O_0}^2.
$$
It then suffices to show that 
$\tr(O_0^2) \leq \tr(O^2)$,  
$\snorm{\infty}{O_0^2} \leq \bigO(1)
\snorm{\infty}{O^2}
$, and
$\snorm{\infty}{O_0} \leq \bigO(1)
\snorm{\infty}{O}
$, which has been shown in Corollary 15 of \cite{grier2022sample}.

In this proof, we will be working with an $(n+2)$-qudit quantum system. 
We will refer to the first and the second qudits as the first and the second output qudits respectively, and again associate a Young diagram with partition $\lambda \partition_d n$ to the remaining qudits.

Suppose the post-processed state is $\Un{n} \ket{\tau}$. Then, the second moment of the estimator $\Psi$ is given by
\begin{align}
&\E \lp[ \Psi \otimes \Psi \rp]     \nonumber \\
&= \hspace{-0.6em} \sum_{j, j' \in [k]} \lp( \lambda_j + d \rp)
\lp( \lambda_{j'} + d \rp)
 \tr_{ -1, -2 }
\lp( \E_{\mu_H^k}\left[\psi_j \otimes \psi_{j'} \otimes
\bigotimes_{i=1}^k \kappa_{\lambda_i} \psi_i^{\otimes \lambda_i}\right] \cdot \bigg(I \otimes I \otimes (\Un{n}\ket{\tau} \bra{\tau} \Und{n})\bigg)
\rp).
\label{eq:var-initial-expansion}
\end{align}
Again, the key step is in simplifying the quantity $\E_{\mu_H^k}\left[\psi_j \otimes \psi_{j'} \otimes
\bigotimes_{i=1}^k \kappa_{\lambda_i} \psi_i^{\otimes \lambda_i}\right]$.
Similar to the observation made in the computation of $\E[ \Psi ]$ in \Cref{lem:expectation}, the quantity can still be simplified into the tensor product of $k$ ``local'' random permutations. 
When $j \neq j'$, the $j$-th and $j'$-th 
random permutations will be over the qudits in the $j$-th row plus the first output qudit, and over the qudits in the $j'$-th row plus the second output qudit respectively.
In this case, we can decompose $\E_{\mu_H^k}\left[\psi_j \otimes \psi_{j'} \otimes
\bigotimes_{i=1}^k \kappa_{\lambda_i} \psi_i^{\otimes \lambda_i}\right]$ into $ \frac{1}{ |A_{\lambda}| } \sum_{ \pi \in A_{\lambda} } P_{\pi}$ and two extra random swap operations involving the output qudits.
The swap operations are formalized below.
\begin{definition}[$2$-Qudit Output Swap Permutation, $\swapt$] \label{def:swap2}
We use $\swapt((i_1, j_1), (i_2, j_2))$ to denote the permutation over $n+2$ qudits which swaps the first output qudit with $(i_1,j_1)$ qudit, and the second output qudit with $(i_2,j_2)$ qudit. We will overload notation to define when either $j_1 = \lambda_{i_1}+1$ or $j_2=\lambda_{i_2}+1$, that the corresponding qudit permutation is the identity.
\end{definition}
When $j = j'$, the $j$-th random permutation will be over the $j$-th row plus both of the output qudits. 
We remark that this random permutation can be decomposed into a random permutation over the qudits in the $j$-th row, a random swap over the two output qudits, and an operation that randomly selects two qudits out of the $\lambda_j + 2$ qudits and swaps them with the two output qudits in order.
The last operation is formalized below.
\begin{definition}[$2$-Qudit Output Ordered Swap, $\swapts$] \label{def:swap3}
We use $\swapts\lp( j, p, p' \rp)$ to denote the permutation over $n+2$ qudits, where $p < p' \in [\lambda_j+2]$, that \textbf{first} swaps the second output qudit with the $(j,p')$-th qudit, and \textbf{then} swaps the first output qudit with the $(j,p)$-th qudit. The main difference from the $\swapt$ permutation from \Cref{def:swap2}, is that now it is possible for the two output qudits to swap with each other. Hence, the order of the two swaps matter.
\end{definition}
The above observation is summarized in the following equality:
\begin{align}
    &\E_{\mu_H^k}\left[\psi_j \otimes \psi_{j'} \otimes
\bigotimes_{i=1}^k \kappa_{\lambda_i} \psi_i^{\otimes \lambda_i}\right] \nonumber \\
&= 
\begin{cases}
&\frac{\kappa_{\lambda_j}}{\kappa_{\lambda_j+1}} \frac{\kappa_{\lambda_{j'}}}{\kappa_{\lambda_{j'}+1}} 
    \frac{1}{ \lambda_j + 1 } \frac{1}{ \lambda_{j'} + 1 }
    \left(\sum_{p=1}^{\lambda_j+1} \sum_{p'=1}^{\lambda_{j'}+1} \swapt((j, p), (j', p'))\right)  \; \lp( I^{\otimes 2} \otimes \frac{1}{|A_{\lambda}|} \sum_{ \pi \in A_{\lambda}} P_{\pi} \rp)  
    \text{ for } j \neq j'  \\
&\frac{ 2 }{ \lambda_j (\lambda_j-1) }
\frac{\kappa_{\lambda_j}}{\kappa_{\lambda_{j}+2}}
\sum_{p' = 1 }^{\lambda_j + 2} \sum_{p = 1}^{p'-1}
\swapts(j, p, p')
\lp( 
\pisymt
\otimes \frac{1}{|A_{\lambda}|} \sum_{\pi \in A_{\lambda}}
P_{\pi} \rp) 
\text{ for } j = j'.
\end{cases}
\label{eq:var-integral-simplify}
\end{align}

We can therefore break the sum 
from \Cref{eq:var-initial-expansion}
into two partial sums $\so$ and $\st$ based on whether $j = j'$, and tackle them separately.
\begin{align*}
\so &\defeq
\sum_{j \neq j'} \lp( \lambda_j + d \rp)
\lp( \lambda_{j'} + d \rp)
 &\tr_{ -1, -2 }
\lp( \E_{\mu_H^k}\left[\psi_j \otimes \psi_{j'} \otimes
\bigotimes_{i=1}^k \kappa_{\lambda_i} \psi_i^{\otimes \lambda_i}\right] \cdot \bigg(I \otimes I \otimes (\Un{n}\ket{\tau} \bra{\tau} \Und{n})\bigg)
\rp) \, , \\
\st &\defeq 
\sum_{j \in [k]} \lp( \lambda_j + d \rp)^2
 &\tr_{ -1, -2 }
\lp( \E_{\mu_H^k}\left[\psi_j \otimes \psi_{j} \otimes
\bigotimes_{i=1}^k \kappa_{\lambda_i} \psi_i^{\otimes \lambda_i}\right] \cdot \bigg(I \otimes I \otimes (\Un{n}\ket{\tau} \bra{\tau} \Und{n})\bigg)
\rp)\, .
\end{align*}
\paragraph{Computing $\so$}
Recall the definition of $2$-qudit swap permutation for our $n+2$ qudit system, $\swapt$ (see \Cref{def:swap2}). 
For this sum, 
{it is convenient to think of the first output qudit as the $(j,\lambda_j + 1)$
qudit, and the second output qudit as the 
$(j',\lambda_{j'} + 1)$ qudit.}
Using \Cref{eq:var-integral-simplify}, we can simplify $S_1$ as: 
\begin{align}
\so&=
\sum_{j \neq j'} \sum_{p =1}^{\lambda_j + 1} \sum_{p' =1}^{\lambda_{j'} + 1}
 \tr_{-1, -2}  
\lp( \swapt((j, p), (j', p'))   \left( I \otimes I \otimes \frac{1}{|A_{\lambda}|} \sum_{\pi \in A_{\lambda}} P_{\pi}\right)  \bigg( I \otimes I \otimes (\Un{n}\ket{\tau}\bra{\tau}\Und{n})\bigg) \rp)  \nonumber \\
&= \sum_{j \neq j'} \sum_{p =1}^{\lambda_j + 1} \sum_{p' =1}^{\lambda_{j'} + 1}
\tr_{-1, -2} 
\lp(  \swapt((j, p), (j', p'))\cdot\bigg( I \otimes I \otimes (\Un{n}\ket{\tau}\bra{\tau}\Und{n})\bigg) \rp) \, ,
\label{eq:fix-j}
\end{align}
where in the first line, we use $$ \frac{ \lambda_j + d }{ \lambda_j + 1}
\frac{ \lambda_{j'} + d }{ \lambda_{j'} + 1}
\frac{\kappa_{\lambda_j}}{\kappa_{\lambda_j+1}}
\frac{\kappa_{\lambda_{j'}}}{\kappa_{\lambda_{j'}+1}} = 1\, ,$$ and in the second line we use \Cref{lem:live-in-row-sym}.


By \Cref{lem:preprocess-weight-space},
we know $\ket{\tau}$ can be written as the following linear combination:
\begin{align*}
\ket{\tau} = \sum_{ \ket{a} \in E_n(w)}
c_{a} \ket{a} \, ,    
\end{align*}
where we denote by $E_n(w)$ the set of standard basis vectors on $n$ qudits that 
lie in the weight subspace of $V(w)$. 
In particular, this gives that 
\[
\ket{\tau} \bra{\tau}
=
\sum_{a, a' \in E_n(w)} c_a c_{a'}^{\dagger} \ket{a} \bra{a'}.
\]
Substituting this back in \Cref{eq:fix-j} leads to many ``cross-terms'' $\ket{a} \bra{a'}$, where $\ket{a} \neq \ket{a'}$. We will see that most of these cross-terms will vanish under the partial-trace. 
We claim the following:
\begin{claim}
\label{clm:partial-trace-vanish}
Let $\ket{a}, \ket{a'} \in E_n(w)$.
For $j \neq j' \in [k]$,
$p \in [\lambda_j+1]$, 
$p' \in [\lambda_{j'}+1]$,  
consider the partial trace
$$
T \defeq
\tr_{-1, -2}  
\lp( \bigg( \swapt((j, p), (j', p'))\bigg) \cdot \bigg( I \otimes I \otimes (\Un{n}\ket{a}\bra{a'}\Und{n})\bigg) \rp).
$$
Then the following holds
\begin{itemize}
\item  If $p = \lambda_j+1$ or $p' = \lambda_{j'} + 1$, $T = 0$ unless $a = a'$.
\item If $p \neq \lambda_j + 1$ and $p' \neq \lambda_{j'} + 1$, $T = 0$ unless $a = a'$ or $\ket{a'} =  \swapc( (j, p), (j', p') ) \ket{a}$.
\end{itemize}
\end{claim}
\begin{proof}
If $p = \lambda_j+1$ or $p' = \lambda_{j'} + 1$, we are swapping at most one qudit from $\ket{a}$ with the output qudits.
If $a \neq a'$, they must disagree with each other in at least two qudits. 
Thus, at least one misaligned qudit will be traced out, leading to a vanishing partial trace.

If $p \neq \lambda_j+1$ and $p' \neq \lambda_{j'} + 1$, we are swapping two qudtis from $\ket{a}$ with the output qudits. Hence, in order for the partial trace to be non-vanishing, we must have $\ket{a} = \ket{a'}$ or that $a, a'$ are mis-aligned precisely at the $(j,p)$-th qudit and the $(j', p')$-th qudit. 
This concludes the proof of \Cref{clm:partial-trace-vanish2}.
\end{proof}

Formally, denote by $S((j, p), (j', p')) \subset E_n(w) \times E_n(w)$, the set of $n$-qudit pairs $(\ket{a}, \ket{a'})$ such that
$\ket{a} \neq \ket{a'}$ and $\ket{a'} = \swapc( (j, p), (j', p') ) \ket{a}$, where $\swapc((j, p), (j', p'))$ denotes the permutation over $n$-qudits that swaps the $(j,p)$-th qudit with $(j',p')$-th qudit \footnote{This should not to be confused with $\text{Swap}_o$ --- the permutation over $(n+2)$-qudits that swaps the output qudits with some other qudits.}.
Using \Cref{clm:partial-trace-vanish}, we see that only cross term pairs in $S((j, p), (j', p'))$ survive the partial trace. Therefore, we can then simplify
\Cref{eq:fix-j} into two partial sums
\begin{align}
\so &= 
\soo
+ \sot \, ,
\end{align}
where we define
\begin{align*}
\soo &\defeq 
\sum_{j \neq j'} \sum_{p =1}^{\lambda_j + 1} \sum_{p' =1}^{\lambda_{j'} + 1}
\quad\sum_{\ket{a} \in E_n(w)} &|c_a|^2
\tr_{-1, -2} 
\lp( \swapt((j, p), (j', p')) \cdot  \bigg( I \otimes I \otimes (\Un{n}\ket{a}\bra{a}\Und{n})\bigg) \rp) \, ,\\
\sot &\defeq 
\sum_{j \neq j'}
\sum_{p =1}^{\lambda_j} \sum_{p' =1}^{\lambda_{j'}} \sum_{\substack{ \left(\ket{a},\ket{a'}\right)\\ \in S((j,p),(j',p')) }} &c_a c_{a'}^{\dagger}
\tr_{-1, -2} 
\lp( \swapt((j, p), (j', p')) \cdot \bigg( I \otimes I \otimes (\Un{n}\ket{a}\bra{a'}\Und{n})\bigg) \rp)\, .
\end{align*}
\paragraph{Computing $\soo$}
We first analyze  $\soo$. Denote the $(j,p)$th qudit of $\ket{a}$ by $\ket{a_{j,p}}$.
We claim the following partial trace for $j\neq j'$ 
\[
\sum_{p =1}^{\lambda_j + 1} \sum_{p' =1}^{\lambda_{j'} + 1} \tr_{-1, -2} 
\lp( 
\swapt((j, p), (j', p'))\cdot \bigg( I \otimes I \otimes (\Un{n}\ket{a}\bra{a}\Und{n})\bigg)
\rp)  
\]
can be simplified into
\begin{align*}
\lp( I +  
\sum_{p=1}^{\lambda_j}
U \ket{a_{j,p}} \bra{a_{j,p}} U^\dagger\rp)
\otimes 
\lp( I +  \sum_{p'=1}^{\lambda_{j'}}
U \ket{a_{j',p'}} \bra{a_{j',p'}} U^\dagger \rp)\, .
\end{align*}
We defer the formal reasoning to \Cref{prop:tensor2}. 
Summing over all terms with $j \neq j'$ then gives us
\begin{align*}
\soo &=
\sum_{\ket{a} \in E_n(w)} |c_{a}|^2
\lp( k I + \sum_{j=1}^k 
\sum_{p=1}^{\lambda_j}
U \ket{a_{j,p}} \bra{a_{j,p}} U^\dagger \rp)^{\otimes 2} -
\sum_{\ket{a} \in E_n(w)} |c_{a}|^2
\sum_{j=1}^k
\lp( I +  
\sum_{p=1}^{\lambda_j}
U \ket{a_{j,p}} \bra{a_{j,p}} U^\dagger \rp)^{\otimes 2}
\, ,\\
&= 
\E[\Psi]^{\otimes 2}
- 
\sum_{\ket{a} \in E_n(w)} |c_{a}|^2
\sum_{j=1}^k
\lp(  I +  
\sum_{p=1}^{\lambda_j}
U \ket{a_{j,p}} \bra{a_{j,p}} U^\dagger \rp)^{\otimes 2}\, .
\end{align*}
where in the last step we use our computation for $\E[\Psi]$ (\Cref{lem:expectation}). 
Since the second term's contribution to $\tr(O^{\otimes 2}  \; \soo)$ is non-negative, we get
\begin{align*}
\tr(O^{\otimes 2}  \; \soo)
\leq \tr(O \E[\Psi])^2.
\end{align*}

\paragraph{Computing $\sot$}
We now look at the sum $\sot$.
For any $j \neq j' \in [k]$, $p \in [\lambda_j]$, $p' \in [\lambda_{j'}]$ and $\ket{a}, \ket{a'} \in S( (j, p), (j', p') )$,
we claim that the partial trace
$$
\tr_{-1, -2}  
\lp( \swapt((j, p), (j', p')) \cdot \bigg(
I \otimes I \otimes (\Un{n}\ket{a}\bra{a'} \Und{n}) \bigg) \rp)
$$
can always be simplified into
$$   
\lp( U \ket{ a_{j, p} }
\bra{ a_{j,p} } U^{\dagger}
\otimes
U \ket{ a_{j', p'} }
\bra{ a_{j',p'} } U^{\dagger} \rp)
\text{Swap}^{(2)} \, ,
$$ 
where $\textrm{Swap}^{(2)}$ is a permutation over $2$-qudits which swaps the first qudit with the second qudit.
We defer the formal reasoning to \Cref{prop:tensor3}.
Thus, we always have the bound
$$
\abs{\tr\lp( (O \otimes O) \; \tr_{-1, -2}  
\lp( \swapt((j, p), (j', p')) \cdot \bigg(
I \otimes I \otimes (\Un{n}\ket{a}\bra{a'} \Und{n}) \bigg) \rp) \rp)}
\leq \snorm{\infty}{O}^2.
$$
Besides, by Cauchy's inequality, 
we have that
$$
\sum_{(\ket{a},\ket{a'}) \in S((j, p), (j, p'))}
\abs{ c_a c_{a'}^{\dagger} }
\leq \sum_{\ket{a} \in E_n(w)} |c_a|^2 = 1.
$$
Combining the above then gives
$$
\tr \lp( (O \otimes O) \; \sot \rp)
\leq \bigO(n^2) \snorm{\infty}{O}^2.
$$
Since $\so = \soo + \sot$, we conclude that
$$
\tr \lp((O \otimes O) \; \so \rp)
\leq \tr( O \E[\Psi] )^2 + \bigO(n^2) \snorm{\infty}{O}^2.
$$

\paragraph{Computing $\st$}
We now look at the partial sum $S_2$
 which encapsulates the terms where $j = j'$.
For computing $\st$, 
{it is convenient to think of the first output qudit as the $(j,\lambda_j + 1)$-th qudit, and the second output qudit as the $(j,\lambda_{j} + 2)$-th qudit.}
Using \Cref{eq:var-integral-simplify}, $\st$ can be simplified as
\begin{align*}
&\sum_{j=1}^k 2C_{\lambda_j, d} \sum_{p' = 2}^{\lambda_j + 2} \sum_{p = 1}^{p'-1}
\tr_{-1, -2} 
\lp(  
\swapts(j, p, p') 
\lp( 
\pisymt
\otimes \frac{1}{|A_{\lambda}|} \sum_{\pi \in A_{\lambda}}
P_{\pi} \rp)  \bigg(I^{\otimes 2} \otimes (\Un{n}\ket{\tau} \bra{\tau}\Und{n})\bigg)\rp)  \\
= &\sum_{j=1}^k 2C_{\lambda_j, d} \sum_{p' = 2 }^{\lambda_j + 2} \sum_{p = 1}^{p'-1}
\tr_{-1, -2} 
\lp(  
\swapts(j, p, p') \cdot
\lp( 
\pisymt
\otimes I^{\otimes n} \rp) \cdot \bigg(I^{\otimes 2} \otimes (\Un{n}\ket{\tau} \bra{\tau}\Und{n})\bigg)\rp)\\
= 
&\sum_{\ket{a}, \ket{a'}\in E_n(w)} c_a c_{a'}^\dagger \\
&\sum_{j=1}^k 2C_{\lambda_j, d} \sum_{p' = 2}^{\lambda_j + 2} \sum_{p = 1}^{p'-1}
\tr_{-1, -2} 
\lp(  
\swapts(j, p, p') \cdot
\lp( 
\pisymt
\otimes I^{\otimes n} \rp) \cdot \bigg(I^{\otimes 2} \otimes (\Un{n}\ket{a} \bra{a'}\Und{n})\bigg)\rp)\, ,
\end{align*}
where in the first line we use $$ C_{\lambda_j, d} \defeq
\frac{ (\lambda_j + d)^2 }{ (\lambda_j+2) (\lambda_j+1) }
\frac{\kappa_{\lambda_j}}{\kappa_{\lambda_{j}+2}} < 1\, ,$$ in the second line we use 
the fact that $\ket{\tau}$ lives in the row symmetric subspace (\Cref{lem:live-in-row-sym}), and in the third line we use the fact that
 $\ket{\tau}$ lives in the weight subspace $V(w)$ (\Cref{lem:preprocess-weight-space} ).

Similar to the computation of $S_1$, we start by characterizing the non-vanishing condition of the partial trace expressions involved.
\begin{claim}
\label{clm:partial-trace-vanish2}
Let $\ket{a}, \ket{a'} \in E_n(w)$.
For $j \in [k]$, 
$p < p' \in [\lambda_j+2]$, consider the partial trace
$$
T \defeq \tr_{-1, -2} 
\lp(  
\swapts(j, p, p') \cdot
\lp( 
\pisymt
\otimes I^{\otimes n} \rp) \cdot \bigg(I \otimes I \otimes (\Un{n}\ket{a} \bra{a'}\Und{n})\bigg)\rp)\, .
$$
Then the following holds
\begin{itemize}
\item  If $p' \geq \lambda_j+1$, $T = 0$ unless $a = a'$.
\item If $p, p' \in [\lambda_j]$, $T = 0$ unless $a = a'$ or $\ket{a'} =  \swapc( (j, p), (j', p') ) \ket{a}$.
\end{itemize}
\end{claim}
\begin{proof}
If $p' \geq \lambda_j+1$, we are swapping at most one qudit from $\ket{a}$ with the output qudits.
If $a \neq a'$, they must disagree with each other in at least two qudits. 
Thus, at least one misaligned qudit will be traced out, leading to vanishing partial trace.

If $p,p' \in [\lambda_j]$, we are swapping two qudtis from $\ket{a}$ with the output qudits. Hence, we must have $\ket{a} = \ket{a'}$ or that $a, a'$ are mis-aligned precisely at the $(j,p)$-th qudit and the $(j, p')$-th qudit.
This concludes the proof of \Cref{clm:partial-trace-vanish2}.
\end{proof}
Due to \Cref{clm:partial-trace-vanish2},
we are going to look at the following cases separately depending on the choices of $p < p' \in [\lambda_j+2]$.
\begin{enumerate}
\item Case I: $p = \lambda_j+1, p' = \lambda_j+2$. We denote the sum of the terms in the case as $S_{21}$.
\item Case II: $p \in [\lambda_j], p' = \lambda_{j}+2$. We denote the sum of the terms in the case as $S_{22}$.
\item Case III: $p \in [\lambda_j], p' = \lambda_{j}+1$. We denote the sum of the terms in the case as $S_{23}$.
\item Case IV: $p < p' \in [\lambda_{j}]$. We denote the sum of the terms in the case as $S_{24}$.
\end{enumerate}
We will examine the variance from each case individually using \Cref{clm:partial-trace-vanish2} that characterizes the non-vanishing conditions of the partial trace involved.

\textbf{Case I}: $p = \lambda_j+1, p' = \lambda_j+2$.
In this case, the partial trace gets simplified significantly
as  $\swapts(j, \lambda_j+1, \lambda_j+2)$ is defined to be the identity operation.
We thus have
\begin{align*}
\tr_{-1, -2} \lp(\lp( \pisymt
\otimes I^{\otimes n} \rp) \cdot \bigg(I \otimes I \otimes (\Un{n}\ket{a} \bra{a'}\Und{n})\bigg)\rp)
=     
\begin{cases}
\pisymt \lp( I \otimes I \rp)  
&\text{ if } a = a' \, , \\
0 &\text{ otherwise}
\end{cases}
\end{align*}
Thus, the total contribution to the variance of the terms in this case will be 
\begin{align*}
\tr( (O \otimes O) S_{21} )
=
\sum_{\ket{a}\in E_n(w)} |c_a|^2 
\sum_{j=1}^k
2C_{\lambda_j, d}
\tr \lp( (O \otimes O) \; \pisymt \lp( I \otimes I \rp) \rp)  
< 2k \; \tr \lp( O^2\rp) \, ,
\end{align*}
when $O$ is a traceless observable.

\textbf{Case II}: $p \in [\lambda_j], p' = \lambda_{j}+2$.
In this case, \Cref{clm:partial-trace-vanish2} says that we must have $a = a'$ or the partial trace vanishes. Moreover, in the $a = a'$ case, the partial trace $$
\tr_{-1, -2} 
\lp(  
\swapts(j, p, p') \cdot
\lp( 
\pisymt
\otimes I^{\otimes n} \rp) \cdot \bigg(I \otimes I \otimes (\Un{n}\ket{a} \bra{a}\Und{n})\bigg)\rp)
$$ simplifies to $$
\lp( (U \ket{a_{j,p}}\bra{a_{j,p}} U^\dagger)  \otimes I\rp)
\pisymt.
$$ 
Again, we defer the formal argument to \Cref{clm:tensor-caseII}.

The total contribution to the variance of the terms in this case is
\begin{align*}
\tr((O \otimes O) S_{22})
= 
\sum_{j=1}^k \sum_{p=1}^{\lambda_j}
\sum_{\ket{a}\in E_n(w)}  2 C_{\lambda_j, d} |c_a|^2
\tr \lp( \lp(O \otimes O\rp)
\;
\lp( U\ket{a_{j,p}}\bra{a_{j,p}}U^\dagger  \otimes I\rp)
\pisymt
\rp)
\end{align*}
Since $U\ket{a_{j,p}}\bra{a_{j,p}}U^\dagger$ is just some pure state and $C_{\lambda_j, d} < 1$, we have that
\begin{align*}
C_{\lambda_j, d} \tr \lp( \lp(O \otimes O\rp)
\;
\lp(U\ket{a_{j',p'}}\bra{a_{j',p'}}U^\dagger  \otimes I\rp)
\pisymt
\rp)
\leq 
\snorm{\infty}{O^2} \, ,
\end{align*}
where in the second inequality we use \Cref{fact:eigen-bound}.
Therefore, the total contribution to the variance of the terms in this case can be bounded from above by
$$  
\tr((O \otimes O) S_{22})
\leq \sum_{ \ket{a} \in E_n(w) } 2|c_a|^2 \sum_{j=1}^k 
\sum_{p=1}^{\lambda_j}
\snorm{\infty}{O^2}
=  2n \; \snorm{\infty}{O^2}.
$$

\textbf{Case III}: $p \in [\lambda_j], p' = [\lambda_j+1]$.
Note that in this case we have
$$
\swapts(j, p, \lambda_j+1) \cdot
\lp( 
\pisymt
\otimes I^{\otimes n} \rp)
= 
\swapts(j, p, \lambda_j+2) \cdot
\lp( 
\pisymt
\otimes I^{\otimes n} \rp)
$$
since the first step of $\swapts(j, p, \lambda_j+1)$ which
 swaps the second output qudit with the $\lambda_j+1$ qudit (which represents the first output qudit) gets absorbed by $\pisymt$.
Thus, the sum in this case is identical to that of Case II and we arrive at the same bound
$$  
\tr((O \otimes O) S_{23})
\leq 2 n \; \snorm{\infty}{O^2}.
$$

\textbf{Case IV}: $p, p' \in [\lambda_j]$. 
Fix some $\ket{a} \in E_n(w)$, 
$\ket{a'} = \ket{a}$ or
$\ket{a'} = \swapc( (j, p), (j, p') )
\ket{a}$.
In this case, we claim that the partial trace
$$
\tr_{-1, -2} 
\lp(  
\swapts(j, p, p') \cdot
\lp( 
\pisymt
\otimes I^{\otimes n} \rp) \cdot \bigg(I \otimes I \otimes (\Un{n}\ket{a} \bra{a'}\Und{n})\bigg)\rp)
$$
can always be simplified into 
$$
U^{\otimes 2} 
\ket{a_{j,p} a_{j', p'} }
\bra{a'_{j,p} a'_{j', p'} }
\Und{2}
\pisymt.
$$
The formal argument is deferred to \Cref{clm:tensor-caseIV}.
Then, we have 
\begin{align*}
&\tr \lp( \lp(O \otimes O \rp) 
\tr_{-1, -2} 
\lp(  
\swapts(j, p, p') \cdot
\lp( 
\pisymt
\otimes I^{\otimes n} \rp) \cdot \bigg(I \otimes I \otimes (\Un{n}\ket{a} \bra{a'}\Und{n})\bigg)\rp)
\rp) \\
&=
\tr \lp( \lp(O \otimes O \rp) 
U^{\otimes 2} 
\ket{a_{j,p} a_{j', p'} }
\bra{a'_{j,p} a'_{j', p'} }
\Und{2}
\pisymt
\rp)
\leq \snorm{\infty}{O}^2  \, ,
\end{align*}
where in the last inequality we again use \Cref{fact:eigen-bound}.


Let $S(j, p, p')$ be the set of $n$-qudit pairs $(\ket{a}, \ket{a'})$ such that
$\ket{a'} = \swapc((j, p), (j, p')) \ket{a}$ and $a \neq a'$.
We can therefore bound from above these terms' contributions to the variance by
\begin{align}
\tr((O \otimes O) S_{24})
&\leq
\bigO(1)
\sum_{j=1}^k
\sum_{p' = 2}^{\lambda_j}\sum_{p =1}^{p'-1}
\sum_{ (\ket{a}, \ket{a'}) \in S(j, p, p') }
\lp( |c_a| + |c_{a'}| \rp)^2
\snorm{\infty}{O}^2 \nonumber \\
&\leq 
\bigO(1) \;
\sum_{j=1}^k \lambda_j^2
\snorm{\infty}{O}^2 
\leq \bigO(n^2) \; \snorm{\infty}{O}^2 \, ,
\label{eq:a^2-term}
\end{align}
where in the first inequality we bound the non-vanishing partial-trace uniformly by $2 |c_a| |c_a'| \snorm{\infty}{O}^2$, 
in the second inequality we bound $\sum_{ a, a' \in S(j, p, p') }
\lp( |c_a| + |c_{a'}| \rp)^2$ 
from above
by $\bigO(1) \; \sum_{a \in E(w)} |c_a|^2$ via the triangle inequality, and in the third inequality we bound $\sum_{j=1}^k \lambda_j^2$ from above by $n^2$ using Cauchy's inequality.

Combining the above case analysis with the analysis of the term $S_1$ gives
\begin{align*}
\E[ \tr\lp( O  \Psi \rp)^2 ]
- \E[ \tr\lp( O  \Psi \rp) ]^2
&= 
\tr\lp( (O\otimes O) \; S_1 \rp)
+ 
\sum_{i=1}^4 \tr\lp( (O\otimes O) \; S_{2i} \rp)
- \E[ \tr\lp( O  \Psi \rp) ]^2 \\
&\leq  2k \; \tr(O^2) + 
O(n) \snorm{\infty}{O^2}
+ O(n^2) \snorm{\infty}{O}^2.    
\end{align*}
This concludes the proof of \Cref{lem:second-moment}.
\end{proof}

\subsection{Proof of \texorpdfstring{\Cref{thm:main_informal}}{main theorem}}
\label{sec:combine}
We first conclude the proof of \Cref{prop:population-shadow}.
\begin{proof}[Proof of \Cref{prop:population-shadow}]
Recall that we divide the joint state into $T = 10/\eps^2$ segments with equal size, where each segment is of the form $U^{\otimes (n/T)} \ket{e^{(t)}}$.

Denote $n' = n/T$.
If we perform the generic pre-processing step on $U^{\otimes n'} \ket{e^{(t)}}$, we obtain a Young tableau partition $\lambda^{(t)} \partition_d n'$, and end up with some state $ \Un{n'} \ket{\tau^{(t)}}$.
By \Cref{lem:preprocess-weight-space}, we have $k: = h(\lambda^{(t)}) \leq r$, where $r$ is the number of different symbols in $\ket{e}$, and $\ket{\tau^{(t)}}$ lives in the weight subspace $V(w^{(t)})$, where $w^{(t)}$ is the weight vector of $\ket{e^{(t)}}$, with probability $1$.
Conditioned on that, 
by \Cref{lem:expectation},
we have that
$$
\E\lp[ 
\lp( \Psi^{(t)} - h(\lambda^{(t)})  I \rp) \mid \ket{\tau^{(t)}} \rp]
=  
\sum_{i=1}^d 
w_i^{(t)}
U \ket{i}\bra{i}U^{\dagger}.
$$
Hence, if we take the sum over $t$, we have that
$$
\sum_{t=1}^T
\E\lp[ 
\lp( \Psi^{(t)} - h(\lambda^{(t)})  I \rp) \mid \ket{\tau^{(t)}} \rp]
= 
\sum_{t=1}^T
\sum_{i=1}^d 
w_i^{(t)}
U \ket{i}\bra{i}U^{\dagger}
= 
\sum_{i=1}^d 
w_i
U \ket{i}\bra{i}U^{\dagger} \, ,
$$
where $w$ is the weight vector of $\ket{e}$.
This shows that 
$
\tr \lp( O \; \frac{1}{n} \lp( \Psi^{(t)} - h(\lambda^{(t)})  I \rp) \rp) $ is indeed an unbiased estimator of the target quantity.

To analyze the variance of the final prediction, by \Cref{lem:second-moment}, 
we note that for all $\lambda^{(t)}$, it holds
$$
\Var\lp[  
\Psi^{(t)}
\mid \lambda^{(t)} \rp]
\leq 
2 \frac{h( \lambda^{(t)} )}{n'^2} \tr(O^2)
+ \bigO\lp(\frac{1}{n'}\rp) \snorm{\infty}{O^2}
+ \bigO \lp( \snorm{\infty}{O}^2 \rp).
$$
Since the expected value of $\Psi^{(t)}$ 
conditioned on $\lambda^{(t)}$ is the same for all possible values of $\lambda^{(t)}$, by the law of total variance, we have that the overall variance of $\Psi^{(t)}$ is bounded from above by the same quantity:
$$
\Var\lp[  
\Psi^{(t)}  \rp]
\leq 
2 \frac{h( \lambda^{(t)} )}{n'^2} \tr(O^2)
+ \bigO\lp(\frac{1}{n'}\rp) \snorm{\infty}{O^2}
+ \bigO \lp( \snorm{\infty}{O}^2 \rp).
$$
Recall that we always have 
$h( \lambda^{(t)} ) \leq r$.
This implies that whenever $n' \gg \sqrt{r \tr(O^2)}$ (which holds as long as $n \gg \sqrt{r \tr(O^2)}/\eps^2$), we will have 
$\Var\lp[  
\Psi^{(t)}  \rp] \leq \bigO(1)$.

Note that $\{ \Psi^{(t)} \}_{t=1}^T$ are all independent.
We thus have
$$
\Var\lp[  
\frac{1}{T}
\sum_{t=1}^T\Psi^{(t)}  \rp]
\leq \bigO(1/T) \leq \eps^2.
$$
The proposition then follows from Chebyshev's inequality.

\end{proof}

Finally, the proof of \Cref{thm:main_informal} follows from \Cref{prop:population-shadow} and \Cref{lem:population-reduction}.

\bibliographystyle{alpha}
\bibliography{bibliography}
\appendix
\section{More Representation Theory}
\label{app:more-rep-theory}
In this section, we provide additional background material on the representation theory of $\U(d)$ and $\S_n$.
Recall that in \Cref{def:irrep-isomorphism}, we define an isomorphism between irreps of the same vector space. We can extend this notion to any pair of representations even when they are defined on different vector spaces:
\begin{definition}[$G$-linear maps and Isomorphic Representations]
\label{def:general-isomorphism}
Let $V, W$ be two vector spaces.
Let $\rho_V:G \mapsto \GL(V)$, $\rho_W:G \mapsto \GL(W)$ be two representations of the group $G$.
We say that the linear map $F \colon V \to W$ is \emph{$G$-linear} if $ F ( \rho_V(g) \ket{x} ) = \rho_W(g) F( \ket{x} )$ for all $g \in G, \ket{x} \in V$. 
In other words, the diagram in \Cref{fig:equivariant_map_commutation_diagram} commutes.

We say that $\rho_V$ and $\rho_W$ are \emph{isomorphic} if there exists an \emph{invertible} $G$-linear map between them. We denote the isomorphism as $\rho_V \cong \rho_W$.
\end{definition}
\begin{figure}
\begin{center}
    \begin{tikzcd}[row sep=large,column sep=large,nodes={inner sep=5pt}]
        V \ar[r, "\rho_V(g)"] \ar[d, swap, "F"] &  V \ar[d, "F"] \\
        W \ar[r, swap, "\rho_W(g)"] & W .
    \end{tikzcd}
\end{center}
\caption{$F \colon V \to W$ is $G$-linear iff the diagram commutes.}
\label{fig:equivariant_map_commutation_diagram}
\end{figure}

A foundational lemma we will be using throughout the section is Schur's lemma:
\begin{lemma}[Schur's Lemma, see Proposition 1.16 from \cite{etingof2011introduction} e.g.]
\label{lem:schurs_lemma}
Let $\rho_V, \rho_W$ be two irreducible representations of a group $G$. We have:
\begin{enumerate}
    \item If $\rho_V \not\cong \rho_W$, then there are no non-trivial\footnote{Trivial $G$-linear maps send each element to $0$.} $G$-linear maps between them.
    \item If $V = W$ and $\rho_V \cong \rho_W$, the only non-trivial $G$-linear maps are scalar multiples of the identity.
\end{enumerate}
\end{lemma}
Recall that we use the following representations of $\S_n$ and {$\U(d)$} on the space $(\C^d)^{\otimes n}$:
\begin{align}
\label{eq:natural-representation}
\pi \mapsto P_\pi \coloneqq \sum_{i_1, \ldots, i_n \in [d]} \ket{i_{\pi^{-1}(1)}, i_{\pi^{-1}(2)}, \ldots, i_{\pi^{-1}(n)}}\bra{i_1, i_2, \ldots, i_n} \hspace{50pt} 
U \mapsto U^{\otimes n}.
\end{align}
Importantly, under these two representations, $(\C^d)^{\otimes n}$ can be decomposed into a direct sum of irreps. To be clear, for two subspaces $W_1,W_2$ of the same vector space $V$, we will say that $V$ is the direct sum of $W_1$ and $W_2$ (i.e., $V = W_1 \oplus W_2$) whenever every element $v \in V$ can be written uniquely as $v = w_1 + w_2$ where $w_1 \in W_1$ and $w_2 \in W_2$ (or in other words, $W_1$ and $W_2$ are orthogonal).
\begin{lemma}[Complete Reducibility]
\label{lem:complete_reducibility}
Let $G$ be $\S_n$ or $\U(d)$. Then, under its corresponding representation in Equation~\eqref{eq:natural-representation}, we have
$(\C^d)^{\otimes n} = \bigoplus_i V_i $ where the $V_i$s are irreps of $G$.
\end{lemma}
\begin{proof}
The statement for $\U(d)$ follows from Weyl's theorem (see Theorem 10.9 of \cite{fulton2013representation} e.g.) on complete reducibility, and the result for $\S_n$ is Maschke's Theorem (see Theorem 4.1.1 from \cite{etingof2011introduction} e.g.).
\end{proof}
For each decomposition of $(\C^d)^{\otimes n}$ in \Cref{lem:complete_reducibility}, we can group the subspaces that are isomorphic to each other. The benefit of this grouping is that the decomposition becomes unique in the following sense:
\begin{lemma}[Uniqueness of Isotypic Components, see Proposition 1.8 of \cite{fulton2013representation} e.g.]
\label{lem:iso-unique}
Let $G$ be $\S_n$ or $\U(d)$. Let $\bigoplus_{i,j} V_{i,j}$ be an orthogonal decomposition of $(\C^d)^{\otimes n}$ with the following properties:
\begin{itemize}[itemsep = 0pt]
    \item Each $V_{i,j}$ is an irrep of $G$.
    \item $V_{i,j} \cong V_{i,j'}$ for all $i, j, j'$.
    \item $V_{i,j} \not \cong V_{i',j'}$ for all $i \neq i'$ and $j, j'$.
\end{itemize}
In other words, each $\Pi_i \coloneqq \bigoplus_j V_{i,j}$ is direct sum of isomorphic irreps. The \emph{isotypic decomposition} $(\C^d)^{\otimes n} = \bigoplus_i \Pi_i$ is unique. In particular, for any irrep $W \subseteq (\C^d)^{\otimes n}$ with $W \cong V_{i,j}$, we have that $W \subseteq \Pi_i$. We refer to the subspace $\Pi_i$ as the \emph{$i$-th isotypic component}.
\end{lemma}
\begin{proof}
Let $W \subseteq (\C^d)^{\otimes n}$ be an arbitrary non-trivial irrep.
We claim that $W \subseteq \Pi_i$ for some $i$.
This immediately implies \Cref{lem:iso-unique}
since if the decomposition were not unique, we could  always find some non-trivial irrep $W \not \subseteq \Pi_i$ for any of the $\Pi_i$ above.

It remains to show the claim $W \subseteq \Pi_i$ for some $i$.
We define the set of maps $\phi_{i,j} \colon W \to V_{i,j}$ that are the orthogonal projectors from $W$ onto $V_{i,j}$. We claim that each map $\phi_{i,j}$ is $G$-linear. 
Let $\ket{x}$ be an arbitrary vector in $W$.
Denote by $\ket{\bar x}$ and $\ket{x^{\perp}}$ its component within $V_{i,j}$ and perpendicular to $V_{i,j}$ respectively.

We claim that $ \rho(g) \ket{x^{\perp}} $ is still perpendicular to $V_{i,j}$ for any $g \in G$.
 Let $\ket{v_1}, \cdots \ket{v_b}$ be an orthogonal basis of $V_{i,j}$.
We have $\ket{x^{\perp}}, \ket{v_1}, \cdots, \ket{v_b}$ are pair-wise orthogonal.
Since $\rho(g)$ is a unitary matrix, it preserves orthogonality.
Thus, $ \rho(g) \ket{x^{\perp}}, \rho(g) \ket{v_1}, \cdots, \rho(g) \ket{v_b}$ are still pair-wise orthogonal. 
Yet, $\rho(g) \ket{v_1}, \cdots, \rho(g) \ket{v_b}$ are all within $V_{i,j}$ and pair-wise orthogonal, implying that they still form a basis of $V_{i,j}$ after transformed by $\rho(g)$.
This shows that $\rho(g) \ket{x^{\perp}}$ is still perpendicular to $V_{i,j}$.
Thus, we have
$$
\phi_{i,j} \lp(  \rho(g) \ket{x} \rp)= 
\rho(g) \ket{ \bar x }
= \rho(g)  \phi_{i,j} \lp( \ket{x} \rp)
$$
for all $\ket{x} \in W$ and $g \in G$.

For any $\ket{w} \in W$ such that $\ket{w} \neq 0$, 
there exists some $i^*,j^*$ such that 
$\phi_{i^*,j^*}(\ket{w}) \neq 0$ since $\bigoplus_{i=1}^m \bigoplus_{j} V_{i,j}$ is a decomposition of the space $(\C^d)^{\otimes n}$.
Then $\phi_{i^*,j^*}$ is a non-trivial $G$-linear map between $W$ and $V_{i^*,j^*}$.
By Schur's Lemma (\Cref{lem:schurs_lemma}), we have $W \cong V_{i^*,j^*}$, so $W$ is an irrep of type $i^*$.
In fact, we can repeat the argument for \emph{any} $i', j'$ such that $\phi_{i',j'}(\ket{w}) \neq 0$ to show that $W \cong V_{i', j'}$. By composing the two isomorphisms, we get $V_{i^*, j^*} \cong V_{i', j'}$.
Hence, $\phi_{i',j'}(\ket{w}) \neq 0$ implies $i' = i^*$.
Therefore, $\ket{w}$ completely lies in the space $\bigoplus_{j} V_{i^*, j}$, showing that $W \subseteq \Pi_{i^*}$.
\end{proof}

We next provide constructions for the irreps of $\S_n$ and $\U(d)$ in $(\C^d)^{\otimes n}$. 
A canonical family of irreps for $\U(d)$ are obtained by the Young symmetrizer $Y_{\lambda}$ applied on  $(\C^d)^{\otimes n}$.
Since it is central to the following constructions, we first recall the definition of the Young symmetrizer:
\youngsymmetrizerdefinition*

\begin{definition}[Canonical Irreps of $\U(d)$]
\label{def:Q-canonical}
For $\lambda \partition_d n$, the \emph{canonical irrep of $\U(d)$} is $\mathcal Q_{\lambda} \coloneqq Y_{\lambda} (\C^d)^{\otimes n}$.
\end{definition}

\begin{lemma}[Labeling Irreps of $\U(d)$ (see 
Theorem 6.3  of \cite{fulton2013representation} e.g.)] \label{thm:young-Q-projector}
$Y_{\lambda} ( \C^d )^{\otimes n}$ is a non-trivial irrep of $\U(d)$ for every partition $\lambda \partition_d n$.
Moreover, every non-trivial irrep of $\U(d)$ in $(\C^d)^{\otimes n}$ under the group homomorphism in \Cref{eq:natural-representation}
is isomorphic to $Y_{\lambda} (\C^d )^{\otimes n}$ for exactly one choice of $\lambda$.\footnote{The theorem is usually stated with the notion of \emph{Schur Functors}, which are defined with the Young symmetrizer. We expand the definition of the Schur Functors here in our statement.}
\end{lemma} 
To describe the irreps of $\S_n$, we will need to recall the concepts of Young diagrams/tableau, introduce the additional notions of \emph{semi-standard, standard} Young tableau, 
and relate them to the standard basis vectors of $(\C^d )^{\otimes n}$:
\begin{definition}[Young diagrams / tableaux]
A \emph{Young diagram} is a visualization of an ordered partition $\lambda = (\lambda_1, \ldots, \lambda_k) \partition n$ with $\lambda_1 \ge \lambda_2 \ge \ldots \ge \lambda_k$ in which there are $\lambda_i$ boxes in the $i$th row of the diagram. 
For example, the partition $(4,3,1)$ corresponds to the diagram
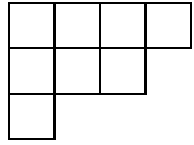
\begin{center}
    \ydiagram{4,3,1}
\end{center}
A \emph{Young tableau} with partition $\lambda$ is a Young diagram where each box has been labeled by some number. For example,
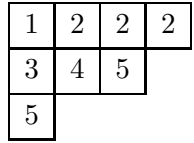
\begin{center}
    \begin{ytableau}
    1 & 2 & 2 & 2 \\
    3 & 4 & 5 \\
    5
\end{ytableau}
\end{center}
A \emph{semi-standard Young tableau} is such that the labels across each row (from left to right) are non-decreasing, and the labels across each column (from top to bottom) are strictly increasing. In a \emph{standard Young tableau}, each label $1, \ldots, n$ occurs exactly once, and the labels across both the rows and the columns are strictly increasing.
The \emph{canonical standard tableau} is a special standard tableau where the numbers $1 \cdots n$ are filled from left to right, and then from top to bottom in order. For example, the canonical standard tableau for partition $(4,3,1)$ is:
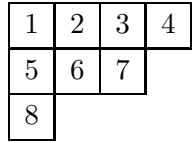
\begin{center}
    \begin{ytableau}
    1 & 2 & 3 & 4 \\
    5 & 6 & 7 \\
    8
\end{ytableau}
\end{center}
\end{definition}

\begin{definition}[Correspondence between standard basis vectors and tableau]
\label{def:corresponding_tableau}
To each partition $\lambda \partition n$ and standard basis vector $\ket{e} = \ket{e_1} \otimes \cdots \otimes \ket{e_n} \in (\C^d)^{\otimes n}$ with $e_i \in [d]$, we define their corresponding Young tableau {$T_{\lambda}(\ket{e})$} as follows: fill the numbers $e_1, \ldots, e_n$ into the boxes of the Young diagram for $\lambda$ from 
left to right, and then from top to bottom in order.\footnote{The order here needs to be consistent with the order the boxes of the tableau were indexed in the definitions of the permutation subgroups $A_{\lambda}, B_{\lambda}$ (see \Cref{def:row-column-permute}).}
\end{definition}


To begin with, we will need to fix a special standard basis vector $\ket{(\lambda)}$:
\begin{definition}
\label{def:cannonical-basis}
For $\lambda \partition n$, we define $\ket{(\lambda)}$ as the standard basis vector whose corresponding tableau of partition $\lambda$ has its $i$th row filled with $i$.
For example, for $\lambda = (3,2,1)$, we have
\[
\ket{(\lambda)} = \ket{1,1,1,2,2,3}.
\]
whose corresponding tableau $T_{\lambda}(\ket{(\lambda)})$ is (see \Cref{def:corresponding_tableau})
\begin{center}
\ytableausetup{smalltableaux}
\begin{ytableau}
    1 & 1 & 1 \\
    2 & 2  \\
    3
\end{ytableau}
\end{center}
\end{definition}

We note that by \Cref{lem:1d-min-weight-space}, the weight vector of $\ket{(\lambda)}$ is the lexicographically largest weight vector amongst standard basis vectors in the image of $Y_\lambda$. In particular, $Y_\lambda\ket{(\lambda)} \neq 0$.
We construct the irreps of $\S_n$ by looking at the minimum permutation invariant subspace containing the vector $Y_{\lambda} \ket{(\lambda)}$.
\begin{definition}[Canonical irrep of $\S_n$]
\label{def:P-canonical}
For $\lambda \partition n$, 
    the \emph{canonical irrep of $\S_n$} is \[\mathcal P_{\lambda} := \Span\{ P_{\pi} Y_{\lambda} \ket{(\lambda)} \mid \pi \in \S_n \}.\]
\end{definition}

The irreps of $\S_n$ are usually described as ``Specht modules'' for which the representation space is the group algebra $\C[\S_n]$. For our purposes, it will be useful to understand the irreps over the space $(\C^d)^{\otimes n}$. We will borrow the much of the background for Specht modules before specializing to our scenario.

\begin{definition}[Group Algebra]
Let $\C[\S_n]$ be the $(n!)$-dimensional complex vector space with the basis vectors labeled by permutations $\pi \in \S_n$. That is, each element of $\C[\S_n]$ is a linear combination of permutations: $\sum_{\pi \in \S_n} \alpha_\pi \pi $ for coefficients $\alpha_{\pi} \in \C$.
The vector space admits a natural representation of $\S_n$ where the group action is defined by
$$
\sigma \cdot \sum_{\pi \in \S_n} \alpha_{\pi} \pi = 
 \sum_{\pi \in \S_n} \alpha_{\pi} (\sigma \cdot \pi)
$$
where $\sigma \cdot \pi$ is the usual composition of permutations.
\end{definition}

We now define an important family of subspaces of $\C[\S_n]$ via the Young symmetrizer $Y_\lambda$. It's worth noting that we will abuse notation and use $Y_{\lambda}$ to denote both 
$  \sum_{a \in A_{\lambda}} a \cdot \sum_{b \in B_{\lambda}}  \sgn(b) b$ as an element in the group algebra $\C[\S_n]$, 
and its representation 
$\sum_{a \in A_{\lambda}} P_a  \sum_{b \in B_{\lambda}} \sgn(b) P_b$
in $(\C^{d})^{\otimes n}$ via the permutation operator (see \Cref{def:young_symmetrizer}).
\begin{definition}[Specht module]
\label{def:specht}
Let $\lambda \partition n$. Define the subspace (called a \emph{Specht module})
\[
\C[\S_n] Y_{\lambda} :=
\lp\{ \sum_{\pi} \alpha_\pi \pi \cdot  \sum_{a \in A_{\lambda}} a \cdot  \sum_{b \in B_{\lambda}} \sgn(b)\; b \mid \sum_{\pi} \alpha_\pi \pi \in \C[\S_n]  \rp\}.
\]
\end{definition}
It is well known that the Specht modules are the irreps of $\S_n$ up to isomorphism:
\begin{lemma}[Characterization of Irreps of $\S_n$ via Specht Modules (e.g.,  
Theorem 5.12.2 from \cite{etingof2011introduction})]
\label{thm:young-P-projector}
For $\lambda \partition n$, the Specht module $\C[\S_n] Y_{\lambda}$ is an irrep of $\S_n$.
Moreover, every irrep of $\S_n$ is isomorphic to $\C[\S_n] Y_{\lambda}$ for exactly one choice of $\lambda$.
\end{lemma}
We are now ready to show that 
the subspaces $\mathcal P_{\lambda}$ from \Cref{def:P-canonical} are indeed irreps of $\S_n$ in $(\C^d)^{\otimes n}$.
This is achieved by showing isomorphism between $\mathcal P_{\lambda}$ and the Specht Modules.
\begin{lemma}[Labeling Irreps of the $\S_n$] \label{thm:canonic-P}
Let $\lambda \partition_d n$.
Then $\mathcal P_{\lambda}$ is an irrep of $\S_n$.
Moreover, every non-trivial irrep of $S_n$ in $(\C^d)^{\otimes n}$ is isomorphic to some $\mathcal P_{\lambda}$ for exactly one choice of $\lambda$.
\end{lemma}
\begin{proof}
Note that for $\ket{\psi} \in \mathcal P_{\lambda}$, 
\[
P_\sigma \ket{\psi} \in \Span \{ P_\sigma P_\pi Y_{\lambda} \ket{(\lambda)} \mid \pi \in \S_n\}
= \Span \{  P_{\pi}Y_{\lambda} \ket{(\lambda)} \mid \pi \in \S_n\}
= \mathcal P_{\lambda}
\]
for any $\ket{\psi} \in W(\lambda)$ and $\sigma \in \S_n$.
Hence, $\mathcal P_{\lambda}$ is a representation space of $\S_n$.


To conclude the proof of \Cref{thm:canonic-P}, we will show that $\mathcal P_{\lambda}$ is isomorphic to $\C[\S_n] Y_{\lambda}$.
We first define a linear function $F \colon \C[\S_n] \to (\C^d)^{\otimes n}$ acting on the basis vectors of $\C[\S_n]$ as
\[
F( \pi ) = P_{\pi} \ket{(\lambda)}.
\]
First notice that $F$ is $G$-linear since
$$
F( \sigma \cdot \pi ) = P_{\sigma \cdot \pi} \ket{(\lambda)} = P_{\sigma} P_{\pi} \ket{(\lambda)}
= P_{\sigma} F( \pi ).
$$
To show that $F$ is invertible, we note that 
the image of the set $\C[\S_n]Y_{\lambda}$ under $F$ is exactly $\mathcal P_{\lambda}$ by the definitions.
It then suffices to show $\C[\S_n]Y_{\lambda}$ and 
$\mathcal P_{\lambda}$
have the same dimensions.
Denote $\text{Std}(\lambda) \subseteq \S_n$ as the subset of permutations that map the canonical standard tableau to a standard tableau; that is, $\pi \in \text{Std}(\lambda)$ if the group action (see \Cref{def:row-column-permute}) of $\pi$ on the canonical standard tableau is standard. For example, the permutation that swaps elements $5$ and $6$ is in $\text{Std}((3,2,1))$ since 
\[
\ytableausetup{centertableaux}
(5 \, 6) \cdot 
\begin{ytableau}
    1 & 2 & 3 \\
    4 & 5  \\
    6
\end{ytableau}
= 
\begin{ytableau}
    1 & 2 & 3 \\
    4 & 6  \\
    5
\end{ytableau}.
\]

A classic result says that
$\{ \pi Y_{\lambda} \mid \pi \in \text{Std}(\lambda) \} \subseteq \C[\S_n]$ forms a basis of the Specht Module (see Proposition~15.55 of \cite{fulton2013representation} e.g.). 
We will see the same set of permutations also gives us a basis of $\mathcal P_{\lambda}$. 
In particular, we will show that the set of vectors
$\{ F(\pi Y_{\lambda}) \mid \pi \in \text{Std}(\lambda) \}$ in $(\C^d)^{\otimes n}$ are linearly independent (and consequently a basis of $\mathcal P_{\lambda}$ since the dimension of $\mathcal P_{\lambda}$, as the image of the map $F$, is at most that of the Specht Module $\C[\S_n] Y_{\lambda}$).  

To do so, we first derive an order over permutations: given two permutations $\pi, \sigma \in \S_n$ which map the canonical standard tableau for $\lambda \partition n$ to $T_{\pi}, T_\sigma$, we say $\pi <_{\lambda} \sigma$ if there exists $i$ such that
\begin{itemize}[itemsep = 0pt]
    \item $i$ appears in a higher\footnote{By ``higher'', we mean it is visually above the other row in the Young diagram.} row in $T_\pi$ than $T_\sigma$.
    \item For all $j > i$, $j$ lies in the same row in $T_\pi$ and $T_\sigma$
\end{itemize}

First, we claim that $<_\lambda$ induces a total ordering on the elements of $\text{Std}(\lambda)$. To see that any two permutations in $\text{Std}(\lambda)$ are ordered, consider the last row of their corresponding tableaux. If the set of numbers in each row is not identical, then we can find the largest number $i$ which is in one tableau but not the other. By definition, one tableau is greater than the other in the $<_\lambda$ ordering. If the set of numbers is identical, there is only one way to order the numbers since the tableaux are standard. That is, the tableaux are identical on the last row. Therefore, we can repeat the same process on the second to last row, and so on. We claim that $<_\lambda$ is also transitive, but leave this as an exercise.

Let $\pi_1 <_{\lambda} \cdots  <_{\lambda} \pi_k$ be the permutations in $\text{Std}(\lambda)$ induced by this total ordering.
It is straightforward to check that 
$\pi_i \cdot b <_{\lambda} \pi_i$ for any non-identity $b  \in B_{\lambda}$, and $\pi_i \cdot a <_{\lambda} \pi_{j}$ for any $a \in A_{\lambda}$ if $\pi_i <_{\lambda} \pi_j$ from the definition of the ordering. 
Consequently, $\pi_j$ dominates any term
$
\pi_i \cdot a \cdot b
$, where $i < j$, $a \in A_{\lambda}$, and $b \in B_{\lambda}$.
This shows that $\pi_j Y_{\lambda}$, which has a non-zero coefficient on $\pi_j$, cannot lie in $\Span\{ \pi_i Y_{\lambda} \mid i < j \} \subseteq \Span\{ \pi_i \cdot a \cdot b \mid i < j, a \in A_{\lambda}, b \in B_{\lambda}\} $. 
Linear independence of the set of vectors $\{ \pi Y_{\lambda}  \mid \pi \in \text{Std}(\lambda)\}$ hence follows.

If $\pi >_{\lambda} \sigma$, then there exists $i$ such that
$\pi^{-1}(i)$ lies in a higher row than $\sigma^{-1}(i)$ in the canonical standard tableau.
Recall that $\ket{(\lambda)}$ is defined to have only $i$ in the $i$-th row. It then follows that
$P_{\pi} \ket{ (\lambda) }$ is different from
$P_{\sigma} \ket{(\lambda)}$ 
in the $i$-th tensor factor, which implies that
$$
\bra{(\lambda)} P_{\pi}^\dag P_{\sigma} \ket{(\lambda)} = 0.
$$
It then follows that 
$P_{\pi_k} Y_{\lambda} \ket{(\lambda)}$
contains a non-zero component of a standard basis vector
that is orthogonal to the space spanned by 
$P_{\pi_m} Y_{\lambda} \ket{(\lambda))}$ for all $m <k$.
This then allows us to conclude that the following set of vectors
$$
\left\{    P_\pi \cdot \sum_{b \in B_{\lambda}}  \sgn(b) P_b \cdot \sum_{ a \in A_{\lambda} }  P_a \ket{(\lambda))} \mid \pi \in \text{Std}(\lambda)  \right\}
$$
are linearly independent, implying that $\mathcal P_{\lambda}$ has the same dimension as the Specht Module, and hence the isomorphism between them.
\end{proof}


Combining \ref{thm:young-Q-projector} and \ref{thm:young-P-projector} then yield a proof of \Cref{lem:irrep_summary}: $\mathcal Q_{\lambda}$ from \Cref{def:Q-canonical} and $\mathcal P_{\lambda}$ from \Cref{def:P-canonical} serve as canonical irreps for $\U(d)$ and $\S_n$ respectively.

We will see next that irreps of $\U(d)$ and $\S_n$ that are isomorphic to $\mathcal P_{\lambda}$ and $\mathcal Q_{\lambda}$ respectively have tight connections: the pairs that have non-trivial intersections must share the same ``label'' $\lambda$.
This connection is made possible by the Schur-Weyl Duality.

\begin{lemma}[Irrep correspondence]
\label{lemma:grid}
Let $\ket{\psi} \in Q \cong \mathcal Q_{\lambda}$ and $\ket{\varphi} \in P \cong \mathcal P_{\lambda}$ for $P, Q \subseteq (\C^d)^{\otimes n}$.
Then, 
\begin{align*}
    \mathcal P_{\lambda} &\cong \Span \{ P_{\pi} \ket{\psi} \mid \pi \in \S_n  \} \\
    \mathcal Q_{\lambda} &\cong \Span \{ U^{\otimes n} \ket{\varphi} \mid  U \in \U(d)  \}
\end{align*}
are irreps of $\S_n$ and $\U(d)$, respectively.
\end{lemma}
\begin{proof}
For convenience, we will use the notation $P(\ket{\psi}) \coloneqq \Span \{ P_{\pi} \ket{\psi} \mid \pi \in \S_n  \}$. Let's focus on the claim that $\ket{\psi} \in Q \cong \mathcal Q_{\lambda}$ implies $ P(\ket{\psi}) \cong \mathcal P_{\lambda}$.
We begin by showing that the claim holds for $\ket{\psi} \in \mathcal Q_{\lambda}$, rather than, say, a subspace $Q$ isomorphic to $\mathcal Q_{\lambda}$. We have that for all $\ket{\psi} \in \mathcal Q_{\lambda}$, $P(\ket{\psi}) $ is an irrep isomorphic to $\mathcal P_{\lambda}$, which is the canonical irrep of type $\lambda$ as shown in \Cref{thm:canonic-P}.
We first note that $U^{\otimes n} \mathcal P_{\lambda}$ 
for any $U \in \U(d)$
is an irrep isomorphic to $\mathcal P_{\lambda}$ since $U^{\otimes n}$ is an invertible $G$-linear map between the two with respect to $\S_n$.
By \Cref{lem:irrep-decomposition}, for all $\ket{\psi} \in \mathcal Q_{\lambda}$, 
there exists a set of unitary matrices $\mathcal U$ and coefficients $\{\alpha_U\}_{U \in \mathcal U}$ such that
$$
\ket{\psi} = \sum_{U \in \mathcal U} \alpha_U U^{\otimes n} Y_{\lambda} \ket{(\lambda)}.
$$
since $Y_{\lambda} \ket{(\lambda)}$ is a non-zero vector in $\mathcal Q_{\lambda}$.
Hence, it can be seen that for any $\ket{\psi} \in \mathcal Q_{\lambda}$, $P(\ket{\psi})$ lies in the span of the set of irreps $\{ 
P( U^{\otimes n} Y_{\lambda} \ket{(\lambda)} ) \mid U \in \mathcal U \}$, 
which are all isomorphic to $\mathcal P_{\lambda}$.
Therefore, $P(\ket{\psi})$ lies in the $\lambda$-isotypic component of $\S_n$.
If $P(\ket{\psi})$ is irreducible (it can be checked $P(\ket{\psi})$ is a representation), it then follows from \Cref{lem:iso-unique} that $P(\ket{\psi}) \cong \mathcal P_{\lambda}$. 

It suffices to show $P(\ket{\psi})$ is irreducible.
Note that $P(\ket{\psi})$ is within the image of the set $\mathcal P_{\lambda}$ under the map $\sum_{U \in \mathcal U} \alpha_U U^{\otimes n}$. Hence, the dimension of $P(\ket{\psi})$ is at most that of $\mathcal P_{\lambda}$.
If $P(\ket{\psi})$ were to have a proper subrepresentation, that will give us an irrep within the $\lambda$ isotypic component that has smaller dimension than $\mathcal P_{\lambda}$, which is impossible.
This concludes the claim that $P(\ket{\psi}) \cong \mathcal P_{\lambda}$

We now extend the claim to any $\ket{\phi} \in Q \cong \mathcal Q_{\lambda}$.
Since $Q \cong \mathcal Q_{\lambda}$, there exists an invertible $\U(d)$-linear map $F \colon  \mathcal Q_\lambda \to Q$ between the two. Furthermore, since $Q, \mathcal Q_\lambda \subseteq (\C^d)^{\otimes n}$, $F$ is just a matrix, i.e., $F \in \GL((\C^d)^{\otimes n})$.
Since $F$ is surjective, there exists $\ket{\psi} \in \mathcal Q_{\lambda}$ such that 
$P(\ket{\phi}) = P( F \ket{\psi} )$.
Since $F$ is $\U(d)$-linear, it commutes with all $U^{\otimes n}$ for $U \in \U(d)$. Therefore, by \nameref{lem:schur-weyl}, $F = \sum_{\sigma \in \S_n} \beta_\sigma P_{\sigma}$.
Hence, 
\[
P(\ket{\phi}) = \Span \lp\{ 
P_{\pi} \sum_{\sigma \in \S_n} \beta_\sigma P_{\sigma} \ket{\psi} \middle \vert \pi \in \S_n \rp\} \subseteq \Span \{ 
P_{\pi} \ket{\psi} \mid \pi \in \S_n
\} \cong \mathcal P_{\lambda}.
\]
Since $\ket{\phi} \neq 0$, $P(\ket{\phi})$ is obviously not a trivial representation. On the other hand, we just showed it is a subrepresentation of some irrep isomorphic to $\mathcal P_{\lambda}$. Therefore, it must be equal to this irrep.

The other part of the claim concerning the irreps of $\U(d)$ can be shown with a similar argument.
In particular, we note that
$\mathcal Q_{\lambda}$ can be equivalently defined as 
$\Span \{ U^{\otimes n} Y_{\lambda} \ket{(\lambda)}
\mid U \in \U(d)\}$ by \Cref{lem:irrep-decomposition}.
This puts irreps of $\U(d)$ in a position that is completely symmetric to those of $\S_n$. The proof follows as all the properties of $\S_n$ used in the above argument are also properties of $\U(d)$.
\end{proof}

\section{Nice Schur basis construction}
\label{app:nice-schur-basis}

This appendix is devoted to proving \Cref{prop:nice-schur-basis}, restated below, which constructs a nice Schur basis (see \Cref{defn:nice_schur_basis}):
\niceschurbasistheorem*
We separate the proof of the construction of the nice Schur basis into two sections, justifying \Cref{func:schur_basis_completion} and \Cref{alg:orthonormal_schur_basis}, respectively.

\subsection{Schur Basis Completion}
\label{app:schur-basis-interpolation}
We first show that any set of vectors constructed via the process of \Cref{func:schur_basis_completion} will be a Schur basis.
\interpolatebasis*
\begin{proof}
We start by recalling the definition of the canonical irrep of $\U(d)$ of type $\lambda \partition_d n$:
\[
\mathcal Q_{\lambda} \coloneqq Y_{\lambda} (\C^d)^{\otimes n}
= \Span \{U^{\otimes n} Y_{\lambda} \ket{(\lambda)} \mid U \in \U(d)\}
\]
where the second equality follows from \Cref{lem:irrep-decomposition}.
By \Cref{lem:complete_reducibility} and \Cref{lem:iso-unique}
, $(\C^d)^{\otimes n}$ uniquely decomposes into $\bigoplus_{ \lambda \partition_d n}
\Pi_{\lambda}
$, where $\Pi_{\lambda}$ is the isotypic component of $\U(d)$ of type $\mathcal Q_{\lambda}$.
We will show it is exactly of the form
\begin{align} \label{eq:isotypic-def}
\Pi_{\lambda} = \Span \{
P_{\pi} U^{\otimes n} Y_{\lambda} \ket{(\lambda)}
\mid U \in \U(d), \pi \in \S_n
\}.
\end{align}
We start with the fact that $\mathcal Q_{\lambda} \subseteq \Pi_{\lambda}$.
Moreover, 
$P_{\pi} \mathcal Q_{\lambda}$ gives an irrep that is isomorphic to $\mathcal Q_{\lambda}$
since $P_{\pi}$ is invertible and $\U(d)$-linear.
Therefore, we have $P_{\pi} \mathcal Q_{\lambda} \subseteq \Pi_{\lambda}$ by \Cref{lem:iso-unique}.
It follows that
 $$
 \Span \{
P_{\pi} U^{\otimes n} Y_{\lambda} \ket{(\lambda)}
\mid U \in \U(d), \pi \in \S_n
\} \subseteq \Pi_{\lambda}.
$$
On the other hand, for any irrep $Q \subseteq (\C^d)^{\otimes n}$ such that
$Q \cong \mathcal Q_{\lambda}$, there exists some $\U(d)$-linear map $F$ between $Q$ and $\mathcal Q_{\lambda}$.
By the definition of $\U(d)$-linearity, $F$ commutes with all $U^{\otimes n}$.
Hence, by \nameref{lem:schur-weyl}, 
we know $F$ must be of the form
$\sum_{ \pi \in \S_n } \alpha_{\pi} P_{\pi}$.
In other words, we have
$Q = \sum_{ \pi \in \S_n } \alpha_{\pi} P_{\pi} \mathcal Q_{\lambda}$.
Hence, we also have
$$
  \Pi_{\lambda} \subseteq \Span \{
P_{\pi} U^{\otimes n} Y_{\lambda} \ket{(\lambda)}
\mid U \in \U(d), \pi \in \S_n
\}.
$$
Equation~\eqref{eq:isotypic-def} therefore follows.

The argument will be exactly the same if we start with $\Pi_{\lambda}$ being the isotypic component of $\S_n$ and we will arrive at
$$
\Pi_{\lambda} = \Span \{
U^{\otimes n} P_{\pi}  Y_{\lambda} \ket{(\lambda)}
\mid U \in \U(d), \pi \in \S_n
\}.
$$
The two defining equations are obviously equivalent and hence we show that the isotypic subspaces of $\U(d)$ and $\S_n$ coincide.




Recall that for each $\lambda \partition_d n$, we are given as input a basis $\{\ket{(\lambda, i,0)}\}_i$ for $\mathcal Q_\lambda$. From the basis vector $\ket{(\lambda, i,0)}$, \Cref{func:schur_basis_completion} creates a basis for $\mathcal P = \Span\{  P_{\pi} \ket{(\lambda, 0, 0)}  \mid \pi \in \S_n\}$:
\[
\ket{(\lambda, 0, j)} \coloneqq \sum_{\pi \in \S_n} \alpha_{ \pi}^{(\lambda, j)} P_{\pi} \ket{(\lambda, 0, 0)}
\]
with coefficients $\alpha_{\lambda, \pi}^{(j)} \in \C$ for all $j \in \{0, \ldots, \dim \mathcal P_\lambda - 1\}$. Notice that \Cref{lemma:grid} implies that $\mathcal P \cong \mathcal P_\lambda$.

Let us now define the subspaces $P_{\lambda, i}, Q_{\lambda, j}$ as
\begin{align*}
P_{\lambda, i} &\coloneqq \Span \{ P_{\pi} \; \ket{(\lambda, i, 0)} \mid \pi \in \S_n\} , \\ 
Q_{\lambda, j} &\coloneqq \Span \{ U^{\otimes n} \; \ket{(\lambda, 0, j)} \mid U \in \U(d) \}.
\end{align*}
Once again, by \Cref{lemma:grid} we have $P_{\lambda, i} \cong \mathcal P_{\lambda}$ and $Q_{\lambda, j} \cong \mathcal Q_{\lambda}$. Moreover, combining the fact that the isotypic components of $\S_n$ and $\U(d)$ collide with the uniqueness of the isotypic decomposition, we get that both $Q_{\lambda, j}$ and $P_{\lambda, i}$ are in fact subspaces of $\Pi_{\lambda}$.


To show the lemma, we make the following claims:
\begin{enumerate}
    \item \label{claim:i1} The $P_{\lambda, i}$ span the space of $\Pi_{\lambda}$ (the same for $Q_{\lambda, j}$), i.e., $$\Pi_{\lambda} = \Span \bigg \{ \bigcup_{i=0}^{\dim \mathcal Q_{\lambda}-1} P_{\lambda, i} \bigg \} = \Span \bigg \{\bigcup_{j=0}^{\dim \mathcal P_{\lambda} -1} Q_{\lambda, j}\bigg\}$$
    \item \label{claim:i2} The subspaces $P_{\lambda, i}$ are linearly independent. That is,  $\dim \Span\{ \cup_i P_{\lambda, i} \} = \sum_i \dim P_{\lambda, i}$. The analogous statement holds for the $Q_{\lambda, j}$.
    \item \label{claim:i3} $\ket{(\lambda, i, j)} \in P_{\lambda, i} \cap Q_{\lambda, j}$.
\end{enumerate}

If the above holds, 
$\ket{(\lambda, i, j)}$ must be linearly independent from all $\ket{(\lambda, i', j')}$ for all $j'$ and $i' < i$ since 
$\ket{(\lambda, i, j)} \in P_{\lambda, i}$, $\ket{(\lambda, i', j')} \in P_{\lambda, i'}$, and claim (2) says that $P_{\lambda, i}$ is linearly independent from all $P_{\lambda, i'}$ for $i' < i$.
Similar argument gives $\ket{(\lambda, i, j)}$ must be linearly independent from all $\ket{(\lambda, i', j')}$ for all $i'$ and $j' < j$.
Hence, it follows that the set of vectors constructed are all linearly independent.
We then immediately have $\{\ket{(\lambda, i, j)} \}_{i=0}^{\dim \mathcal Q_{\lambda}-1}$ spans $Q_{\lambda, j}$, and 
$\{\ket{(\lambda, i, j)} \}_{j=0}^{\dim \mathcal P_{\lambda}-1}$ spans $P_{\lambda, i}$ since $Q_{\lambda, j}$ and $P_{\lambda, i}$ must have the same dimension as $\mathcal Q_{\lambda}$ and $\mathcal P_{\lambda}$ respectively. 
Together, this implies that the vectors $\ket{(\lambda, i, j)}$ form a Schur basis and conclude the proof of \Cref{lem:interpolate-basis}. We now proceed to show the three claims.

\paragraph{Proof of claim~\ref{claim:i1}} 
Notice that we could define $\Pi_\lambda$ in Equation~\eqref{eq:isotypic-def} starting from any state in $\mathcal Q_\lambda$ (not just $Y_\lambda \ket{(\lambda)}$) using the exact same analysis. By assumption, $\ket{(\lambda, 0, 0)}$ is a basis vector for $\mathcal Q_\lambda$, so we can write
$$
\Pi_{\lambda} = \Span \{ P_{\pi} U^{\otimes n} \ket{(\lambda, 0, 0)} 
\mid \pi \in \S_n, U \in \U(d)
\}.
$$
Therefore, we can write any vector $\ket{\psi} \in \Pi_{\lambda}$ as
\begin{align*}
   \ket{\psi} &= 
   \sum_{\pi} P_{\pi}
   \lp(\sum_{U} \beta_{U, \pi}
   U^{\otimes n} \ket{(\lambda, 0, 0)} \rp).
\end{align*}
Since $\ket{(\lambda, 0, 0)} \in \mathcal Q_\lambda$ and $\mathcal Q_\lambda$ is an irrep of $\U(d)$, we can expand any state $U^{\otimes n}\ket{(\lambda, 0, 0)}$ in the basis of $\mathcal Q_\lambda$. Hence, the inner summation above can be expanded as 
\[
\sum_{U} \beta_{U, \pi}
   U^{\otimes n} \ket{(\lambda, 0, 0)}
= \sum_{i} \gamma_{\pi}^{(i)} \ket{(\lambda, i, 0)}.
\]
for coefficients $\gamma_{\pi}^{(i)} \in \C$. Together, this gives
\begin{align*}
\ket{\psi} &= \sum_{\pi} P_{\pi}
\sum_{i} \gamma_{\pi}^{(i)} \ket{(\lambda, i, 0)}
= \sum_{i, \pi} \gamma_{\pi}^{(i)}
P_{\pi} \ket{(\lambda, i, 0)}.
\end{align*}
Note that for each of the summand we have
$P_{\pi} \ket{(\lambda, i, 0)} \in 
P_{\lambda, i}
$ by definition.
Therefore, it follows that 
$\Pi_{\lambda} = \Span \bigg \{ \bigcup_{i=0}^{\dim \mathcal Q_{\lambda} - 1} P_{\lambda, i} \bigg \}$. Similar argument shows $\Pi_{\lambda} = \Span \bigg \{\bigcup_{j=0}^{\dim \mathcal P_{\lambda}-1} Q_{\lambda, j}\bigg\}$.


\paragraph{Proof of claim~\ref{claim:i2}} 
Here, we invoke the well-known fact that $\lambda$-isotypic component of $(\C^d)^{\otimes n}$ has dimension $\dim \mathcal Q_\lambda \times \dim \mathcal P_\lambda$ (see, e.g., \cite{haah2016sample}). By dimension counting, we know the only way for $\{P_{\lambda, i}\}_{i=0}^{\dim \mathcal Q_{\lambda}-1}$ to span the subspace
$\Pi_{\lambda}$ of dimension at least $\dim \mathcal P_{\lambda} \times \dim \mathcal Q_{\lambda}$
is that they are all linearly independent.
The same argument gives $Q_{\lambda, j}$ must also be linearly independent.

\paragraph{Proof of claim~\ref{claim:i3}} Recall that we define the basis vector
$$
\ket{(\lambda, i, j)}
\coloneqq \sum_{ \pi \in \S_n } 
\alpha_{\pi}^{(\lambda, j)}
P_{\pi} \ket{(\lambda, i, 0)}.
$$
Hence, $\ket{(\lambda, i, j)} \in P_{\lambda, i}$ follows by definition.
We claim that we also have
$ \ket{(\lambda, i, j)} \in Q_{\lambda, j} $.
Again using \Cref{lem:irrep-decomposition}, we have that
\[
\ket{(\lambda, i, 0)}
= \sum_{ U }  \zeta_U^{(i)} U^{\otimes n} \ket{(\lambda, 0, 0)}.
\]
for coefficients $\zeta_U^{(i)} \in \C$. This further implies that
$$
\ket{(\lambda, i, j)} = 
\sum_{ U }  \zeta_U^{(i)} U^{\otimes n}
\sum_{ \pi \in \S_n } 
\alpha_{\pi}^{(\lambda, j)}
P_{\pi}  \ket{(\lambda, 0, 0)}
= \sum_{ U }  \zeta_U^{(i)} U^{\otimes n}  \ket{(\lambda, 0, j)}\, ,
$$
which is in $Q_{\lambda, j}$ by definition.
\end{proof}

\subsection{Orthonormality of Nice Schur Basis}
\label{app:orthogonality}

We start by showing the vanishing condition of the Young symmetrizer:
\nullspaceyoungsymmetrizer*
\begin{proof}
Let $\ket{e}$ be a standard basis vector and $T_{\lambda}(\ket{e})$ be its corresponding tableau (see \Cref{def:corresponding_tableau}).
We start with a simple observation: if some number $x \in [d]$ appears more than once in a column of $T_{\lambda}(\ket{e})$, then the signed sum over the column permutation group $\sum_{b \in B_{\lambda}} \sgn(b) P_b \ket{e}$ is $0$, and consequentially, $Y_\lambda \ket{e} = 0$. To finish the proof, we claim that when $\mathcal W(\ket{e})$ is not majorized by $\lambda$, there must be some column of that has a repeated element.

To prove the claim, first notice that we can permute the columns of $T_{\lambda}(\ket{e})$ without affecting the number of elements in each column. Therefore, permute each column according to the sorted weight vector $\mathcal W(\ket{e})_{(i)}^{\downarrow}$ so that the elements $x \in [d]$ with higher frequencies in $\ket{e}$ appear higher in the column. 

Since $\mathcal W(\ket{e})$ is not majorized by $\lambda$, there exists $m$ such that $\sum_{i=1}^m \mathcal W(\ket{e})_{(i)}^{\downarrow} > \sum_{i=1}^m \lambda_i$. In other words, the sum of the first $m$ elements with the highest frequencies is greater than the number of boxes in the first $m$ rows of the tableau. Therefore, some element $x \in [d]$ which is in the first $m$ elements with highest frequency must appear in row $m+1$ or greater. However, since the columns are sorted, this implies that there are at least $m$ elements above it with higher (or equal) frequency. This implies some element must have been repeated.
\end{proof}

Notice that in \Cref{alg:orthonormal_schur_basis}, we iterate through weight vectors in reverse lexicographical order to construct the nice Schur basis. Importantly, the first weight vector $w$ whose weight space $V(w)$ has nontrivial intersection with the image of the Young symmetrizer (i.e., $Y_\lambda V(w) \neq 0$) will also be our first basis vector in our nice Schur basis. Using \Cref{clm:non-vanish}, we can explicitly construct this weight vector:
\onedminweightspace*
\begin{proof}
By \Cref{clm:non-vanish}, if $\dim Y_{\lambda} V(w) > 0$, then $w$ is majorized by $\lambda$.
Besides, we also need the sum of the entries in $w$ to be $n$.
Among such weight vectors, it is not hard to see that the $w$ described in the statement of the lemma is the lexicographically smallest.

Next we show the space is $1$-dimensional. It is easy to see that the intersection is not the null space since one can verify that the special basis vector $\ket{(\lambda)}$ (see \Cref{def:cannonical-basis}), whose corresponding tableau has the $i$-th row filled with the number $i$ is not in the kernel of $Y_{\lambda}$ and it has weight exactly $w$. 
Hence, it suffices to show $V(w) \cap \Img(Y_{\lambda})$ is at most $1$-dimensional. In fact, we will show that
\begin{align}
\label{eq:column-stablizer-dimension}
\text{dim} \lp( \lp(\sum_{b \in B_{\lambda}} \sgn(b) P_b\rp) V(w)
\rp) = 1. 
\end{align}
Given that, it follows that 
$Y_{\lambda} V(w) = 
\sum_{ a \in A_{\lambda} } P_{a} \lp( \Img\lp(\sum_{b \in B_{\lambda}} \sgn(b) P_b\rp) V(w)
\rp)$ is at most $1$-dimensional (since the linear operator 
{$\sum_{ a \in A_{\lambda} } P_a $} can only further decrease the dimension). 

It remains to show \Cref{eq:column-stablizer-dimension}. We will show below if two standard basis vector $\ket{u}, \ket{v} \in V(w)$ satisfies $\sum_{b \in B_{\lambda}} \sgn(b) P_b \ket{u} \neq 0$ and $\sum_{b \in B_{\lambda}} \sgn(b) P_b \ket{v} \neq 0$, then there exists $\pi \in B_\lambda$ such that $$
\sum_{b \in B_{\lambda}} \sgn(b) P_b \ket{v}
    =
    \sgn(\pi)
    \sum_{b \in B_{\lambda}} \sgn(b) P_b \ket{u}.$$
This will be enough for us because if we pick an arbitrary vector $\ket{l} \in V(w)$ and suppose $\sum_{b \in B_{\lambda}} \sgn(b) P_b \ket{v} \neq 0$. Then, we can write it as a linear combination of standard basis vectors in $V(w)$ i.e. $\ket{l} = \sum_{e \in V(w)} \alpha_e \ket{e}$, and therefore 
$$
\sum_{b \in B_{\lambda}} \sgn(b) P_b \ket{l} = \sum_{e \in V(w)} \alpha_e \lp(\sum_{b \in B_{\lambda}} \sgn(b) P_b  \ket{e}\rp)
$$ where each of the inner terms are the same vector upto a sign change.

Now, we prove above claim for the standard basis vector $\ket{v}$. Let $\ket{v}$ be a standard basis vector in $V(w)$. Furthermore, suppose $\sum_{b \in B_{\lambda}} \sgn(b) P_b \ket{v} \neq 0$.
Then the first 
$ \lambda_k $ columns in the corresponding tableau (these are the columns with $k$ cells) have to each contain the symbols $\{1, 2, \cdots k\}$ (in any order) since this is the only way for the symbols to not repeat.
At this point, we have already run out of all the symbols $k$.
Hence, the next $\lambda_{k-1} - \lambda_k$ columns (these are the columns with $(k-1)$ cells) have to each contain the symbols $\{1, 2, \cdots k-1\}$ (in any order) with a similar argument.
One can repeat the argument to show that the $\lambda_{i+1} + 1$ to $\lambda_{i}$ columns must each contain the symbols $\{1, \cdots, i\}$.
Hence, we can conclude that
for two vectors $\ket{v}, \ket{u}$, if $
\sum_{b \in B_{\lambda}} \sgn(b) P_b \ket{v} \neq 0$ and
$\sum_{b \in B_{\lambda}} \sgn(b) P_b \ket{u} \neq 0$, then each pair of columns in the corresponding tableaux of $\ket{u}, \ket{v}$ must be the same up to column-wise permutation. 
Hence, there exists $\pi \in B_{\lambda}$ such that
$ \ket{v} = P_{\pi} \ket{u}$.
Thus, we have
\begin{align*}
    \sum_{b \in B_{\lambda}} \sgn(b) P_b \ket{v}
    =  
    \sum_{b \in B_{\lambda}} \sgn(b) P_b P_{\pi} \ket{u}
    =
    \sgn(\pi) \; \sum_{b \in B_{\lambda}} \sgn(b \cdot \pi ) P_{b \cdot \pi}  \ket{u}
    = 
    \sgn(\pi)
    \sum_{b \in B_{\lambda}} \sgn(b) P_b \ket{u}.
\end{align*}
One can see the two resulting vectors are the same up to a sign change.
This shows the intersection space is $1$-dimensional, and concludes the proof of \Cref{lem:1d-min-weight-space}.
\end{proof}

As an immediate corollary, we have that vectors of the form $\ket{(\lambda, 0, k)}$ and $\ket{(\lambda, i, \ell)}$ for $i > 0$ must live in different weight subspace. Their orthogonality hence follows.
\basecaseorthogonality*
\begin{proof}
This follows from the definition:
we define $\ket{(\lambda, 0, 0)} \coloneqq \ket{(\lambda)}$, i.e., the standard basis vector with lexicographically largest weight vector in $\Img(Y_{\lambda})$.
{Since this space is $1$-dimensional by \Cref{lem:1d-min-weight-space} and $\Img(Y_{\lambda})$ is spanned by linearly independent vectors $\{\ket{(\lambda, i, 0)}\}_{i\geq 0}$ by construction in \Cref{alg:orthonormal_schur_basis}, there is only one vector among $\{\ket{(\lambda, i, 0)}\}_{i\geq 0}$ that can lies in such $V(w)$. Since, $\ket{(\lambda, 0, 0)}$ lies in $V(w)$, weight of $\ket{(\lambda, i, 0)}$ for $i\neq 0$ must be different from that of $\ket{(\lambda, 0, 0)}$. Moreover, since $\ket{(\lambda, i, l)} = P_\pi \ket{(\lambda, i, 0)}$ for some permutation $\pi$, the weight of $\ket{(\lambda, i, l)}$ for $i\neq 0, l\geq 0$ must be different from that of $\ket{(\lambda,0 , 0)}$. The orthogonality immediately follows.}
\end{proof}


We are now ready to show orthogonality of the full basis.
Fix a partition $\lambda \partition [n]$.
To show the orthogonality of the constructed basis vectors, 
we will introduce the following function {$\gamma_{\lambda}: \{0,\ldots,\dim Q_{\lambda} -1\} \times \U(d) 
 \times \{0,\ldots,\dim Q_{\lambda} -1\} \mapsto \C$ }defined as
$$
\gamma_{\lambda}(j, U, i) = \bra{ (\lambda, j, 0) } U^{\otimes n} \ket{ (\lambda, i, 0) }.
$$
In plain language, fixing $U$ and $i$ and varying $j$, 
the function gives us the coefficients of the vector
$U^{\otimes n} \ket{(\lambda, i, 0)}$ represented under the set of basis we have chosen for $Q_{\lambda}$. 
In particular, we have
\begin{equation}
\label{eq:gamma-expansion}
U^{\otimes n} \ket{(\lambda, i, 0)}    
= \sum_{ q=0 }^{ \dim Q_{\lambda}-1 }
\gamma_{\lambda}(q, U, i) \ket{ (\lambda, q, 0) }.
\end{equation}
\completeschurbasisorthogonality*
\begin{proof}
We can expand the definition of $\ket{(\lambda, i, k)}, \ket{(\lambda, j, \ell)}$ (see \Cref{func:schur_basis_completion}) and get
\begin{align}
\label{eq:permute-expansion}
\braket{(\lambda, i, k)}{(\lambda, j, \ell)}
= \sum_{ \pi, \sigma \in \S_n} 
\overline{\alpha_{\lambda, \pi}^{(k)}}
\alpha_{\lambda, \sigma}^{(\ell)} 
\bra{ (\lambda, i, 0) } P_{\pi}^{\dagger} P_{\sigma}
\ket{(\lambda, j, 0)}.
\end{align}
Note that $\ket{(\lambda, j, 0)}$, $\ket{(\lambda, 0, 0)}$
both lie in $Q_{\lambda}$. 
Hence, using \Cref{lem:irrep-decomposition}, there exists a finite set of unitary matrices $\mathcal U_{\lambda, j} \subset \U(d)$ and a set of coefficients $\beta_{U,j}$ indexed by them such that
$$
\ket{(\lambda, j, 0)} = \sum_{ U \in \mathcal U_{\lambda, j} } \beta_{U,j} \; U^{\otimes n} \; \ket{(\lambda, 0, 0)}.
$$
Substituting this into \Cref{eq:permute-expansion} gives that
\begin{align}
\braket{(\lambda, i, k)}{(\lambda, j, \ell)}
&= 
\sum_{U \in \mathcal U_{\lambda, j}} \beta_{U,j}
\sum_{ \pi, \sigma } 
\overline{\alpha_{\lambda, \pi}^{(k)}}
\alpha_{\lambda, \sigma}^{(\ell)}
\bra{ (\lambda, i, 0) }
P_{\pi}^{\dagger} P_{\sigma}
U^{\otimes n} \ket{(\lambda, 0, 0)} \nonumber \\
&= 
\sum_{U \in \mathcal U_{\lambda, j}} \beta_{U,j}
\sum_{ \pi, \sigma } 
\overline{\alpha_{\lambda, \pi}^{(k)}}
\alpha_{\lambda, \sigma}^{(\ell)} 
 \lp( (U^{\otimes n})^{\dagger} \ket{ (\lambda, i, 0) } \rp)^{\dagger}
P_{\pi}^{\dagger} P_{\sigma}
\ket{(\lambda, 0, 0)} \, ,
\label{eq:second-expansion}
\end{align}
where in the second equality we use the fact that $U^{\otimes n}$ commute with any permutation operator.
Note that $(U^{\otimes n})^{\dagger} \ket{ (\lambda, i, 0)}$ is still within the subspace $Q_{\lambda}$. 
Hence, using the definition of $\gamma_{\lambda}$ function (specifically \Cref{eq:gamma-expansion}), we can write
$$
(U^{\otimes n})^{\dagger} \ket{ (\lambda, i, 0)}
= \sum_{q=0}^{ \dim Q_{\lambda} - 1 } \gamma_{\lambda}(q, U^{\dagger}, i) \; \ket{(\lambda, q, 0)}.
$$
Substituting this into \Cref{eq:second-expansion} gives that
\begin{align*}
\braket{(\lambda, i, k)}{(\lambda, j, \ell)}
&= 
\sum_{U \in \mathcal U_{\lambda, j}  }
\sum_{q=0}^{\dim Q_{\lambda} - 1}
\beta_{U,j}  
\overline{\gamma_{\lambda}(q, U^{\dagger}, i)}
\sum_{ \pi, \sigma } 
\overline{\alpha_{\lambda, \pi}^{(k)}}
\alpha_{\lambda, \sigma}^{(\ell)} 
 \bra{\lambda, q, 0}
P_{\pi}^{\dagger} P_{\sigma}
\ket{(\lambda, 0, 0)} \\
&= \sum_{U \in \mathcal U_{\lambda, j},q} 
\beta_{U,j}
\overline{\gamma_{\lambda}(q, U^{\dagger}, i)}
\braket{(\lambda, q, k)}{ (\lambda, 0, \ell) }.
\end{align*} 
When $k \neq l$, we claim each term in the summation must be $0$.
If $q = 0$, we know $\ket{(\lambda, 0, k)}$ and $\ket{(\lambda, 0, \ell)}$ must be orthogonal for $k \neq \ell$ by construction {in \Cref{func:schur_basis_completion}}.
If $q \neq 0$, the orthogonality follows from \Cref{cor:base-orthogonality}.

When $k = \ell$ and $i \neq j$, we need a more careful analysis. 
When $q \neq 0$, we still have $\ket{\lambda, q, k}$, $\ket{\lambda, 0, k}$ are orthogonal
by \Cref{cor:base-orthogonality}. 
Therefore, the only remaining terms in the summation are those with $q = 0$.
Observe that the vectors $\ket{\lambda, q , k} $ and $\ket{\lambda, 0 , \ell} $ are the same vector when $q = 0$ and $k = \ell$.
Hence, we can simplify the expression as
$$
\braket{(\lambda, i, k)}{(\lambda, j, \ell)} = \braket{(\lambda, i, k)}{(\lambda, j, k)} = 
\sum_{U \in \mathcal U_{\lambda, j}} \beta_{U,j}
\overline{\gamma_{\lambda}(0, U^{\dagger}, i)}.
$$
We claim this is exactly the inner product of $\ket{\lambda,i ,0}$ and $\ket{\lambda, j, 0}$.
Observe that
\begin{align*}
\braket{\lambda,i ,0}{\lambda, j, 0}
= \sum_{U \in \mathcal U_{\lambda, j}} \beta_{U,j} \bra{\lambda, i, 0} U^{\otimes n} \ket{\lambda, 0, 0}
= \sum_{U \in \mathcal U_{\lambda, j}} 
\sum_{q = 0}^{ \dim \mathcal Q_{\lambda} - 1 }
\beta_{U,j}  
\overline{\gamma_{\lambda}(q, U^{\dagger}, i)} \braket{\lambda, q, 0}{\lambda, 0, 0}.
\end{align*} 
Still, $\braket{\lambda, q, 0}{\lambda, 0, 0}$ is $1$ if $q = 0$ and $0$ otherwise. 
Hence, the summation can also be simplified, which yields the identity
$$
\braket{(\lambda, i, 0)}{(\lambda,j, 0)}
=
\sum_{U \in \mathcal U_{\lambda, j}} \beta_{U,j}
\overline{\gamma_{\lambda}(0, U^{\dagger}, i)}.
$$
This expression is exactly $0$ since $\braket{(\lambda, i, 0)}{(\lambda,j, 0)} = 0$ for $i \neq j$ by construction {in \Cref{alg:orthonormal_schur_basis}}. 
Therefore, we conclude that $\braket{(\lambda, i, k)}{(\lambda, j, \ell)} = 0$ if $i \neq j$ and $k = \ell$ as well. Combining the two cases allows us to conclude the proof of \Cref{lem:basis-orthogonality}.

Lastly, when $k = \ell$ and $i = j$, we argue that the inner product equals $1$.
In particular, we can reuse the above computation and conclude that
$$
\braket{(\lambda, i, k)}{(\lambda, i, k)} =
\sum_{U \in \mathcal U_{\lambda, i}} \beta_{U,i}
\overline{\gamma_{\lambda}(0, U^{\dagger}, i)}
= 
\braket{(\lambda, i, 0)}{(\lambda,i, 0)} \, ,
$$
and the right hand side is $1$ by definition since $\{  \ket{(\lambda, i, 0)} \}_{i}$ are orthonormal basis 
found through the Gram Schmidt process.
\end{proof}

\begin{proof}[Proof of \Cref{prop:nice-schur-basis}]
It follows from \Cref{lem:interpolate-basis}  that the basis are Schur basis.
By \Cref{lem:basis-orthogonality}, the vectors within each isotypic component $\Pi_{\lambda}$ are pair-wise orthogonal. 
For pairs of vectors that live in different $\Pi_{\lambda}$, their orthogonality follow from the orthogonality of the isotypic components, which is implied by \Cref{lem:iso-unique}.
\end{proof}

\section{Omitted proofs for Variance Computation}
In this subsection, we will working with an $(n+2)$-qudit quantum system. 
We will refer to the first and second qudits as the first and second output qudits respectively, and again associate a Young diagram with partition $\lambda \partition_d n$ to the remaining qudits. We use $(i,j)$ to index the qudit corresponding to  the $j$-th cell in the $i$-th row of the Young diagram. We define swap permutations for our $n+2$ qudit system.

\begin{claim}
\label{prop:tensor2}
Let $\ket{a}$ be a standard basis vector on $n$ qudits. Then, for $j \neq j'$ and $p \in [\lambda_j+1], p' \in [\lambda_{j'} + 1]$,
\begin{align*}
&
\tr_{-(1,2)} 
\lp( 
\swapt((j, p), (j', p')) \cdot \bigg( I \otimes I \otimes (\Un{n}\ket{a}\bra{a}\Und{n})\bigg)
\rp)  
\\
& = 
\begin{cases}
    \bigg( 
U \ket{a_{j,p}} \bra{a_{j,p}} U^\dagger\bigg)
\otimes 
\bigg( 
U \ket{a_{j',p'}} \bra{a_{j',p'}} U^\dagger \bigg) & p\neq \lambda_j + 1,~p' \neq \lambda_{j'}+1 \\
I 
\otimes 
\bigg( 
U \ket{a_{j',p'}} \bra{a_{j',p'}} U^\dagger \bigg) & p = \lambda_j+1,~p' \neq \lambda_{j'}+1\\
\bigg( 
U \ket{a_{j,p}} \bra{a_{j,p}} U^\dagger\bigg)
\otimes 
I & p \neq \lambda_j+1,~p' = \lambda_{j'} + 1\\
I \otimes I & p = \lambda_j + 1,~p' = \lambda_{j'} + 1
\end{cases}
\, .
\end{align*}
\end{claim}
\begin{proof}
Decompose
$I = \sum_{k=1}^d U \ket{k}\bra{k} U^{\dagger}$.
Then the left hand side is equal to
\begin{align}
&
\sum_{ k, \ell \in [d] }
\tr_{-(1,2)} 
\lp(  
\Un{(n+2)} \;
\text{Swap}_o( (j, p), (j', p') )
 \lp( \ket{k \; \ell \; a} \bra{k \; \ell \; a} \rp)
\Und{(n+2)}
\rp) \nonumber \\
&= 
U^{\otimes 2}
\sum_{ k, \ell \in [d] }
\tr_{-(1,2)} 
\bigg(  
\text{Swap}_o( (j, p), (j', p') )
 \lp( \ket{k \; \ell \; a} \bra{k \; \ell \; a}  \rp)
\bigg)
\lp(U^{\otimes 2}\rp)^{\dagger}.
\label{eq:extract-U}
\end{align}
When $p \neq \lambda_j+1$, it is not hard to see that the partial-trace is non-vanishing only if $k = a_{j, p}$.
Otherwise, if $p = \lambda_j+1$, the first qudit is independent from all  other qudits, and hence it will be simplified into $I$.
Similarly, when $p' \neq \lambda_{j'}+1$, the partial-trace is non-vanishing only if $\ell = a_{j', p'}$. If $p' = \lambda_{j'}+1$, the second qudit will be simplified into $I$. 
Combining these observations with a careful case analysis then gives
\begin{align*}
&
\sum_{k, \ell \in [d]}
\tr_{-(1,2)} 
\bigg(  
\text{Swap}_o( (j, p), (j', p') )
 \lp( \ket{k \; \ell \; a} \bra{k \; \ell \; a}  \rp)
\bigg)
\\
& = 
\begin{cases}
\ket{a_{j,p}} \bra{a_{j,p}} 
\otimes 
 \ket{a_{j',p'}} \bra{a_{j',p'}} & p\neq \lambda_j + 1,~p' \neq \lambda_{j'}+1 \\
I 
\otimes 
\ket{a_{j',p'}} \bra{a_{j',p'}}  & p = \lambda_j+1,~p' \neq \lambda_{j'}+1\\ 
 \ket{a_{j,p}} \bra{a_{j,p}} 
\otimes 
I & p \neq \lambda_j+1,~p' = \lambda_{j'} + 1\\
I \otimes I & p = \lambda_j + 1,~p' = \lambda_{j'} + 1
\end{cases}
\, .
\end{align*}
Substituting the above back into \Cref{eq:extract-U} concludes the proof of \Cref{prop:tensor2}.
\end{proof}

\begin{claim}
\label{prop:tensor3}
Let $j \neq j'$, $p \in [\lambda_j], p' \in [\lambda_{j'}]$, and $\ket{a}, \ket{a'}$ be two standard basis vector on $n$ qudits such that
$\ket{a'} = \swapc( (j, p), (j', p') ) \ket{a}$
and $\ket{a'} \neq \ket{a}$.
Then we have
\begin{align*}
&
\tr_{-(1,2)}  
\lp( \swapt((j, p), (j', p')) \cdot \bigg(
I \otimes I \otimes (\Un{n}\ket{a}\bra{a'} \Und{n}) \bigg) \rp)
\\
& = 
\lp( U \ket{ a_{j, p} }
\bra{ a_{j,p} } U^{\dagger}
\otimes
U \ket{ a_{j', p'} }
\bra{ a_{j',p'} } U^{\dagger} \rp)
\text{Swap}^{(2)} \, ,
\end{align*}
where $\textrm{Swap}^{(2)}$ is a permutation over $2$-qudits which swaps the first qudit with the second qudit.
\end{claim}
\begin{proof}
Decompose $I = \sum_{k=1}^d U \ket{k} \bra{k} U^{\dagger}$. Then the left hand side is equal to   \begin{align*}
&
\sum_{ k, \ell \in [d] }
\tr_{-(1,2)} 
\lp(  
\Un{(n+2)} \;
\text{Swap}_o( (j, p), (j', p') )
 \lp( \ket{k \; \ell \; a} 
 \bra{k \; \ell \; a'} \rp)
\Und{(n+2)}
\rp) \nonumber \\
&= 
U^{\otimes 2}
\sum_{ k, \ell \in [d] }
\tr_{-(1,2)} 
\bigg(  
\text{Swap}_o( (j, p), (j', p') )
 \lp( \ket{k \; \ell \; a} 
 \bra{k \; \ell \; a'}  \rp)
\bigg)
\lp( U^{\otimes 2} \rp)^{\dagger}.
\end{align*}
It is not hard to verify that the partial-trace is non-vanishing only if $k = a'_{j, p}$
and $\ell = a'_{j', p'}$.
Thus, the expression simplifies into
$$
U \ket{ a_{j, p} }
\bra{ a'_{j,p} } U^{\dagger}
\otimes
U \ket{ a_{j', p'} }
\bra{ a'_{j',p'} } U^{\dagger}.
$$
We note that this is exactly the right hand side from the claim since $\ket{a'} = \swapc( (j, p), (j', p') ) \ket{a}$.
\end{proof}


We next define some basic operators over $2$-qudit system. 
\begin{definition}\label{def:pisymt}
    Define $\textrm{Swap}^{(2)}$ a permutation over $2$-qudits which swaps the first qudit with second qudit. Define $\pisymt$, a projector on the symmetric subspace over $2$-qudit system, as the sum $(\textrm{Swap}^{(2)} + I)/2$.
\end{definition}

\begin{claim}
\label{clm:tensor-caseII}
Let $\ket{a} \in E_n(w)$.
Let $j \in [k]$, $p \in [\lambda_j]$, $p' = \lambda_j+2$.
Then we have
\begin{align*}
&\tr_{-(1,2)} 
\lp(  
\swapts(j, p, p') \cdot
\lp( 
\pisymt
\otimes I^{\otimes n} \rp) \cdot \bigg(I \otimes I \otimes (\Un{n}\ket{a} \bra{a}\Und{n})\bigg)\rp) \\
&= 
\lp( (U \ket{a_{j,p}}\bra{a_{j,p}} U^\dagger)  \otimes I\rp)
\pisymt.
\end{align*}
\end{claim}
\begin{proof}
Decompose $I = U \sum_{k \in [d]} \ket{k}\bra{k} U^{\dagger}$. 
The left hand side is equal to
\begin{align*}
&
\sum_{k \in [d]}
\tr_{-(1,2)} 
\lp(  
\Un{(n+2)}
\swapts(j, p, p') \cdot
\lp( 
\pisymt
\otimes I^{\otimes n} \rp) \cdot \bigg( 
\ket{k \; \ell \; a} \bra{k \; \ell \; a}
\bigg)
\Und{(n+2)}
\rp) \\
&= 
\frac{1}{2}
\sum_{k, \ell \in [d]}
U^{\otimes 2}
\tr_{-(1,2)} 
\lp(  
\swapts(j, p, p') \cdot \bigg( 
\ket{k \; \ell \; a} \bra{k \; \ell \; a}
\bigg)
\rp) 
\Und{2}  \\
&+
\frac{1}{2}
\sum_{k, \ell \in [d]}
U^{\otimes 2}
\tr_{-(1,2)} 
\lp(  
\swapts(j, p, p') \cdot 
\bigg( 
\ket{\ell \; k \; a} \bra{k \; \ell \; a}
\bigg)
\rp) 
\Und{2}.
\end{align*}
For the first term, we note that it is non-vanishing only if $k = a_{j,p}$.
Then the first term can be simplified into 
$$
\frac{1}{2}
U^{\otimes 2}
\sum_{\ell \in [d]}
\ket{ a_{j, p} \ell }
\bra{ a_{j, p} \ell }
\Und{2}
= 
\frac{1}{2}
U \ket{a_{j, p}}
\bra{a_{j, p}} U^{\dagger} \otimes I.
$$
For the first term, we note that it is non-vanishing only if $\ell = a_{j,p}$.
Then the second term can be simplified into 
$$
\frac{1}{2}
U^{\otimes 2}
\sum_{k \in [d]}
\ket{ a_{j, p} k }
\bra{ k a_{j, p} }
\Und{2}
= 
\frac{1}{2}
 U \ket{a_{j', p'}}
\bra{a_{j', p'}} U^{\dagger} 
\otimes I
\text{Swap}^{(2)}
\, ,
$$
where $\text{Swap}^{(2)}$ is the permutation that swaps the two qudits.
Combining the two terms then concludes the proof of the claim.
\end{proof}

\begin{claim}
\label{clm:tensor-caseIV}
Let $j \in [k]$, $p, p' \in [\lambda_j]$.
Let $\ket{a} \in E_n(w)$, 
$\ket{a'} = \ket{a}$ or $
\ket{a'} = \swapc((j,p), (j', p')) \ket{a}$.
Then we have
\begin{align*}
&
\tr_{-(1,2)} 
\lp(  
\swapts(j, p, p') \cdot
\lp( 
\pisymt
\otimes I^{\otimes n} \rp) \cdot \bigg(I \otimes I \otimes (\Un{n}\ket{a} \bra{a'}\Und{n})\bigg)\rp)
\\
&= 
U^{\otimes 2} \ket{a_{j,p} \; a_{j', p'}
}\bra{a'_{j,p} \; a'_{j', p'}}
\Und{2}
\pisymt.
\end{align*}
\end{claim}
\begin{proof}
Decompose $I = U \sum_{k \in [d]} \ket{k}\bra{k} U^{\dagger}$. 
The left hand side is equal to
\begin{align*}
&
\sum_{k, \ell \in [d]}
\tr_{-(1,2)} 
\lp(
\Un{(n+2)}
\swapts(j, p, p') \cdot
\lp( 
\pisymt
\otimes I^{\otimes n} \rp) \cdot \bigg(
\ket{k \; \ell \; a} \bra{k \; \ell \; a'}
\bigg)
\Und{(n+2)}
\rp)\\
&= 
\frac{1}{2}
\sum_{k, \ell \in [d]}
U^{\otimes 2}
\tr_{-(1,2)} 
\lp(  
\swapts(j, p, p') \cdot \bigg( 
\ket{k \; \ell \; a} \bra{k \; \ell \; a'}
\bigg)
\rp) 
\Und{2}  \\
&+
\frac{1}{2}
\sum_{k, \ell \in [d]}
U^{\otimes 2}
\tr_{-(1,2)} 
\lp(  
\swapts(j, p, p') \cdot 
\bigg( 
\ket{\ell \; k \; a} \bra{k \; \ell \; a'}
\bigg)
\rp) 
\Und{2}.
\end{align*}
For the first term, we note that the partial-trace is non-vanishing only if 
$\ket{k} = \ket{ a'_{j,p} }$
and $\ket{\ell} = \ket{a'_{j', p'}}$. 
Thus, it can be simplified into 
$$
\frac{1}{2}
U^{\otimes 2}
\ket{ a_{j, p} \; a_{j' p'} }
\bra{ a'_{j, p} \; a'_{j' p'} }
\Und{2}.
$$
For the second term, we note that the partial trace is non-vanishing only if $\ket{\ell} = \ket{ a'_{j,p} }$
and $\ket{k} = \ket{a'_{j', p'}}$.
Thus, it can be simplified into
$$
\frac{1}{2}
U^{\otimes 2}
\ket{ a_{j, p} \; a_{j', p'} }
\bra{ a'_{j', p'} \; a'_{j, p} }
\Und{2}
=
\frac{1}{2}
U^{\otimes 2}
\ket{ a_{j, p} \; a_{j', p'} }
\bra{ a'_{j, p} \; a'_{j', p'} }
\Und{2}
\text{Swap}^{(2)}
$$
where $\text{Swap}^{(2)}$ is the permutation that swaps the two qudits.
Combining the two terms then concludes the proof of the claim.
\end{proof}
\begin{fact}
\label{fact:eigen-bound}
Let $\ket{x}, \ket{y}$ be two unit vectors. 
Furthermore, assume $\tr(O) = 0$.
Then we have
\begin{align*}
&\tr \lp( \lp(O \otimes O\rp)
\;
\lp( \ket{x}\bra{x}   \otimes I\rp)
\pisymt
\rp)
\leq 
\frac{1}{2} \snorm{\infty}{O^2}  \\
&\tr \lp( \lp(O \otimes O\rp)
\;
\lp( \ket{x}\bra{x}   \otimes 
\ket{y}\bra{y} \rp)
\pisymt
\rp)
\leq 
\snorm{\infty}{O}^2.
\end{align*}
\end{fact}
\begin{proof}
To show the first inequality, we note that
$$
\tr \lp( \lp(O \otimes O\rp)
\;
\lp( \ket{x}\bra{x}   \otimes I\rp)
\pisymt
\rp)
= 
\frac{1}{2}
\tr(O) \; \bra{x} O \ket{x}
+ \frac{1}{2}
\bra{x} O^2 \ket{x}
\leq \frac{1}{2} \snorm{\infty}{O^2}.
$$
To show the second inequality, we note that
$$
\tr \lp( \lp(O \otimes O\rp)
\;
\lp( \ket{x}\bra{x}   \otimes 
\ket{y}\bra{y} \rp)
\pisymt
\rp)
= 
\frac{1}{2}
\bra{x} O \ket{x}
\bra{y} O \ket{y}
+
\frac{1}{2}
\bra{y} O \ket{x}
 \bra{y} O \ket{x}
\leq  \snorm{\infty}{O}^2.
$$

\end{proof}

\section{Variance Lower Bound}
\label{app:var-lb}
In this section, we present a concrete setup of 
the post-processing state $U^{\otimes n} \ket{\tau}$, partition outcome $\lambda \partition n$, and observable $O$ such that the bound from \Cref{lem:second-moment} on $\Var[O \Psi]$ is tight up to a constant factor.

First, assume the generic pre-processing $\Gamma$ returns the partition $\lambda = (n)$ and a state of the form 
$U^{\otimes n} \ket{(\lambda, i, 0)}$, where $\ket{(\lambda, i, 0)}$ lies in the same weight subspace as $\ket{0}^{\otimes p} \otimes \ket{1}^{\otimes q}$ for some $p, q \in \N$ with $p + q = n$. 
We note that this is a valid output for our generic pre-processing $\Gamma$, because in this case each subspace $Y_{\lambda} V(w)$ is at most $1$-dimensional. Hence, without loss of generality, we can assume that the output state is of the form mentioned above (with a change in unitary).

We proceed to compute the variance of the shadow matrix $\Psi$ after running the POVM $\mathcal M^{\lambda}$ on state $U^{\otimes n} \ket{(\lambda, i, 0)}$.
To simplify the analysis even further, we assume that $U = I$ and $O$ is a traceless observable.
\begin{assumption}
\label{ass:setup}
The generic pre-processing step returns 
$\lambda = (n)$, and the state $
\ket{(\lambda, i, 0)}$, where $\ket{(\lambda, i, 0)} \in \Span \{ P_{\pi} \ket{0}^{\otimes p} \ket{1}^{\otimes q} \mid \pi \in S_n \}$ for some natural numbers $p + q = n$.
Moreover, the observable $O$ is traceless, i.e., $\tr(O) = 0$.
\end{assumption}
\begin{lemma}
Let $O$ be the observable received and
$\Psi$ be the shadow matrix defined as in \Cref{lem:expectation}.
Under \Cref{ass:setup}, it holds that
\begin{align*}
\Var[ O \Psi ]    
= 
&\frac{n+d}{n+d+1} 
\bigg( 
\tr (O^2) + 2p \bra{0} O^2 \ket{0}
+ 2q \bra{1} O^2 \ket{1}
+ 2 p q  \; \lp( \bra{0} O \ket{1} \rp)^2
\bigg) \\
- &\frac{1}{n+d+1}
\lp( 
+ p^2  \lp( \bra{0} O \ket{0} \rp)^2
+ q^2 \lp( \bra{1} O \ket{1} \rp)^2 
+ 2 p q \;  \bra{0} O \ket{0} \;  \bra{1} O \ket{1}
\rp).
\end{align*}
Moreover, whenever
$O = \begin{bmatrix}
1 &1 & \cdots \\
1 &-1 & \cdots \\
\vdots &\vdots &\ddots
\end{bmatrix}$, $p, q = \Theta(n)$, we have that
$$
\frac{1}{n^2} \Var[ O \Psi ]   \geq 
\Omega \lp( \frac{pq}{n^2}  \; \lp( \bra{0} O \ket{1} \rp)^2
\rp) \geq \Omega(1).
$$
for all sufficiently large $n$.
\end{lemma}
\begin{proof}
We begin by computing the second moment $\E[ \Psi \otimes \Psi  ]$ in a fashion similar to the proof of \Cref{lem:second-moment}.
The computation is easier as we now assume $\lambda$ only has one part.
In particular, following identical steps, we reach the expression
\begin{align*}
&\E [\Psi \otimes \Psi]     \\
&= \frac{2 (n+d)^2 }{n(n+1)}
\frac{ \kappa_{n} }{ \kappa_{n+2} }
\sum_{p'=1}^{n+2}
\sum_{p = 1}^{p'-1}
\tr_{-1, -2}
\lp( 
\swapts(1, p, p')
\lp( 
\pisymt
\otimes I^{\otimes 2}\rp)
\lp(
I^{\otimes 2} \otimes
\ket{(\lambda, i, 0)}\bra{(\lambda, i, 0)}\rp) \rp) \\
&= 2  \frac{n+d}{n+d+1}
\sum_{p'=1}^{n+2}
\sum_{p = 1}^{p'-1}
\tr_{-1, -2}
\lp( 
\swapts(1, p, p')
\lp( 
\pisymt
\otimes I^{\otimes 2}\rp)
\lp(
I^{\otimes 2} \otimes
\ket{(\lambda, i, 0)}\bra{(\lambda, i, 0)}\rp) \rp) \, ,
\end{align*}
where in the second line we use
$ \frac{1}{n(n+1)} \frac{\kappa_n}{\kappa_{n+2}} 
= \frac{  n+d }{ n+d+1 }.
$

Let $\ket{a}, \ket{a'}$ be two arbitrary standard basis vector.
Define
\begin{align*}
&T_{p, p'}(\ket{a}, \ket{a'}) := \tr_{-1, -2} \lp(\lp( \pisymt
\otimes I^{\otimes n} \rp) \cdot \bigg(I \otimes I \otimes (\Un{n}\ket{a} \bra{a'}\Und{n})\bigg)\rp) \, ,\\
&T_{p, p'}:= \tr_{-1, -2} \lp(\lp( \pisymt
\otimes I^{\otimes n} \rp) \cdot \bigg(I \otimes I \otimes  \ket{(\lambda, i, 0)}\bra{(\lambda, i, 0)} \bigg)\rp).    
\end{align*}
To simplify the partial trace, we again break into four cases depending on the values of $p, p'$.

\textbf{Case I}: $p = n+1, p' = n+2$. Recall that in this case we have the following more general expression:
\begin{align*}
T_{p, p'}( \ket{a}, \ket{a'} )
=     
\begin{cases}
\pisymt \lp( I \otimes I \rp)  
&\text{ if } a = a' \, , \\
0 &\text{ otherwise}
\end{cases}
\end{align*}
where $\ket{a}, \ket{a'}$ are two standard basis vectors that lie in the same weight subspace.
In the current setting, this implies that
\begin{align*}
T_{n+1, n+2}
= \tr_{-1, -2} \lp(\lp( \pisymt
\otimes I^{\otimes n} \rp) \cdot \bigg(I \otimes I \otimes \ket{(\lambda, i, 0)}\bra{(\lambda, i, 0)} \bigg)\rp)
=   \pisymt \lp( I \otimes I \rp).
\end{align*}
Thus, its contribution to variance is
$
\E[ O \otimes O T_{n+1, n+2} ]
= \frac{1}{2} \tr(O^2).
$

\textbf{Case II}: $p \in [n], p' = n+2$. Recall that in this case we have the following more general expression:
\begin{align*}
T_{p, p'}(\ket{a}, \ket{a'})
=     
\begin{cases}
\lp( (U \ket{a_{j,p}}\bra{a_{j,p}} U^\dagger)  \otimes I\rp)
\pisymt. 
&\text{ if } a = a' \, , \\
0 &\text{ otherwise}
\end{cases}
\end{align*}
In the current setting, this implies that
\begin{align*}
\sum_{p=1}^n T_{p, n+2}
=   
\lp(\lp(  p \ket{0}\bra{0} + q \ket{1}\bra{1}  \rp)  \otimes I\rp) \pisymt. 
\end{align*}
Thus, its contribution to variance is
$
\E\lp[ O \otimes O \sum_{p=1}^n T_{p, n+2} \rp]
= \frac{1}{2} \lp( p \bra{0} O^2 \ket{0} + q \bra{1} O^2 \ket{1} \rp).
$

\textbf{Case III}: $p \in [n], p' = n+1$. Recall that the general expression in this case is identical to Case II. So its contribution to the variance is also given by
$$
\E[ O \otimes O \sum_{p=1}^n T_{p, n+1} ]
= \frac{1}{2} \lp( p \bra{0} O^2 \ket{0} + q \bra{0} O^2 \ket{0} \rp).
$$

\textbf{Case IV}: 
$p \in [n], p' = [n]$.
Recall that in this case we have the following more general expression:
\begin{align*}
T_{p, p'}(\ket{a}, \ket{a'})
=     
\begin{cases}
U^{\otimes 2} 
\ket{a_{j,p} a_{j', p'} }
\bra{a'_{j,p} a'_{j', p'} }
\Und{2}
\pisymt
&\text{ if } a = a' \text{ or } \ket{a'} = \swapc( (1, p), (1, p') ) \ket{a}
\, , \\
0 &\text{ otherwise }.
\end{cases}
\end{align*}
Note that when $\lambda = (n)$ and $\ket{(\lambda, i, 0)}$ lies in the the same weight subspace as $\ket{0}^{\otimes p} \otimes \ket{1}^{\otimes q}$, $\ket{(\lambda, i, 0)}$, $\ket{(\lambda, i, 0)}$ has a particularly nice structure --- it is the uniform super position of all permutations of
$\ket{0}^{\otimes p} \otimes \ket{1}^{\otimes q}$.
This then implies that
$$
\sum_{p'=2}^n \sum_{p=1}^{p'-1} T_{p, p'}
=   
\frac{p(p-1)}{2} \lp( \ket{0} \bra{0} \rp)^{\otimes 2}
+ 
\frac{q(q-1)}{2} \lp( \ket{1} \bra{1} \rp)^{\otimes 2}
+ 2 \; p q \ket{0} \bra{0} \otimes \ket{1} \bra{1} \pisymt.
$$
Hence, the contribution to the variance is given by
\begin{align*}
&\E\lp[ O \otimes O \sum_{p'=2}^n \sum_{p=1}^{p'-1} T_{p, p'} \rp] \\
&= 
\frac{p(p-1)}{2}  \lp( \bra{0} O \ket{0} \rp)^2
+
\frac{q(q-1)}{2}  \lp( \bra{1} O \ket{1} \rp)^2
+ 
p q  \; \lp( \bra{0} O \ket{1} \rp)^2
+ p q \;  \bra{0} O \ket{0} \;  \bra{1} O \ket{1}.    
\end{align*}
Overall, we thus have
\begin{align*}
&\E[ \tr \lp(  \lp( O \otimes O \rp) \; \lp( \Psi \otimes \Psi \rp) \rp)] \\
&= \frac{n+d}{n+d+1} 
\bigg( 
\tr (O^2) + 2p \bra{0} O^2 \ket{0}
+ 2q \bra{1} O^2 \ket{1}
+ p^2  \lp( \bra{0} O \ket{0} \rp)^2
+ q^2 \lp( \bra{1} O \ket{1} \rp)^2 \\
&+ 2 p q  \; \lp( \bra{0} O \ket{1} \rp)^2
+ 2 p q \;  \bra{0} O \ket{0} \;  \bra{1} O \ket{1}
\bigg).
\end{align*}
By \Cref{lem:expectation}, we have
$$
\E[\Psi] = (I + p \ket{0}\bra{0} + q \ket{0} \bra{0}).
$$

Thus, we have
\begin{align*}
\Var[ O \Psi ]    
= 
&\frac{n+d}{n+d+1} 
\bigg( 
\tr (O^2) + 2p \bra{0} O^2 \ket{0}
+ 2q \bra{1} O^2 \ket{1}
+ \\
&2 p q  \; \lp( \bra{0} O \ket{1} \rp)^2
- p \lp( \bra{0} O \ket{0} \rp)^2
-  q \lp( \bra{1} O \ket{1} \rp)^2
\bigg) \\
- &\frac{1}{n+d+1}
\lp( 
p^2  \lp( \bra{0} O \ket{0} \rp)^2
+ q^2 \lp( \bra{1} O \ket{1} \rp)^2 
+ 2 p q \;  \bra{0} O \ket{0} \;  \bra{1} O \ket{1}
\rp).
\end{align*}
It is not hard to see that all negative terms are dominated by the positive term 
$\frac{n+d}{n+d+1} \; 2 p q  \; \lp( \bra{0} O \ket{1} \rp)^2 $ when $p, q = \Theta(n)$.  
Thus, if we set $O = \begin{bmatrix}
1 &1 & \cdots \\
1 &-1 & \cdots \\
\vdots &\vdots &\ddots
\end{bmatrix}$, the variance will be 
bounded from below by $ \Theta(pq / n^2)  \geq \Omega(1)$.
\end{proof}

\end{document}